\pdfoutput=1
\documentclass[aps,pra,10pt,superscriptaddress,floatfix,notitlepage]{revtex4-2}
\usepackage[T1]{fontenc} %256 bit font encoding
\usepackage[english]{babel} %english language
\usepackage[utf8]{inputenc}
\usepackage[margin=0.9in]{geometry}
\usepackage{times}
\usepackage[ruled,vlined,linesnumbered]{algorithm2e}

\usepackage[table,usenames,dvipsnames]{xcolor}
\usepackage[colorlinks=true,hypertexnames=false]{hyperref}
\usepackage{amsmath}
\usepackage{amsthm} % thm environmen00000
\usepackage{amsmath, amsfonts, amssymb, mathtools}
\usepackage{dsfont}
\usepackage{varwidth}
\usepackage{booktabs}
 \usepackage[]{natbib}
\usepackage[normalem]{ulem} % to strikeout comments, normalem needed to bring emph
\usepackage{graphicx}
\usepackage{tikz}
\usepackage[utf8]{inputenc}
\usepackage[nolist]{acronym}

\graphicspath{{./pics/}{./}}

\newtheorem{theorem}{Theorem}
\newtheorem{corollary}[theorem]{Corollary}
\newtheorem{lemma}[theorem]{Lemma}
\newtheorem{proposition}[theorem]{Proposition}
\newtheorem{definition}[theorem]{Definition}
\newcommand{\myleft}{\mathopen{}\mathclose\bgroup\left}
\newcommand{\myright}{\aftergroup\egroup\right}
\DeclareMathOperator{\Tr}{Tr}

\DeclareMathOperator{\diag}{diag}

\newcommand{\fro}{\mathrm{F}}

\renewcommand{\fro}{F}

\DeclareMathOperator\tr{Tr}

\DeclareMathOperator{\gend}{g_{\mathrm{end}}}

\newcommand{\CC}{\mathbb{C}}
\newcommand{\RR}{\mathbb{R}}
\newcommand{\1}{\mathds{1}}
\newcommand{\EE}{\mathbb{E}}
\newcommand{\PP}{\mathbb{P}}
\newcommand{\md}[1]{\mathbb{#1}}

\newcommand{\mc}[1]{\mathcal{#1}}

\newcommand{\ct}{{}^\dagger}

\newcommand{\ad}{^\dagger}%symbol for conjugate transpose
\newcommand{\norm}[1]{\left\Vert #1 \right\Vert} %norm with variable height
\newcommand{\normn}[1]{\lVert #1 \rVert} %norm with normalsize height
\newcommand{\snorm}[1]{\norm{#1}_\infty} %spectral norm  =  (2->2)-norm
\newcommand{\snormn}[1]{\normn{#1}_\infty}

\newcommand{\fnorm}[1]{\norm{#1}_\fro} %Frobenius norm
\newcommand{\fnormn}[1]{\normn{#1}_\fro}

\newcommand{\dnorm}[1]{\norm{#1}_\diamond} %diamond norm

\newcommand{\lTwoNorm}[1]{\norm{#1}_{\ell_2}} %\ell_2 norm

\newcommand{\ket}[1]{\left.\left|{#1}\right.\right\rangle}

\newcommand{\bra}[1]{\left.\left\langle{#1}\right.\right|}

\newcommand{\braket}[2]{\left\langle #1 \middle| #2 \right\rangle}

\newcommand{\ketbra}[2]{\ket{#1} \!\! \bra{#2}}

\newcommand{\kett}[1]{|{#1}{\rangle\!\rangle}}
\newcommand{\braa}[1]{{\langle\!\langle}{#1}|}

\newcommand{\Var}{\operatorname{Var}}

\newcommand{\Ev}{\operatorname{\EE}} %expectation value
\newcommand{\tn}[1]{^{\otimes#1}} % Command for tensor power notation
\newcommand{\dens}[1]{\ket{#1}\!\!\bra{#1}}

\newcommand{\gr}{\ensuremath{\md{G}}}
\newcommand{\gsum}[1]{\ensuremath{\frac{1}{|\gr |}\!\sum_{#1\in \gr}}}
\newcommand{\ggsum}[1]{\ensuremath{\frac{1}{|\gr |^2}\!\sum_{#1\in \gr}}}
\newcommand{\Irr}{\ensuremath{\mathrm{Irr}(\gr)}}

\newcounter{example}[section]
\newenvironment{example}[1][]{\refstepcounter{example}\par\medskip
   \noindent \textbf{Example~\theexample. #1} \rmfamily}{\medskip}

\newcommand{\wh}[1]{\widehat{#1}}

\newcommand{\dsum}{\frac {1}{d^2} \sum_{\lambda \in \Lambda} d_{\sigma_\lambda} \, }
\newcommand{\set}[1]{\left\{ #1 \right\}}
\newcommand{\Fnorm}[1]{\fnormn{#1}}

\newcommand{\infnorm}[1]{\snormn{#1}}
\newcommand{\abs}[1]{\lvert #1 \rvert}
\newcommand{\vecOL}{\lvert \braket{z (\sigma_\lambda) }{ r_{\max} (\sigma_\lambda) } \braket{ \ell_{\max} (\sigma_\lambda)}{z (\sigma_\lambda)} \rvert}
\newcommand{\onevecOL}{\big\lvert 1-\braket{z (\sigma_\lambda) }{ r_{\max} (\sigma_\lambda) } \braket{ \ell_{\max} (\sigma_\lambda)}{z (\sigma_\lambda)} \rvert}
\newcommand{\onevecOLcases}{\big\lvert 1-\braket{z (\sigma_\lambda) }{ r_{\max} (\sigma_\lambda) } \braket{ \ell_{\max} (\sigma_\lambda)}{z (\sigma_\lambda)} \rvert}
\makeatletter
\def\l@subsubsection#1#2{}
\makeatother

\definecolor{ingo}{rgb}{.6,.05,.05}

\definecolor{jonas}{rgb}{.06,.5,.05}

\definecolor{eo}{RGB}{246,76,25}

\definecolor{cy}{RGB}{70,200,150}

\newcommand{\newtext}[1]{{\color{black}#1}}

\makeatletter % metadata for hyperref
 \hypersetup{pdftitle = {A general framework for randomized benchmarking},
	     pdfauthor = {J. Helsen, I. Roth, E. Onorati, A. H. Werner, J. Eisert},
	     pdfsubject = {quantum physics},
 	     pdfkeywords = {randomized benchmarking, Fourier analysis, perturbation bounds, representation theory
			    }
	    }
\makeatother

\begin{document}%%% =========================================================
%%% =========================================================================
	
\title{A general framework for randomized benchmarking}

\author{J. Helsen}
\affiliation{QuSoft \& Korteweg-de Vries Institute, University of Amsterdam, Science Park 123 1098 XG Amsterdam, The Netherlands}
\author{I. Roth}
\affiliation{Quantum Research Centre, Technology Innovation Institute, Abu Dhabi, UAE}
\affiliation{Dahlem Center for Complex Quantum Systems, Arnimallee 14, Freie Universit\"{a}t Berlin, 14195 Germany}
\author{E. Onorati}
\affiliation{University College London,
Dept. of Computer Science
66-72 Gower Street,
London WC1E 6EA,
United Kingdom
}
\author{A. H. Werner}
\affiliation{{Department of Mathematical Sciences, University of Copenhagen, 2100 K{\o}benhavn, Denmark}}
\affiliation{{NBIA, Niels Bohr Institute, University of Copenhagen, Blegdamsvej 17, 2100 K{\o}benhavn, Denmark}}
\author{J. Eisert}
\affiliation{Dahlem Center for Complex Quantum Systems, Arnimallee 14, Freie Universit\"{a}t Berlin, 14195 Germany}

\affiliation{Mathematics and Computer Science,
Takustra{\ss}e 9, Freie Universit\"{a}t Berlin, 14195 Berlin,
Germany}
\affiliation{Helmholtz-Zentrum Berlin f{\"u}r Materialien und Energie, Hahn-Meitner-Platz 1, 14109 Berlin, Germany}

\begin{abstract}
Randomized benchmarking refers to a collection of protocols that in the past decade have become central methods for characterizing quantum gates. These protocols aim at efficiently estimating the quality of a set of quantum gates in a way that is resistant to state preparation and measurement errors.
Over the years many versions have been developed, however, a comprehensive theoretical treatment of randomized benchmarking has been missing. In this work, we develop a rigorous framework of randomized benchmarking general enough to encompass virtually all known protocols as well as novel, more flexible extensions. Overcoming previous limitations on error models and gate sets, this framework allows us, for the first time, to formulate realistic conditions under which we can rigorously guarantee that the output of any randomized benchmarking experiment is well-described by a linear combination of matrix exponential decays. We complement this with a detailed analysis of the fitting problem associated with randomized benchmarking data. We introduce modern signal processing techniques to randomized benchmarking, prove analytical sample complexity bounds, and numerically evaluate performance and limitations. In order to reduce the resource demands of this fitting problem, we introduce novel, scalable post-processing techniques to isolate exponential decays, significantly improving the practical feasibility of a large set of randomized benchmarking protocols. These post-processing techniques overcome shortcomings in efficiency of several previously proposed methods such as character benchmarking and linear-cross entropy benchmarking. Finally, we discuss, in full generality, how and when randomized benchmarking decay rates can be used to infer quality measures like the average fidelity. On the technical side, our work substantially extends the recently developed Fourier-theoretic perspective on randomized benchmarking by making use of the perturbation theory of invariant subspaces, as well as ideas from signal processing.
\end{abstract}

\maketitle

% {

% \setlength{\parskip}{0em}
%   \hypersetup{linkcolor=black}
%   \tableofcontents
% }

\section{Introduction}

In the last few years significant steps have been taken towards the development of large-scale quantum computers.
A key part of the development of these quantum computers are tools that provide diagnostics, certification, and benchmarking. Particularly for quantum operations, stringent conditions have to be met to achieve fault tolerance. Motivated by this observation, in recent years a significant body of work has been dedicated to the development of tools for the
certification and benchmarking of quantum gates.
A prominent role in this collection of tools is taken by methods that can be collectively referred to as \emph{\acf{RB}}. These methods have risen to prominence because they conform well to the demands of realistic experimental settings. They estimate the magnitude of an average error of a set of quantum gates in a fashion that is robust to errors in \emph{state preparation and measurement} (SPAM) and moreover is, in many settings, efficient, in the sense that the resources required scale polynomially with the number of qubits in the device. The various versions of \ac{RB} apply sequences of randomly chosen quantum gates of varying length. Small errors are thus amplified with the sequence length, and gate quality measures can be extracted from the dependence \newtext{of the output data on sequence length.}

In \ac{RB} protocols, group structures feature strongly, in that the gate set considered is in almost all cases a subset of a finite group. Such group structures not only make it possible to efficiently make predictions for error-free sequences and compute inverses, but they also provide the means to analyze the error contribution after averaging. Originally proposed for random unitary gates \citep{FirstRB,dankert_exact_2009,PhysRevA.75.022314},
\ac{RB} is now most prominently executed with gates from the
so-called \emph{Clifford group}
\citep{magesan2012efficient,knill2008randomized,emerson2007symmetrized}, a set of
efficiently classically simulatable quantum gates that take a key role
specifically in fault tolerant quantum computing \citep{Roads}.
It has also been considered for other (subsets of) finite groups
\citep{PhysRevA.90.030303,PhysRevLett.123.060501,PhysRevA.92.060302,BeyondCliffordRB,PhysRevLett.123.060501, HelsenEtAl:2019:character, CycleBenchmarking,francca2018approximate,proctor2019direct}.  Moreover \ac{RB} has been generalized to capture
other figures of merit of gate sets, such as relative average gate fidelities  to
specific anticipated target gates \citep{magesan2012efficient}, fidelities within a symmetry sector \citep{PhysRevA.92.060302,PhysRevLett.123.060501}, or
the unitarity \citep{wallman2015estimating}. Specifically recently, with challenges of realizing fault-tolerant quantum
computers in mind, emphasis has been put on capturing
 losses, leakage, and cross-talk in a scheme
 \citep{GambettaEtAl:2012:simultaneousRB, WallmanEtAl:2015:LossRates,WallmanEtAl:2016:Leakage}. Also, data from \ac{RB} -- or rather suitably combining data from multiple such experiments --
can be sufficient to acquire full tomographic information about a quantum gate
\citep{KimmOhki,AverageGateFidelities, 2019arXiv190712976F}.
This adds up to a wealth of \ac{RB} protocols \cite{Review}
proposed over the previous years. Fig.\ \ref{fig:results overview} summarizes
a (to our knowledge) up to date list of theoretical proposals for \ac{RB} procedures presently known.

A significant body of work moreover deals with the limitations and precise
preconditions of \ac{RB}. The originally
rather stringent assumptions on noise being necessarily identical across different quantum gates have over time been relaxed for particular protocols in later work \citep{wallman2018randomized,Merkel18,Proctor17}, and the connection between the output of \ac{RB} and operationally relevant quantities (such as average fidelity) has been studied in some detail~\cite{Proctor17,carignan2018randomized}.

And yet, it seems fair to say that a comprehensive picture of \ac{RB}
schemes for the quantum technologies \cite{Roadmap} has been lacking so far.
 In particular a theoretical framework that is broad enough to formalize the required pre-conditions ensuring the proper functioning of \ac{RB} protocols beyond case-by-case arguments for specific protocols.
  This is unsatisfactory, as the development of higher quality quantum gates and currently relies heavily on a plethora of tailor-made variants of \ac{RB}.
This motivates our current effort at providing a clear rigorous underpinning for \ac{RB} and exploring its underlying mathematical structure, putting all variants of \ac{RB} on a common footing.
%
%

%-> Emphasize extending RB theory beyond specialists...
%-> Filtered RB into the introductions  (This leads to a new class of RB protocols which we call filtered RB, which has no inversion gates and is more flexible.

%)

With this effort we aim to not only \emph{better understand} these protocols, but also to
\emph{increase trust} in them, making it possible to \emph{reliably use them} without a detailed understanding of their inner workings. \newtext{This is a timely effort, as procedures that fit within the \ac{RB} framework, such as linear-cross-entropy benchmarking~\cite{arute2019quantum} and the behaviour of noisy random circuits more generally, have been the topic of significant attention recently \cite{bouland2019complexity,noh2020efficient,dalzell2021random}, including for the purpose of benchmarking \cite{liu2021benchmarking}. Given how we identify linear-cross-entropy benchmarking
as a randomized benchmarking procedure, we relate our general framework to this
timely discussion.}

At the same time our framework allows us to \emph{go significantly beyond current protocols} and establish a series of novel theoretical results and benchmarking schemes, addressing several shortcomings of the current state of the art. Among others, these novel results include a rigorous error bound for generator-style \acl{RB}, a formal equivalence between linear cross entropy benchmarking and \acl{RB} and a novel, scalable method for isolating signals in RB experiments, an absolute requirement if one wants to apply RB to non-standard gate-sets. \newtext{This latter method, which we call \emph{filtered} \ac{RB}, is a significant conceptual improvement over standard \ac{RB} schemes, promising greater flexibility and applicability.} Notably, it also obviates the need for physically implemented inversion gates in \acl{RB} experiments and the preparation of specific input states, making its implementation significantly more straight-forward.
\newtext{As such, this framework therefore also constitutes a solid basis for \emph{developing new schemes}
of randomized benchmarking.}
Altogether these results substantially advance the understanding of the possibilities and requirements of \acl{RB} as a practical tool for estimation and certification.

\section{Overview of results}
In this work, we aim at developing a \emph{mathematically comprehensive framework of \acl{RB} protocols}, synthesizing, generalizing, and substantially strengthening previous work. This paper covers a variety of different aspects of \acf{RB}, from general theorems on the validity of \ac{RB} data, to a detailed study of the classical post-processing of data generated by \ac{RB} and an in-depth discussion of the connection between the outputs of \ac{RB} and average fidelity. \newtext{As our work is often quite technical, we have formulated a series of `take-home messages at the end of this section, summarizing the key takeaways of our work for experimental practice.}

\subsection{A general framework for \acl{RB}}

We begin by providing a general framework that generalizes and covers (to the best of our knowledge) all \ac{RB} procedures currently present in the literature. This can also be thought of as an attempt at {\it a formal definition of \ac{RB} protocols}, and is largely an effort to organize and formalize knowledge already present in the \ac{RB} literature.
{\ac{RB} protocols can be divided} into two separate phases: a data collection phase, and a data processing phase.
\begin{itemize}
\item The \emph{data collection phase} corresponds to the part of the protocol involving the actual quantum computer and can be described as (1) the preparation of a quantum state, (2) the application of a sequence of random quantum operations, capped by (3) an inversion operator mapping the state (ideally) to a specified final state (usually the initial state), upon which (4) a measurement is then performed.
\item
 This process yields estimates of a {success probability} $p(m)$ for different sequence lengths $m$, which constitutes the input to the \emph{data processing phase}. In this phase -which is completely classical-  the data $p(m)$ is fitted to a functional model, generally a linear combination of exponential decays. One can consider the decay rates of these exponential decays as direct measures of quality of the implementation, or further relate it to operational quantities like the average fidelity.
\end{itemize}
Starting with the data collection phase, we write down a general \ac{RB} protocol (alg. \ref{prot:rand_bench}). This protocol depends on a number of input parameters, and by making particular choices for these parameters we can obtain all \ac{RB} protocols currently in the literature. The key parameters are as follows:
\begin{enumerate}
\item {\bf A group} $\gr$, encoding the gates which are applied during the \ac{RB} protocol. A common choice for this group is the multi-qubit Clifford group $\mathbb{C}_q$ but many other choices are possible.
\item {\bf A reference implementation} $\phi_{\mathrm{r}}$ assigning to each element of the group $\gr$ an ideal quantum operation to be implemented. In the standard scenario this map is a \emph{representation} of the group $\gr$ (denoted $\omega$). In general this map need not be a representation, but it is in all known cases obtained from a representation by some fixed mapping. The paradigmatic example of such an implementation map is the standard conjugation action $g\mapsto U_g \rho U_g\ct$ which associates a unitary action to every element $g$ of the group.
\item {\bf A probability distribution} $\nu$ encoding the probability with which gates are selected from $\gr$. In the standard case this probability distribution is simply the uniform distribution over the group. We will also consider the situation where this probability distribution can vary throughout different steps of the protocol.
\item {\bf An ending gate} $\gend$ governing the total operation performed in each \ac{RB} sequence. Typically this is the identity, but other choices are relevant, and it can even be chosen at random.
\end{enumerate}
\begin{figure}[!t]
\centering~
\includegraphics[width=.6\columnwidth]{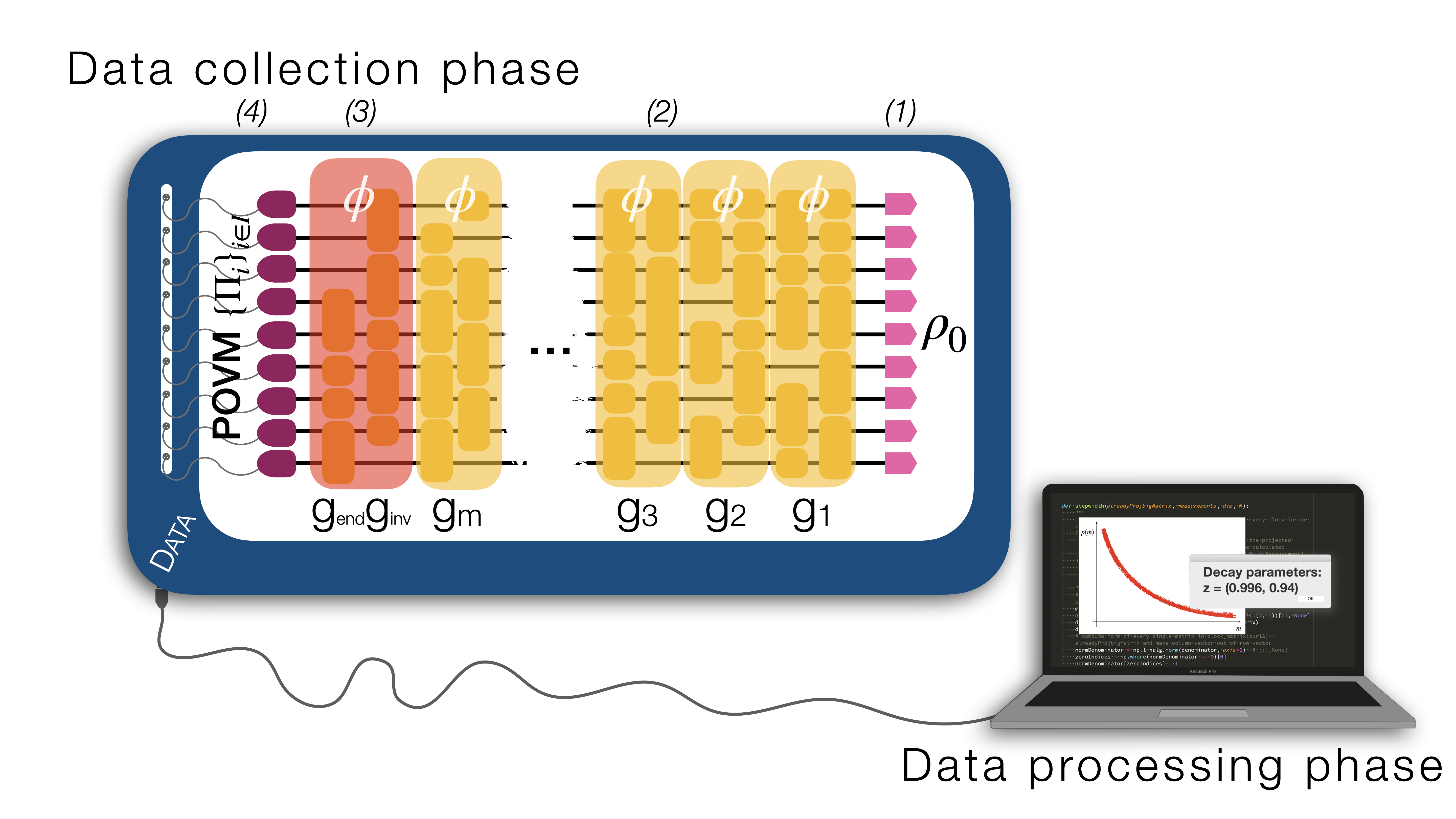}
\caption{The basic structure of \ac{RB}.
The \ac{RB} data collection phase iterates the following steps:
After (1) the preparation of an initial state $\rho_0$, (2) a sequence of $m$ random gates $g_i$ is applied, (3) followed by the gate inverting the sequence $g_\text{inv}$ up to an end gate $\gend$ and (4) a final POVM measurement. In the data processing phase the measurement data, for many random sequences and different sequence lengths is post-processed in a classical computer to extract decay parameters quantifying the imperfections in the implemented gates.}
\end{figure}
Different choices for these key parameters can be collected into classes, yielding a \emph{typology of \acl{RB} procedures},
an overview of which can be seen in
fig.\ \ref{fig:results overview}. This typology consists of three classes:
\begin{itemize}
\item {\bf Uniform \ac{RB}}, which is characterized by uniform random sampling of operations and reference implementations that are representations.
\item
{\bf Interleaved \ac{RB}}, where the reference implementations involve the application of `interleaved' gates.
\item {\bf Non-uniform \ac{RB}}, which is characterized by non-uniform random sampling of operations. This last class comes with two subtypes: {\bf{approximate \ac{RB}}}, where the sampling distribution is close to uniform, and {{\bf subset \ac{RB}}}, where the sampling distribution is very far from being uniform (for instance taking only non-zero values on a small set of generators).
\end{itemize}
These classes of \ac{RB} procedures are motivated by the qualitatively different behaviour of the associated output data $p(m)$, which we will discuss in more detail later. They also partially but not completely align with notions already present in the literature.
In particular we will see that the behaviour of this data is dictated by the group $\gr$ and the reference representation $\omega$.
We can always decompose this representation $\omega$ into a direct sum of \emph{irreducible} sub-representations, i.e., $\omega = \bigoplus_{\lambda\in \Lambda} \sigma_\lambda^{\oplus n_\lambda}$ where the $\sigma_\lambda$ are irreducible (and occur with multiplicity $n_\lambda$).

A key tenet of \ac{RB} is that this decomposition decides the functional form of the output data $p(m)$ as a function of sequence length $m$. More precisely, we expect behaviour of the form
\begin{equation}\label{eq:exp_form_results}
p(m) \approx \sum_{\lambda\in \Lambda}\tr(A_\lambda M_\lambda^m),
\end{equation}
where $A_\lambda, M_\lambda$ are $n_\lambda\times n_\lambda$
matrices encoding
\emph{state preparation and measurement errors} (SPAM), and the quality of gate implementation respectively. This formalizes in a precise way the general idea that \ac{RB} data is well described by a linear combination of exponential decays and allows for the classical processing of \ac{RB} output data, thus providing the connection between the data collection and the data processing phases.  Note, however, that if irreducible sub-representations appear with non-trivial multiplicities the functional {form of
eq.~(\ref{eq:exp_form_results}) includes \emph{matrix} exponential decays}. These can have qualitatively different features than scalar exponential decays, requiring a more sophisticated data processing approach. {It is for instance possible for these matrices to have \emph{complex} eigenvalue pairs, which will lead to damped-oscillation behaviour in randomized benchmarking data.}

\subsection{The functional form of \acl{RB} data}

At the core of the \ac{RB} literature is the promise that \ac{RB} output data has a very specific form, namely that of a linear combination of (matrix) exponential decays (as expressed in eq.\ (\ref{eq:exp_form_results})), decaying with the length of the sequences of random gates. {Moreover, this linear combination is of a specific structure, determined by the implemented gate-set.} However, this functional form of the \ac{RB} output data is not guaranteed by the protocol itself, but is instead derived from assumptions on the noisy implementation of the random quantum operations. In early work this assumption took the form of the \emph{gate-independent noise assumption}. Later, it was realized that this assumption is not satisfactory~\cite{Proctor17} and it was subsequently generalized for standard Clifford \ac{RB} to the more general assumption that the noisy implementations of gates are Markovian and time-independent, and moreover either that the gate-dependent variation of the noise is upper bounded in the diamond norm (in the work of ref.~\citep{wallman2018randomized}), {or lower bounded in average fidelity } (in the work of ref.~\cite{Merkel18}). Here, we provide a series of theorems generalizing these works to (almost) all existing \ac{RB} protocols, justifying eq.\  (\ref{eq:exp_form_results}) in broad circumstances. The theorems we prove make claims of different strength for different classes of \ac{RB} protocols, as per the typology outlined in fig.\ \ref{fig:results overview}.

\begin{figure}[t]
\centering

\begin{tikzpicture}

\draw[thick, rounded corners=8pt,gray] (-2.3,-4) -- (-2.3,4.3) -- (2.9,4.3) -- (2.9,-4) -- (-2.3,-4);

\draw[thick, rounded corners=8pt,gray] (3.1,0.6) -- (3.1,3.7) --(12,3.7) -- (12,0.6) -- (3.1,0.6);

\draw[thick, rounded corners=8pt,gray] (3.1,-4) -- (3.1,-0.5) --(12,-0.5) -- (12,-4) -- (3.1,-4);

\draw[thick, rounded corners=8pt,green] (-2.5,-4.2) -- (-2.5,4.5) --(7.5,4.5) -- (7.5,-4.2) -- (-2.5,-4.2);

\draw[thick, rounded corners=8pt,red] (7.7,0.3) -- (7.7,4.5) --(12.2,4.5) -- (12.2,0.3) -- (7.7,0.3);

\draw[thick, rounded corners=8pt,orange] (7.7,-4.2) -- (7.7,-0.2) --(12.2,-0.2) -- (12.2,-4.2) -- (7.7,-4.2);

\node at (0,4.1) {{\bf Uniform \ac{RB}}};
\node at (5,4.1) {{\bf Covered by theorem} \ref{thm:mother}};
\node at (5,3.45) {{\bf Interleaved \ac{RB}}};
\node at (9.5,4.1) {{\bf Discussed in section} \ref{subsec:interleaved}};
\node at (5,-0.8) {{\bf Non-uniform \ac{RB}}};
\node at (9.5,-3.5) {{\bf Covered by theorem} \ref{thm:subset_mother}};
\node at (4.2,-1.3) {{\bf (Approximate)}};
\node at (8.4,-1.3) {{\bf (Subset)}};

\node at (0, 0) {\makebox{
    {\begin{varwidth}{\linewidth}\begin{itemize}\setlength\itemsep{0em}
        \item Standard Clifford \ac{RB} \cite{CliffordRBPRL,knill2008randomized}
        \item Real \ac{RB} \cite{hashagen2018real}
        \item Simultaneous \ac{RB}~\cite{gambetta2012characterization}
        \item dihedral \ac{RB}* \cite{carignan2015characterizing}
        \item CNOT-dihedral \ac{RB} \cite{cross2016scalable}
        \item Character \ac{RB}* \cite{helsen2019new}
        \item Restricted gate set \ac{RB} \cite{brown2018randomized}
        \item Monomial \ac{RB} \cite{francca2018approximate}
        \item Complete \ac{RB} \cite{chasseur2015complete}
        \item Leakage \ac{RB} (1)** \cite{WallmanEtAl:2016:Leakage}
        \item Leakage \ac{RB} (2)** \cite{wood2018quantification}
        \item Unitarity \ac{RB}** \cite{wallman2015estimating}
        \item Loss \ac{RB}** \cite{WallmanEtAl:2015:LossRates}
        \item Measurement based \ac{RB} \cite{alexander2016randomized}
        \item Logical \ac{RB} \cite{combes2017logical}
        \item Pauli channel tomography* \cite{flammia2019efficient,harper2019efficient}
        \item Linear XEB* \cite{arute2019quantum}
    \end{itemize}\end{varwidth}}
}};

\node at (5, 2) {\makebox{
    {\begin{varwidth}{\linewidth}\begin{itemize}\setlength\itemsep{0em}
        \item Standard interleaved \ac{RB} \cite{magesan2012efficient}
        \item T-gate interleaved \ac{RB} \cite{harper2017estimating}
        \item Iterative \ac{RB} \cite{sheldon2016characterizing}
        \item Individual gate \ac{RB} \cite{PhysRevLett.123.060501}
        \item Hybrid \ac{RB}* \cite{chasseur2017hybrid}
        \item Cycle \ac{RB}* \cite{CycleBenchmarking}
    \end{itemize}\end{varwidth}}
}};

\node at (9,3.1) {\makebox{
    {\begin{varwidth}{\linewidth}\begin{itemize}\setlength\itemsep{0em}
        \item \ac{RB} tomography \cite{kimmel2014robust}
    \end{itemize}\end{varwidth}}
}};

\node at (4.65, -2) {\makebox{
    {\begin{varwidth}{\linewidth}\begin{itemize}\setlength\itemsep{0em}
        \item Approximate \ac{RB} \cite{francca2018approximate}
        \item NIST \ac{RB} \cite{knill2008randomized,boone2019randomized}
    \end{itemize}\end{varwidth}}
}};

\node at (9.3, -2.25) {\makebox{
    {\begin{varwidth}{\linewidth}\begin{itemize}\setlength\itemsep{0em}
        \item Generator \ac{RB} \cite{francca2018approximate,ryan2009randomized}
        \item Direct \ac{RB} \cite{proctor2019direct}
        \item Coset (2-for-1) \ac{RB}* \cite{helsen2019new}
    \end{itemize}\end{varwidth}}
}};

\end{tikzpicture}

\caption{An overview of \ac{RB} schemes, indicating how they fit within our typology (see section \ref{subsec:rb_typology}) of \ac{RB} schemes and what theorem covers the behaviour of their output data (see section \ref{sec:output_data}). A $*$ indicates that the protocol has a non-trivial post-processing scheme, while a $*\!*$ indicates that the protocol in its original specification has no inversion gate. We discuss how this is equal to uniform \ac{RB} (with inversion) together with a post-processing step in section \ref{sec:new_crb}.}
\label{fig:results overview}
\end{figure}
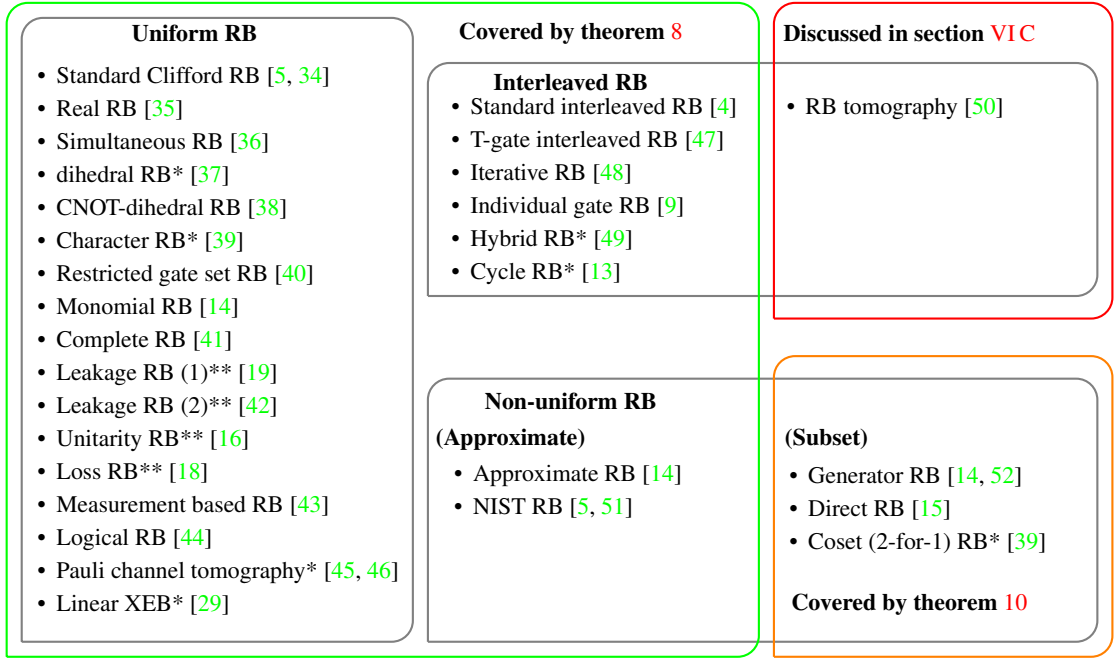

\begin{itemize}
\item We prove that the output data of {\it uniform \ac{RB}} protocols (as per the typology in fig.\ \ref{fig:results overview}) can be described as a linear combination of exponentials, up to an exponentially small error, provided that the gate implementations are Markovian, time-independent and are on average close  \emph{in diamond norm} to an ideal implementation that is a representation. This closeness is independent of the particular \ac{RB} protocol and independent of the underlying Hilbert space dimension.
The complete statement is given as theorem~\ref{thm:mother} that can be summarized as follows:

\begin{theorem}[Informal version of theorem~\ref{thm:mother}]
Consider an \ac{RB} experiment with sequence length $m$, with gates uniformly drawn from a group $\gr$ and implemented through a reference representation $\omega(g)=\bigoplus_{\lambda\in \Lambda} \sigma_\lambda^{\oplus n_\lambda}(g)$. Denote the corresponding noisy implementation on the quantum computer as $\phi(g)$ {(note that this assumes time independent and Markovian noise)}. If we have
\begin{equation}\label{eq:norm_bound_results}
\frac{1}{|\gr|}\sum_{g\in \gr}\dnorm{\omega(g) - \phi(g)}\leq \delta \leq \frac{1}{9},
\end{equation}
then the output data  $p(m)$ of the \ac{RB} experiment obeys the relation
\begin{equation}
\Big|p(m) - \sum_{\lambda\in \Lambda}\tr(A_\lambda M_\lambda^m)\Big|\leq O(\delta^{m})
\end{equation}
with the error exponentially suppressed in $m$. Here $A_\lambda$ and $M_\lambda$ are $n_\lambda\times n_\lambda$ matrices, with $M_\lambda$ depending only on the actual implementation $\phi$.
\end{theorem}

The proof of this theorem relies on a combination of techniques from earlier works: Taking the matrix Fourier transform perspective introduced to \ac{RB} in ref.\ \citep{Merkel18} and combining it with the realization in ref.\ \citep{wallman2018randomized} that the diamond distance (averaged over random gates) is the correct distance measure for the formulation of assumptions on noisy gate implementations. We also make heavy use of the perturbation theory of invariant subspaces of non-normal matrices~\citep{KatoPTLO,StewartSun}. {We note that the specific parameter $1/9$ is an artifact of the proof techniques and probably sub-optimal.}
\item Building on theorem \ref{thm:mother}, we prove multiple theorems for {\it non-uniform \ac{RB}} protocols. The first subtype, approximate \ac{RB}, is covered by theorem \ref{thm:non_uniform}, a direct generalization of theorem \ref{thm:mother}, and also features an exponentially suppressed error.
For the second subtype, subset \ac{RB}, on the other hand, we can only give a weaker statement, guaranteeing that the \ac{RB} output data is described by a linear combination of exponentials up to \emph{constant error} (in sequence length) as long as the sequence length $m$ is taken to be larger than a mixing length $m_{\mathrm{mix}}$. This mixing length indicates the moment where the $m$-fold convolution $\nu^{*m}$ of the probability distribution $\nu$, which governs the sampling of random gates, becomes close to the uniform distribution and is a function of both the initial distribution $\nu$ and the underlying group $\gr$.
We can summarize our result on subset \ac{RB} as follows:

\begin{theorem}[Informal version of theorem \ref{thm:subset_mother}]
Consider a \ac{RB} experiment with sequence length $m$, with gates drawn from a group $\gr$ according to a probability distribution $\nu$ and implemented through a reference representation $\omega(g)=\bigoplus_{\lambda\in \Lambda} \sigma_\lambda^{\oplus n_\lambda}(g)$. Denote the corresponding (noisy) actual implementation on the quantum computer as $\phi(g)$. If we have, for some sequence length $m_{\mathrm{mix}}$ that
\begin{align}
\sum_{g\in \gr}\nu(g)\dnorm{\omega(g) - \phi(g)}&\leq \frac{\delta}{m_{\mathrm{mix}}},\\
\sum_{g \in \gr} \big|\nu^{*m_\mathrm{mix}}(g) - \frac{1}{|\gr|}\big|\leq \delta',
\end{align}
and $\delta+ \delta' \leq 1/9$, then the output data $p(m)$ of the \ac{RB} experiment obeys the relation
\begin{equation}
\Big|p(m) - \sum_{\lambda\in \Lambda}\tr(A_\lambda M_\lambda^{m-m_\mathrm{mix}})\big|\leq O(\delta+\delta')
\end{equation}
with the error bound independent of $m$. Here $A_\lambda$ and $M_\lambda$ are $n_\lambda\times n_\lambda$ matrices, with $M_\lambda$ depending only on the actual implementation $\phi$.
\end{theorem}
This theorem cannot guarantee an exponential error bound, but still improves on the state of the art~\citep{proctor2019direct,francca2018approximate}, both in the generality of the assumptions made and the size of the possible error. Note also the appearance of the ${m_\mathrm{mix}}^{-1}$
term in the average diamond norm deviation. This can be read as the requirement that the generating gates are of sufficiently high quality that any (composite) uniformly randomly chosen gate will be close in diamond norm to its ideal version. In this sense this requirement is of the same stringency as eq.\
(\ref{eq:norm_bound_results}).

\item We discuss the behaviour of {\it interleaved \ac{RB}} protocols, illustrating how standard interleaved \ac{RB}, as well as all but one non-standard interleaved \ac{RB} protocol, are covered by theorem \ref{thm:mother}. We consider two  non-standard interleaved \ac{RB} protocols, namely cycle benchmarking~\citep{CycleBenchmarking}, which is covered by our theorems in a non-trivial way and robust \ac{RB} tomography~\citep{kimmel2014robust}, which is not covered by our theorems. We argue that this is not a weakness of our argument but rather that the \ac{RB} output data of this protocol behaves in a non-standard manner, requiring tailor-made analysis.
\item In section \ref{sec:RBdiamond}, we provide a discussion of the central assumption
${|\gr|}^{-1}\sum_{g\in \gr}\dnorm{\omega(g) - \phi(g)}\leq \delta$, made on the behaviour of noisy gates in the above theorems. We argue that this assumption is a natural one to make (theorem \ref{thm:stability}) and moreover that it can not be replaced by a similar assumption involving the average fidelity without requiring the gate to be exponentially close to perfect in the number of qubits. This also answers an open question posed in ref.\  \cite{Merkel18}
in the negative.
\end{itemize}

The unifying conceptual theme of all of our theorems is the fact that \ac{RB} can be seen as a `power iteration in frequency space'. The behaviour of the output data is dictated by the dominant eigenvalues of a fixed matrix that is obtained from the Fourier transform~\cite{Merkel18} (in a specific sense defined later) of the noisy implementation map $\phi$. Taking powers of this matrix results in the exponential suppression of all but the largest eigenvalues.
\emph{Together,
these results provide a {rigorous} justification for the folkloric knowledge that \ac{RB} protocols function under broad experimental circumstances.}

\subsection{A framework for \acl{RB} data processing}

The second phase of the \ac{RB} protocol, the data processing phase, takes in \ac{RB} output data, which is well-described by a linear combination of exponentials and outputs the decay rates associated with those exponentials. If the data is well described by a single exponential decay this can be done by off-the-shelf
curve fitting procedures, but if the \ac{RB} output data is of a more complex form (such as a linear combination of several exponentials) a more flexible approach is required. Here we provide a self-contained discussion of modern signal processing methods for extracting decay parameters from data with a functional form given by eq.~(\ref{eq:exp_form_results}).
We review signal processing algorithms, in particular the MUSIC and ESPRIT algorithms, that are
at least in principle applicable to the most general form of \ac{RB} output data, even including matrix exponentials.
Beyond that, we discuss theoretical guarantees that were derived for these algorithms and discuss their implications for \ac{RB} data processing.
Building upon these guarantees, we derive a sampling complexity statement that ensures the recovery of decay rates with these algorithms
under measurements with finite statistics.
We complement our analytical discussion with numerical evaluations and simulations that demonstrate the practical performance of these algorithms.
Importantly, our discussions
detail the fundamental limitations
of post-processing \ac{RB} output data featuring many exponential decays.
% This observation motivates
% modifications to the data collection phase that allows for the isolation of exponential decays.

\subsection{A general post-processing scheme for isolating exponential decays}

Even with modern methods, fitting multiple exponential decays is a difficult affair, and in many scenarios one is only interested in a subset of the decay parameters that describe the output data of a particular \ac{RB} experiment. Because of this, several methods have been developed to isolate particular exponential decays. Examples of this include the class of uniform \ac{RB} protocols without inversion gates (indicated with a double asterisk `$**$' in fig.\ \ref{fig:results overview}) and a variety of other protocols that take linear combinations of \ac{RB} output data with different ending gates $\gend$ to isolate particular exponential decays (indicated with a single asterisk `$*$'). In section \ref{sec:new_crb}, we \newtext{give a novel class of protocols called filtered \ac{RB}} that subsumes all these earlier approaches. For simplicity, we only consider uniform \ac{RB}, but our results generalize to other types of \ac{RB}.

This class of protocols is based on the realization that \ac{RB} output data (indexed by an ending gate $\gend$) can be seen as a vector in the group algebra of the group being benchmarked. This allows for the design of {\it filter functions} $\alpha_\lambda:\gr\to \mathbb{C}$, based on the matrix elements of irreducible representations, that isolate exponential decays associated with sub-representations of the ideal implementation of the gates in the group $\gr$. Using these filter functions we can write down a general post-processing scheme for the isolation of exponential decays and prove that it works when the assumptions of theorem \ref{thm:mother} are satisfied.
We prove a theorem of the following form.

\begin{theorem} [Theorem \ref{thm:char_avg} (informal)]
Let $\alpha_\lambda:\gr\to \mathbb{C}$ be the \emph{filter function} associated with the irreducible representation $\sigma_\lambda$ and let $p(m,\gend)$ be the output data associated with a uniform \ac{RB} experiment with ending gate $\gend$, satisfying the condition eq.\ (\ref{eq:norm_bound_results}) with parameter $\delta$.
We have that
\begin{equation}
k_\lambda(m) :=\frac{1}{|\gr|}\sum_{\gend\in \gr} \alpha_\lambda(\gend) p(m,\gend)
\end{equation}
satisfies
\begin{equation}
\big|k_\lambda(m) -\tr(B_\lambda M_\lambda^m)\big| \leq  O(\delta^m)
\end{equation}
with $M_\lambda$ associated with the irreducible sub-representation $\sigma_\lambda$ (as per eq. (\ref{eq:exp_form_results})).
\end{theorem}

Beyond this theoretical result we note that this novel class of protocols allows one (by a simple re-parametrization) to eliminate the need for an explicitly implemented inversion gate in \ac{RB}, making the protocol significantly simpler to implement in practice.

We also give a statistical analysis of this post-processing scheme.
In particular, we prove that if the measurement POVM performed in the \ac{RB} experiment is (proportional to) a state $3$-design, the sample complexity of the complete benchmarking procedure (data collection plus post-processing) is asymptotically independent of the dimension of the underlying Hilbert space for arbitrary benchmarking groups.
This is a strong improvement on previous attempts at such a general post-processing procedure.
Note that the $3$-design condition appearing here plays a similar role in controlling the variance in scalable estimation procedures such as shadow estimation~\cite{huang2020predicting, KlieschRoth:2020:Tutorial}.

We stress, however, that the $3$-design condition is a sufficient condition and there are examples in the literature covered by this post-processing scheme where this condition is not met but the overall procedure is still scalable. In particular we discuss the recently proposed \emph{linear cross entropy benchmarking procedure (XEB)}~\cite{arute2019quantum} in section \ref{subsec:lin_cross_ent}.
We argue that the variant of XEB that performs multiple random gate sequences is an \emph{example of uniform \ac{RB}} (as per the typology) combined with an instance of our general post-processing scheme. Furthermore, we argue that the sample complexity of linear XEB is asymptotically independent of the underlying Hilbert space dimension even though the POVM being measured is not itself a $3$-design.

\subsection{Randomized benchmarking and average fidelity}\label{sec:intro_gate_fidelity_connection}

\newtext{\ac{RB} has originally been designed to estimate the {\it average gate fidelity} of a group of gates. Under the assumption of gate-independent noise, it can be proven (as has already been done in
ref.~\cite{FirstRB}) that the decay rates estimated in an \ac{RB} experiment correspond exactly to the average fidelity of the noise associated to the gates. However, if this condition is relaxed, the connection between these decay rates and the average fidelity is less clear. Even more strongly, it has been argued in ref.~\cite{Proctor17} that due to a so called gauge freedom in the representation of the gate set, the entire premise of a connection between \ac{RB} decay rates and average fidelity may be suspect. This is because the choice of the gauge does not influence the \ac{RB} decay rates, but it does affect the average gate fidelity. Indeed, it has been shown that under some transformations the two quantities may differ by orders of magnitude, even in the gate dependent noise case (where the previously proven connection can be seen as a `natural' gauge choice).}

\newtext{Subsequently proposals have been made to reconnect the average gate fidelity and \ac{RB} decay rates in the context of standard Clifford \ac{RB}: A natural gauge called the depolarizing gauge~\cite{Merkel18} and the noise-in-between-gates framework. Both of these proposals provide an exact connection between the decay rates of \ac{RB} and the average fidelity. However,
several crucial questions of interpretation have still been left open, and in this work we aim to address some of them, and sharpen others.}

\newtext{In section \ref{sec:entanglement_fidelity}, we substantially generalize both proposed connections between decay rates and average fidelity to \ac{RB} with \emph{arbitrary} finite groups. What is more, we argue that these two proposals are in fact \emph{equivalent}. Moreover, we present an explicit example of a completely positive implementation map which is not completely positive in the depolarizing gauge (or equivalently has non-completely positive noise-in-between-gates). This implies that both these interpretations of \ac{RB} decay rates are not fully satisfactory, because they can not be guaranteed to correspond to the average fidelity of a \emph{physical} process. That said, this does not mean
that \ac{RB} decay rates are not useful figures of merit, as they can always be interpreted as
meaningful benchmarks in their own right.}

\newtext{Complementing this, following the approximate approach of ref.~\cite{carignan2018randomized}, we show that the problem of connecting \ac{RB} decay rates with the average gate fidelity can be (approximately) reduced to the deviation between the dominant (ideal) unperturbed eigenvectors and their (implemented) perturbed version in Fourier space. We show that, as long as this overlap is sufficiently close to 1, any gauge choice that corresponds to a CPT channel will connect \ac{RB} parameters to the average gate fidelity. Hence we obtain, under precise conditions, an approximate version of the connection between average fidelity and \ac{RB} decay rates.}

\newtext{More formally, we leverage the Fourier transform framework introduced in ref. \cite{Merkel18} to derive the following expression for the \emph{entanglement fidelity}, which is linearly related to the average fidelity, averaged over all elements of the group as
\begin{equation}\label{eq:intro_entanglement_fidelity}
F_e(\phi,\omega)
=
\dsum f_{\max}(\sigma_\lambda) \,\alpha_{\mathrm{Overlap}}+\alpha_{\mathrm{Res}},
\end{equation}
where $f_{\max}(\sigma_\lambda)$ is the \ac{RB} decay parameter associated with the irreducible subrepresentation $\sigma_\lambda$. In the Fourier framework $f_{\lambda,\max}$ corresponds to the largest eigenvalue of the Fourier transform of the implementation map $\phi$ evaluated at $\sigma_\lambda$. Furthermore, the parameter $\alpha_{\mathrm{Overlap}}$ encodes the overlap between the (left and right) eigenvectors associated with this largest eigenvalue, and the eigenvector of the Fourier transform of the reference representation $\omega$ evaluated at $\sigma_\lambda$.  Finally, the term $\alpha_{\mathrm{Res}}$, the residuum, encodes information about the sub-dominant eigenspaces of the Fourier transform. The factors $\alpha_{\mathrm{Overlap}},\alpha_{\mathrm{Res}}$ are gauge dependent. We give bounds on the overlap and residuum in terms of the deviation of $\phi$ from the reference $\omega$ and discuss relevant scenarios where these terms contribute only negligibly to the entanglement fidelity (and thus when \ac{RB} decay data corresponds approximately to an average fidelity).}

%\newtext{
%{\it This is a first unfinished attempt at rewriting the conclusion to be more like the referees want it.}}

\subsection{Non-technical discussion}

\newtext{In this work, we develop a comprehensive theory of \acf{RB}. Our main motivation has been our desire to give a mathematical framework for \ac{RB} and to classify known schemes. It should be clear, however, that our work goes significantly beyond a mere classification of what is present in the literature. Since our work is in parts rather technical, we will in the following formulate a series of `take home messages': Actionable advice for experimentalists interested in using \ac{RB} in the laboratory
and developing new protocols to suit their needs.
\begin{enumerate}
  \item \textbf{\ac{RB} gives exponential decays under broad (Markovian) circumstances.}
  		Confirming experimental intuition, and extending earlier results for specific groups, our main result (theorem \ref{thm:mother}) proves that \ac{RB} protocols behave (up to an exponentially small correction factor) as expected whenever the noise afflicting the gate-set is Markovian and time independent. Because the correction factor is so small, any deviation from the prescribed functional form can in fact be taken as evidence of non-Markovian or time-dependent noise processes (as suggested earlier by ref.~\cite{wallman2018randomized}). We do wish to emphasize that the error term in theorem \ref{thm:mother} can be quite significant for small sequence lengths. Hence we recommend as a rule of thumb that \ac{RB} experiments should not include very short ($m\leq 5$) sequence lengths, especially when strong gate-dependent (but Markovian) noise is suspected, as this might bias the estimator for the decay rate.
  \item \textbf{\ac{RB} is broadly resistant to deviations from uniform sampling.}
  		Similar to robustness against gate-dependent Markovian noise, we prove
		(theorem \ref{thm:non_uniform}) that \ac{RB} gives correct results even when the group is not being sampled exactly uniformly. This broadly justifies the use of (generically applicable) Markov chain techniques for sampling group elements~\cite{francca2018approximate}, overcoming a key technical hurdle in running \ac{RB} protocols with new groups.

  \item \textbf{The decay rates given by \ac{RB} can be interpreted as an average fidelity (but caveats apply).}
  We find that the decay rates of general \ac{RB} experiments can always be exactly associated to the average fidelity of a fixed process, however, this process need not be physical (i.e., it does not always correspond to a CPTP map). Alternatively, we show that \ac{RB} decay rates can always be connected approximately to a the average fidelity of a physical process, but the degree of approximation is dependent upon external beliefs about the underlying noise process. Hence, we believe the interpretation of \ac{RB} decay rates as an average fidelity to be broadly valid, but subject to technical caveats.

\end{enumerate}
These three messages can be considered folk knowledge in the \ac{RB} community, for which we provide a rigorous underpinning. However, our work also contains new conceptual developments, notably
the following.
\begin{enumerate}
  \item \textbf{Filtering scalably extends \ac{RB} to a large class of groups.}
  As formalized in section \ref{sec:data_processing}, a major practical hurdle to applying \ac{RB} with \emph{arbitrary} finite groups, is that this generically requires the fitting of output data to multi-exponential decays. This is a difficult problem both in theory and in practice and it has so far contributed to the limited experimental use of \ac{RB} beyond a few groups (such as the Clifford group). Our new procedure, which we call \emph{filtering} (or filtered \ac{RB}), takes a major step towards solving this problem by giving a generic procedure for isolating exponential decays in a fully scalable manner. \newtext{This approach is discussed in great detail in 
  sections \ref{sec:data_processing} and  \ref{sec:new_crb}, with filter functions being introduced in section   \ref{sec:new_crb}.A.}

  \item \textbf{Inversion gates are not required for \ac{RB}.} Another key practical difficulty in performing randomized benchmarking has been the necessity to compute and implement a global inversion gate.
  However, filtered \ac{RB} has the bonus property that it does not require the application of inverses. Instead a random noisy gate sequence can be directly compared to a perfect classical simulated version to extract the same \ac{RB} decay rates., making the quantum part of the protocol significantly easier to implement. However, this simplicity is gained at a (constant) extra sampling overhead, as the inversion gate in standard \ac{RB} also suppresses the sampling complexity~\cite{helsen2019multiqubit}.
\end{enumerate}
With these new contributions, our framework serves as a convenient basis to design new schemes that \emph{come with rigorous performance bounds built in}. We expect this to facilitate and accelerate the development of more sophisticated and tailor-made benchmarking schemes as required by experimental practitioners. Steps in this direction have already been made \cite{MatchgateRB,LinghangKongCompactGroups,RandomSequences}. \newtext{In particular, ref.~\cite{LinghangKongCompactGroups} explores the framework put forth here for
continuous groups of quantum gates.}
}
\subsection{Structure of this work}

In section \ref{sec:prelim_one}, we discuss mathematical preliminaries: We set the notation for the rest of the work and recall standard notions from representation theory. This section can be skipped by experienced readers.

In section \ref{sec:fourier_and_pert}, we discuss implementation maps: linear maps from finite groups to super-operators, a central concept in our treatment of \ac{RB}. We also give an introduction into matrix valued Fourier theory and explicitly state several results from the perturbation theory of non-normal matrices which we use throughout the rest of the work.

In section \ref{sec:rand_bench}, we give a general framework for \ac{RB}, {with its two phases}: the \emph{data collection} and
\emph{data processing phases},  and give a general protocol for the data collection phase. This protocol, which depends on a range of input parameters, covers (the data collection phase of) all known versions of \ac{RB}. We also discuss a typology of \ac{RB} schemes, dividing up the known protocols into a few generic classes.

In section \ref{sec:output_data}, we present a series of general theorems that govern the behaviour of the output data of a \ac{RB} protocol. We confirm the folklore knowledge that \ac{RB} data is well described by a linear combination of (matrix) exponentials, under some general assumptions.

In section \ref{sec:data_processing}, we discuss general procedures for extracting decay parameters from \ac{RB} output data. We discuss implementation and general limitations and prove a sampling complexity statement for \ac{RB}.

In section \ref{sec:new_crb}, we propose a general post-processing method for isolating exponential decays associated with particular sub-representations. We argue that this post processing method covers many previously proposed procedures. We also prove a sufficient condition under which this post processing scheme is scalable for any \ac{RB} protocol and analyze linear cross-entropy benchmarking as an example.

In section \ref{sec:fidelity_interpretation},
we discuss the relation between the decay rates generated by \ac{RB} and the average fidelity, focusing in particular on the gauge freedom in the presentation of the underlying noise channels.

Finally, in section \ref{sec:RBdiamond}, we \newtext{finally} argue that the assumptions made in section \ref{sec:output_data} are natural and in some sense necessary for the correct behaviour of \ac{RB}.

\section{Preliminaries: quantum channels and group representations}\label{sec:prelim_one}

In this section, we will go over some of the basic mathematical machinery needed to talk about \acf{RB}\, and prove our central theorems. We will discuss quantum channels and their matrix representations (section\ \ref{subsec:channels}), and groups and group representations (section\ \ref{subsec:rep_theory}). This is fairly standard material, and beyond the setting of notation it can be skipped by an experienced reader.

We begin by setting \newtext{the stage and introducing} some basic notation used throughout our work. We will denote complex vector spaces by $V$ or more explicitly by $\mathbb{C}^d$. We denote by $\mc{M}_d$ the vector space of complex linear transformation of $\mathbb{C}^d$ and by $\mc{S}_d$ the space of linear transformations of $\mc{M}_d$, often called super-operators. Here $d$ is an integer that in many cases can be thought of as being a power of two $(d=2^q)$, however, all theorems are valid for general $d$ unless explicitly stated. We will denote by $\Tr_{V}$ the partial trace over a tensor factor $V$ (of an implied tensor product space $V\otimes W$ for some $W$). {Finally we will denote the complex conjugate by a bar  (i.e., $\overline{A}$ is the entry-wise complex conjugate of $A$)}

\subsection{Quantum channels and the operator-matrix representation}\label{subsec:channels}

Unitary operations as they are generated by quantum gates -- in the focus of
attention in this work -- are quantum channels. Formally,
quantum channels are super-operators, that is elements of $\mc{S}_d$, that are trace preserving and completely positive.
In order to represent quantum channels (and elements of $\mc{S}_d$ more generally), we make use of the \emph{operator matrix representation}. Given a quantum channel $\mathcal E\in \mc{S}_d$, we can represent it as an element of $\mc{M}_{d^2}$ by choosing an orthonormal basis (with respect to the trace or Hilbert-Schmidt inner product) $\set{b_j}_{j=1}^{d^2}$ for $\mc{M}_d$. Thus $\mc{E}$ (abusing notation) is a $d^2 \times d^2$ matrix with components
\begin{equation}
	\mc{E}_{j,k} \coloneqq   \Tr \left[ b_j^\dagger \, \mathcal{E}(b_k) \right].
\end{equation}
Analogously, (density) matrices $\rho \in \mc{M}_d$ can be represented as vectors,
\begin{equation}
\kett{\rho} = \begin{pmatrix} \rho_1 \\ \rho_2 \\ \dots \\ \rho_{d^2} \end{pmatrix} \qquad \text{with} \qquad
\rho_k\coloneqq  \Tr \left[b_k^\dagger\, \rho \right] .
\end{equation}
Note that the action $\mathcal{E}(\rho)$ now corresponds to a matrix-vector multiplication $\mc{E} \kett{\rho}$ and the concatenation of two channels $\mathcal{E}$ and $\mathcal{E'}$ into a matrix multiplication $\mc{E}\mc{E}'$.
We can analogously write a (POVM element) matrix $\Pi \in \mc{M}_d$ as a co-vector
\begin{equation}
\braa{\Pi} = \left( \Pi_1 \ \Pi_2 \dots \Pi_{d^2} \right)
\qquad \text{with} \qquad
\Pi_k \coloneqq \langle\!\langle \Pi , b_k \rangle\!\rangle = \Tr \left[ \Pi \, b_k \right] .
\end{equation}
With this, the probability to obtain an outcome described by the POVM element $\Pi$ when measuring $\rho$ is $p(\Pi|\rho)= \langle\!\langle \Pi, \rho \rangle\!\rangle= \Tr [\Pi \rho ] $.

\subsection{Representations of groups}\label{subsec:rep_theory}
At the heart of our discussion will be notions of representations of groups. In this section, we will hence recall some basic facts about the representations of finite (and compact) groups over complex vector spaces, with a focus on their use in quantum computation. For a more in depth treatment of this topic we refer to refs.~\citep{goodman2000representations,fulton2013representation}. 
\newtext{
	We in this work restrict our attention to finite groups keeping the notation more concise.  
	Most results can be analogously stated for continuous, compact groups and derived following the same strategy. 
	Ref.~\cite{LinghangKongCompactGroups} carefully discusses the required modifications and gives explicite reformulations for continuous compact groups. 
}

% Indeed, the central object in \ac{RB} is the notion of a \emph{group}, denoted by $\gr$. This is a set of elements (in our specific case, of quantum gates) that is closed under multiplication and contain a group identity and inverses.

\subsubsection{Representations}
% Groups can be embedded into the space of linear transformations of a vector space. This is called a representation of the group.
Let $\gr$ be a finite group and consider the space $\mc{M}_d$ of linear transformations of $\mathbb{C}^d$. A representation $\omega$ is a map $\omega:\gr \rightarrow \mc{M}_d$ that preserves the group multiplication, i.e.,
\begin{equation}
\omega(g)\omega(h) = \omega(gh),\;\;\;\;\;\; \forall g,h \in \gr.
\end{equation}
We will require the operators $\omega(g)$ to be unitary as well (for finite groups this can always be done).
% The line between an abstract group and its representations is sometimes thin. As an example, we typically define the single qubit Pauli group $\mathbb{P}_1$ as being the group generated by the matrices $X$ and $Z$ (under matrix multiplication). However, we can also think of the Pauli group as being an abstract group (generated by elements $x,z$) and the $X$ and $Z$ matrices as being the image of a representation of this abstract group. In this case we call this representation the defining representation, or quantum circuit representation.  More broadly, and more salient to our setting we can consider any set of quantum circuits $\{U_g\}$ (a subset of unitary matrices of $\mc{M}_{2^q}$) that is closed under composition (multiplication) as a defining representation for an abstract group.

% Thinking in this abstract way has merit, as it allows us to think of other representations of this group. We could for instance consider a trivial representation $\omega_{\mathrm{tr}}:\gr \to \mc{M}_1:g\mapsto 1$. Another example of a representation is the conjugate representation, which acts on the space of linear operators (and hence forms super-operators): $\omega_{\mathrm{cong}}: \gr \to \mc{S}_d: g \mapsto U_g \cdot U_g^\dagger $ where $U_g$ is the quantum circuit representation.

\subsubsection{Reducible and irreducible representations}
If there is a non-trivial subspace $W$ of $\mathbb{C}^d$ such that for all vectors $w\in W$ we have
\begin{equation}\label{eq:sub-rep}
\omega(g)w\in W,\;\;\;\;\;\;\forall g\in \gr,
\end{equation}
then the representation $\omega$ is called \emph{reducible}. The restriction of $\omega$ to the subspace $W$ is also a representation, which we call a \emph{sub-representation} of $\omega$. If there are no non-trivial subspaces $W$  such that
eq.\  (\ref{eq:sub-rep}) holds the representation $\omega$ is called \emph{irreducible}. We will generally reserve the letter $\sigma$ to denote irreducible representations.
Two representations $\omega,\omega'$ of a group $\gr$ are called \emph{equivalent} if there exists an invertible linear map $T$ such that
\begin{equation}\label{eq:equivalent}
T\omega(g) = \omega'(g)T,\;\;\;\;\;\;\;  \forall g\in \gr.
\end{equation}
We will denote this by $\omega\simeq \omega'$.
For finite groups $\gr$ the set of irreducible representations (up to the above equivalence) is finite. We will denote it by $\Irr$.

\subsubsection{Sums, products, and Maschke's Lemma}
We will make use of sums and products of representations.
Given representations $\omega, \omega'$, the maps
\begin{align}
&\omega\oplus\omega':\gr \to \mc{M}_d \oplus \mc{M}_{d'}:g \mapsto \omega(g)\oplus \omega'(g),\\
&\omega\otimes\omega':\gr \to \mc{M}_d \otimes \mc{M}_{d'}:g \mapsto \omega(g)\otimes \omega'(g),
\end{align}
are again representations. They are, however, generally not irreducible (even if $\omega$ and $\omega'$ are).
However, Maschke's Lemma ensures that every representation $\omega$ of a
group can be uniquely written as a direct sum of irreducible representations, that is
\begin{equation}\label{eq:Maschke_Lemma}
\omega(g) \simeq \bigoplus_{\lambda\in \Lambda} \sigma_\lambda(g)^{\oplus n_\lambda},\;\;\;\;\;\;\forall g\in \gr
\end{equation}
where the index set $\Lambda$ is a subset of the set $\Irr$ and $n_\lambda$ is an integer denoting the number of copies (or multiplicity) of $\sigma_\lambda$ present in $\omega$.

\subsubsection{Characters}
\emph{Characters} are a central object in representation theory, given by the trace of a representation.
\begin{definition}[Character of a representation]
	The character $\chi_\omega$ of a representation $\omega$ of a group $\gr$ is defined as
	\begin{equation}
	\chi_\omega (g) = \Tr [\omega (g) ] .
	\end{equation}
\end{definition}
One of the most important properties for characters of irreducible representations is the following orthogonality relation.
\begin{proposition}[Orthogonality formula]
	Let $\chi_\lambda,\chi_{\lambda'}$ be the characters of two irreducible representations $\sigma_\lambda,\sigma_{\lambda'}$ of a group $\gr$.
	Then
	\begin{equation}
	\frac{1}{|\gr|} \sum_{g \in \gr} \overline{\chi}_\lambda(g)  \chi_{\lambda'} (g)
	=
	\begin{cases}
	1 & \text{if } \sigma_{\lambda} \simeq \sigma_{\lambda'} \\
	0 & \text{if } \sigma_{\lambda} \not\simeq \sigma_{\lambda'}
	\end{cases}.
	\end{equation}
\end{proposition}

\subsubsection{Projections onto irreducible representation}

Given a representation $\omega = \bigoplus_{\lambda\in \Lambda} \sigma_\lambda^{\oplus n_\lambda}$ on a vector space $V_\omega = \bigoplus V_\lambda^{\oplus n_\lambda}$ we can choose a basis $\set{v_j^\lambda \mid j \in \set{1,\dots ,d_\lambda}}$ for each $V_\lambda$. Each vector $v$ in $V_\omega$ can thus be written as a linear combination $v = \sum_{\lambda\in \Lambda} \sum_{j=1}^{d_\lambda} c_j^\lambda \, v^\lambda_j$. We can conversely identify the basis vector components of any vector $v$ by application of an appropriate projection $\mathrm{P}_j^\lambda$, such that $\mathrm{P}_j^\lambda\, v = c_j^\lambda \, v_j^\lambda$, where
\begin{equation}\label{eq:projections_onto_irred_basis}
\mathrm{P}_j^\lambda
=
\frac{d_\lambda}{\left|\gr \right|} \sum_{g \in \gr} \big[\overline{\sigma}_{\lambda} (g)\big]_{{j,j}} \, \omega(g).
\end{equation}
Note that, in order to construct these projections, the knowledge of the diagonal elements of the corresponding irreducible representation $\sigma_\lambda$ is required. However,
it is also possible to project any vector onto distinct irreducible subspaces (up to multiplicity) by using only knowledge of the character of a representation:
\begin{equation}\label{eq:projections_onto_irred_subspace}
\mathrm{P}_\lambda
=
\frac{d_\lambda}{\left|\gr\right|} \sum_{g \in \gr}   \overline{\chi}_\lambda (g) \, \omega(g).
\end{equation}
This last formula follows simply from the definition of the character as $\chi_\lambda(g) = \tr( \sigma_\lambda(g))$.

\section{Fourier transforms and perturbation theory of implementation maps}\label{sec:fourier_and_pert}
 In this section, we review the concept of \emph{group implementation maps} and their
 \emph{Fourier theory} (section\ \ref{subsec:matrix_fourier}). Mathematically this corresponds to non-commutative harmonic analysis of matrix-valued functions. We also discuss perturbation theory  for non-normal matrices. This material is somewhat less well known, so we spend more time discussing these concepts.

\subsection{Implementation maps}\label{subsec:impl_maps}
Given a group $\gr$, we can assign quantum circuits (elements of $U(d)$) to each group element, which gives rise to a representation of the group. However, in practice quantum circuits will not be executed perfectly, but rather include noise. This noise can be modelled by a quantum channel, and we can thus envision assigning to each group element a quantum channel modelling the real implementation of that circuit. These quantum channels can be composed, but this composition will not necessarily maintain group structure and will thus in general not form a representation. However, we can define the more general concept of an `implementation map' $\phi$, which is a function from a finite group $\gr$ to the space of super-operators $\mc{S}_d$,
\begin{equation}
\phi:\gr \to \mc{S}_d\, ,
\end{equation}
where we will usually assume that $\phi(g)$ is a trace non-increasing quantum channel for all $g$. If we want to draw explicit attention to this fact we will call $\phi$ completely positive if and only if $\phi(g)$ is completely positive for all $g\in \gr$. Finally, note that if $\phi(g)\phi(h) = \phi(gh)$ for all $g,h \in \gr$ then $\phi$ would be a representation.
We can think of the implementation map as being an abstract presentation of the noisy implementation of the group elements, which depends on the noise processes in the quantum computer but also on other choices such as the compilation of circuits into elementary gates.

\subsection{Fourier transforms of implementation maps}\label{subsec:matrix_fourier}
When considering an implementation map one can ask precisely when it is a representation, and failing that, if it is close to a representation (in some reasonable way). To answer this question we need to introduce some mathematical machinery. This machinery was first introduced into the theory of \acf{RB} by ref.\ \citep{Merkel18}, based on work by Gowers \& Hatami \citep{GowersHatami}, which is itself a partial review of older mathematical work.
In this section, we will consider general maps $\phi$ from a group $\gr$ to a space of $d\times d$ matrices $\mc{M}_d$. Thinking of $\mc{S}_d$ as a matrix space, our notion of implementation map can be seen to be a special case of these maps. Given a map $\phi$ we define its \emph{Fourier transform} $\mc{F}(\phi)$ as
\begin{equation}
\mc{F}(\phi)[\sigma_\lambda]  = \gsum{g} \overline{\sigma_\lambda}(g)\otimes \phi(g)
\end{equation}
for all $\lambda \in \Irr$. So the Fourier transform $\mc{F}(\phi)$ is a function from the set $\Irr$ of irreducible representations of $\gr$ to a set of matrices. This definition has all the properties of a Fourier transform. Firstly, it has an inverse transform, which maps $\mc{F}(\phi)$ back to $\phi$, given by
\begin{equation}
\mc{F}^{-1}\big[\mc{F}(\phi)\big](g) = \sum_{\lambda \in \Irr} d_\lambda \tr_{V_\lambda} \big(\mc{F}(\phi)[\sigma_\lambda] \overline{\sigma_\lambda}(g^{-1})\otimes \1\big)
\end{equation}

 for all $g \in \gr$ and where $d_\lambda$ is the dimension of $V_\lambda$, the space on which the representation $\sigma_\lambda$ acts.

Secondly, it has the correct behaviour with respect to convolutions of implementation maps: the Fourier transform of a convolution corresponds to a product of Fourier transforms. Recalling the definition of a convolution of two implementation maps $\phi,\phi'$
\begin{equation}\label{eq:convolution}
\phi* \phi' (g) = \gsum{g'}\phi(gg'^{-1})\phi'(g')
\end{equation}
we can easily see the following
\begin{equation}
\mc{F}(\phi* \phi')[\sigma_\lambda] = \ggsum{g,g'} \overline{\sigma_\lambda}(g)\otimes \phi(gg'^{-1})\phi'(g') = \ggsum{g,g'} \overline{\sigma_\lambda}(gg')\otimes \phi(g)\phi'(g') = \mc{F}(\phi)[\sigma_\lambda] \mc{F}(\phi')[\sigma_\lambda]
\end{equation}
for all $\lambda \in \Irr$. Another useful property is the Parseval identity
\begin{equation}\label{Parseval_identity}
\gsum{g} \tr\big(\phi(g)\ct \phi'(g)\big) = \sum_{\lambda \in \Irr} d_\lambda \tr(\mc{F}(\phi)[\sigma_\lambda]\ct \mc{F}(\phi')[\sigma_\lambda]) .
\end{equation}
Finally, we note that the Fourier transform (evaluated at an irreducible representation) of a representation is an orthogonal projector with its rank given by the multiplicity of that irreducible representation. To see this, consider a representation $\omega = \bigoplus_{\lambda \in \Lambda} \sigma_\lambda^{\oplus n_\lambda}$. We have that
\begin{equation}\label{eq:FT_projection}
(\mc{F}(\omega)[\sigma_{\lambda'}])^2 = \frac{1}{|\gr|^2}\sum_{g,g'\in\gr} \overline{\sigma_{\lambda'}}(gg')\otimes \omega(gg') = \frac{|\gr|}{|\gr|^2} \sum_{g\in \gr}\overline{\sigma_{\lambda'}}(g)\otimes \omega(g) = \mc{F}(\omega)[\sigma_{\lambda'}]
\end{equation}
for all $\lambda' \in \Irr$. Moreover for $\lambda'\in \Lambda$ we have
\begin{equation}\label{eq:FT_projection_rank}
\tr(\mc{F}(\omega)[\sigma_{\lambda'}]) = \gsum{g}\overline{\chi_{\sigma_{\lambda'}}}(g)\chi_{\omega}(g) = n_{\lambda'}
\end{equation}
by the character orthogonality formula.
\subsubsection*{Fourier operators}
We also give another, useful way to think about the matrix Fourier transform, namely in terms of what we call {\it Fourier operators}.

Note that the set of maps $\phi \to \mc{S}_d$ can be seen as a vector space under point-wise addition (of the super-operators).
We can further lift this vector space to an algebra by considering the convolution operator $*$ (as defined in eq.\ (\ref{eq:convolution})) on the functions in the vector space. We can construct a faithful {(i.e., injective)} matrix representation of this algebra as
\begin{equation}\label{eq:fourier_operator}
F(\phi) = \frac{1}{|\gr|}\sum_{g\in \gr}\bigoplus_{\lambda\in \Irr} \overline{\sigma}_{\lambda}(g)\otimes \phi(g) = \frac{1}{|\gr|}\sum_{g\in \gr} \overline{\omega_{\gr}}(g)\otimes \phi(g),
\end{equation}
with $\omega_{\gr} = \bigoplus_{\lambda\in \Irr} \sigma_\lambda$.  This is just the Fourier transform of $\phi$ gathered in a direct sum (note that $\Irr$, and hence the sum, is finite for any finite group). By the Peter-Weyl theorem for finite groups one can equally well think of $\omega_\gr(g)$ as an element of the group algebra $\mathbb{C}[\gr]$ associated with $\gr$, we will not be using this point of view explicitly. We will call $F(\phi)$ the {\it Fourier operator} of $\phi$. From the properties of the Fourier transform we immediately see that $F(\phi)F(\phi') = F(\phi*\phi')$.  It will be useful to equip the algebra of Fourier operators with several norms, based on the diamond norm $\dnorm{\cdot}$ for $\mc{S}_d$ (In principle this construction will work for any norm on $\mc{S}_d$). We define
\begin{align}\label{eq:fourier_norms}
\norm{F(\phi)}_{\mathrm{max}} &= \max_{g\in \gr} \dnorm{\tr_{V_{\omega_{\gr}}}\big(D_{\omega_{\gr}}\overline{\omega_\gr}(g^{-1})\otimes \1 F(\phi)\big) } = \max_{g\in \gr} \dnorm{\phi(g)},\\
\norm{F(\phi)}_{\mathrm{m}} &= \frac{1}{|\gr|}\sum_{g\in \gr} \dnorm{\tr_{V_{\omega_{\gr}}}\big(D_{\omega_\gr}\overline{\omega_\gr}(g^{-1})\otimes \1 F(\phi) \big)} =\frac{1}{|\gr|}\sum_{g\in \gr} \dnorm{\phi(g)},
\end{align}
where $D_{\omega'} =  \bigoplus_{\lambda\in \Irr} d_{\lambda}\1_{\lambda}$ collects the relevant dimensional factors and where the second equality follows from the properties of the Fourier transform. These norms are bona fide matrix norms on the algebra of Fourier operators, notably they are sub-multiplicative viz.,
\begin{eqnarray}
\norm{F(\phi)F(\phi')}_{\mathrm{max}} &=& \norm{F(\phi*\phi')}_{\mathrm{max}}
= \max_{g \in \gr} \dnorm{\phi*\phi'(g)}\nonumber\\
&\leq& \max_{g \in \gr} \frac{1}{|\gr|} \sum_{\hat{g}\in \gr}\dnorm{\phi(g\hat{g}^{-1})}\dnorm{\phi'(\hat{g})} \nonumber\\
&\leq  &\max_{g,\hat{g} \in \gr} \dnorm{\phi(g)}\dnorm{\phi'(\hat{g})}\nonumber\\
&=&   \norm{F(\phi)}_{\mathrm{max}}\norm{F(\phi')}_{\mathrm{max}}
\end{eqnarray}
and similarly for $\norm{\cdot}_{\mathrm{m}}$. We also have an identity involving both norms
\begin{eqnarray}\label{eq:norm_id}
\norm{F(\phi)F(\phi')}_{\mathrm{max}} &=& \norm{F(\phi*\phi')}_{\mathrm{max}} = \max_{g \in \gr} \dnorm{\phi*\phi'(g)}\\
&\leq& \max_{g \in \gr} \frac{1}{|\gr|} \sum_{\hat{g}\in \gr}\dnorm{\phi(g\hat{g}^{-1})}\dnorm{\phi'(\hat{g})} \nonumber\\
&=&\norm{F(\phi)}_{\mathrm{max}}\norm{F(\phi')}_{\mathrm{m}}
\end{eqnarray}
that will be helpful later.

\subsection{Perturbation theory}\label{subsec:pert_theort}
In this section, we gather some technical tools from matrix perturbation theory that are essential to the many of the proofs in this paper. Our sources for this section are the standard books of Stewart and Sun \citep{StewartSun} and Kato \citep{KatoPTLO}. For the rest of this section, we will assume that $\norm{\cdot}$ denotes a sub-multiplicative matrix norm on $\mc{M}_d$, i.e.,
$\norm{AB}\leq \norm{A}\norm{B}$ for all $A,B\in \mc{M}_d$.

Let $A\in \mc{M}_d$ be a complex Hermitian matrix. Assume that there exists a unitary matrix $X = [X_1,X_2]$ such that the columns of $X_1$ and $X_2$ span invariant subspaces of $A$, that is
\begin{equation}
[X_1,X_2]\ct A [X_1,X_2] = \begin{pmatrix} A_1 & 0 \\ 0 & A_2\end{pmatrix}
\end{equation}
with $A_1 = X_1\ct A X_1$ and $A_2 = X_2\ct A X_2$. We call this a \emph{spectral resolution} of $A$. We can think of $A_1,A_2$ as the matrix $A$ restricted to subspaces of $\mathbb{C}^d$ spanned by the columns of $X_1,X_2$, respectively, and furthermore we assume that the eigenvalues of $A_1$ are all distinct from the ones of $A_2$: the subspaces are then said to be \emph{simple}. These subspaces are invariant under the action of $A$ in the sense that $AX_1 = X_1A_1$ and are hence called invariant subspaces.  It turns out that spectral resolutions, and invariant subspaces more generally, are stable against (small) perturbations. That is, given a perturbation matrix $E$ (not necessarily Hermitian) we can find matrices $R = [R_1,R_2]$ and $L = [L_1,L_2]$ such that $L\ct = R^{-1}$ and
\begin{equation}
[L_1,L_2]\ct (A+E) [R_1,R_2] = \begin{pmatrix} A'_1 & 0 \\ 0 & A'_2\end{pmatrix}
\end{equation}
for some $A'_1,A'_2$ and the matrices $R,L$ are close to $X$ in a well specified sense. This is what one would expect from a perturbation theorem. It, however,
only holds if the perturbation $E$ is small with respect to the difference between $A_1$ and $A_2$. This difference is made quantitative by the so-called separation function:
\begin{equation}
\mathrm{sep}(A_1,A_2) = \min_{Z\neq 0} \frac{\norm{A_1 Z - ZA_2}}{\norm{X_1 Z X_2\ct}}.
\end{equation}
This separation function has some rather nice properties. Firstly it is symmetric in its arguments:
\begin{align}
\mathrm{sep}(A_1,A_2) &= \mathrm{sep}(A_2,A_1)\, .
\end{align}
Secondly it is stable against perturbations, i.e., given a perturbation $A+ E$ of $A$ we have
\begin{equation}
|\mathrm{sep}(A_1+ E_1,A_2+ E_2) -\mathrm{sep}(A_1,A_2)| \leq \norm{E_1} + \norm{E_2}.
\end{equation}
With this function we can state the following theorem, which can be derived from theorem 2.8 in ref.\ \cite[p.~238]{StewartSun}.
\begin{theorem}[\citep{StewartSun}]\label{thm:subspace_pert}
Let $A$ be a complex Hermitian matrix with spectral resolution $\mathrm{diag}(A_1,A_2)$ induced by a unitary $X = [X_1,X_2]$. Also, let $\norm{\cdot}$ be a matrix norm. Now let $E$ be a complex matrix. If $E$ has the properties
\begin{align}\label{eq:pert_props}
\frac{\norm{X_1\ct EX_2} \norm{X_2\ct E X_1}}{\big( \mathrm{sep}(A_1,A_2)- \norm{X_1\ct E X_1} - \norm{X_2\ct E X_2}\big)^2 } &< \frac{1}{4},\\
\frac{\norm{X_1\ct EX_2} \norm{X_2\ct E X_1} +\norm{X_1\ct E X_2} \norm{X_1\ct E X_2} }{\big( \mathrm{sep}(A_1,A_2)- \norm{X_1\ct E X_1} - \norm{X_2\ct E X_2}\big)^2 } &< \frac{1}{2}
\end{align}
then there exist matrices $P_1,P_2$ such that
\begin{align}
\norm{P_1} &\leq \frac{\norm{X_2\ct E X_1}}{\mathrm{sep}(A_1,A_2) - \norm{X_1\ct E X_1} - \norm{X_2\ct E X_2}},\\
\norm{P_2} &\leq \frac{\norm{X_2\ct E X_1}}{\mathrm{sep}(A_1,A_2) - \norm{X_1\ct E X_1 + X_1\ct EX_2P_1} - \norm{X_2\ct E X_2- P_1X_1\ct EX_2}}
\end{align}
and
\begin{equation}\label{eq:matrix_2_blocks}
[L_1, L_2]\ct (A+E)[R_1,R_2] = \begin{pmatrix} A'_1 & 0 \\ 0 & A'_2\end{pmatrix}
\end{equation}
with
\begin{align}
[R_1,R_2] = [X_1,X_2] \begin{pmatrix} \1 & 0 \\ P_1 & I\end{pmatrix}\begin{pmatrix} \1 & P_2 \\ 0 & I\end{pmatrix}\, ,\\
[L_1,L_2]\ct = \begin{pmatrix} \1 & -P_2 \\ 0 & I\end{pmatrix}\begin{pmatrix} \1 & 0 \\ -P_1 & I\end{pmatrix}[X_1,X_2]\ct,
\end{align}
and  $A'_1 = A_1 + X_1\ct E X_1  - X_2\ct E X_1 P_1$ and $A'_2 = A_2 + X_2\ct E X_2  -P_1 X_1\ct E X_2$. Equivalently,
we have
\begin{equation}
A+E = R_1 A'_1 L_1\ct + R_2 A'_2 L_2\ct.
\end{equation}
\end{theorem}

\begin{proof}
From the first property in (\ref{eq:pert_props}), and theorem 2.8 in ref.\ \cite{StewartSun} we conclude the existence of a matrix $P_1$ s.t.
\begin{equation}
\norm{P_1} \leq \frac{\norm{X_2\ct E X_1}}{\mathrm{sep}(A_1,A_2) - \norm{X_2\ct E X_2} - \norm{X_1\ct E X_1}},\\
\end{equation}
and
\begin{equation}
\begin{pmatrix} \1 & 0 \\ -P_1 & I\end{pmatrix}[X_1,X_2]\ct(A+E)[X_1,X_2] \begin{pmatrix} 1\ & 0 \\ P_1 & I\end{pmatrix} = \begin{pmatrix} A'_1 & E_{12} \\ 0 & A'_2\end{pmatrix}
\end{equation}
with $E_{12} = X_1\ct EX_2$ and $A'_1 = A_1 + X_1\ct E X_1  - X_2\ct E X_1 P_1$ and $A'_2 = A_2 + X_2\ct E X_2  -P_1 X_1\ct E X_2$.  Now considering the above as a perturbation of $A' = \begin{pmatrix} A'_1 & 0 \\0  & A'_2\end{pmatrix}$ we can apply theorem 2.8 from ref.\ \cite{StewartSun} again so long as
\begin{equation}
\mathrm{sep}(A'_2,A'_1) >0 .
\end{equation}
Using the stability and symmetry of the $\mathrm{sep}$ function a necessary condition for the above is
\begin{equation}
\mathrm{sep}(A_1,A_2) - \norm{X_1\ct EX_1 + X_1\ct EX_2P} - \norm{X_1\ct EX_1 -PX_1\ct E X_2}>0
\end{equation}
which, by sub-multiplicativity and the norm bound on $P_1$ is true if the second property in eq.\  (\ref{eq:pert_props})
holds. Hence theorem 2.8 in ref.\ \cite{StewartSun} provides for the existence of a $P_2$ with norm bound
\begin{align}
\norm{P_2} &\leq \frac{\norm{X_2\ct E X_1}}{\mathrm{sep}(A_1,A_2) - \norm{X_2\ct E X_2 + X_1\ct EX_2P_1} - \norm{X_2\ct E X_2- P_1X_1\ct EX_2}}
\end{align}
and the property that
\begin{equation}
[L_1, L_2]\ct (A+E)[R_1,R_2] = \begin{pmatrix} A'_1 & 0 \\ 0 & A'_2\end{pmatrix}
\end{equation}
with
\begin{align}
[R_1,R_2] &= [X_1,X_2] \begin{pmatrix} \1 & 0 \\ P_1 & I\end{pmatrix}\begin{pmatrix} \1 & P_2 \\ 0 & I\end{pmatrix},\\
[L_1,L_2]\ct &= \begin{pmatrix} \1 & -P_2 \\ 0 & I\end{pmatrix}\begin{pmatrix} \1 & 0 \\ -P_1 & I\end{pmatrix}[X_1,X_2]\ct.
\end{align}
\end{proof}
We note that in eq.\ (\ref{eq:pert_props}) the first property implies the second if $\norm{X_1\ct EX_2}\leq  \norm{X_2\ct EX_1} $.

While eigenvalues and invariant subspaces are stable under small perturbations (as discussed above), that is, they are holomorphic functions with respect to analytic perturbations, the same is not true for eigenvectors. This is due to the fact that a vector basis spanning a multi-dimensional eigenspace is not uniquely determined, and thus the eigenvectors of the perturbed eigenspace may be completely different from the unperturbed basis. However, if an unperturbed eigenvalue $a_1$ is simple, the related eigenvector $x_1$ is unique (up to a scalar factor), and it is thus stable. We can make this more explicit by specializing theorem \ref{thm:subspace_pert} to simple invariant subspaces of dimension one.
Let us again consider an hermitian matrix $A$ and adopt a unitary basis transformation $X=[x_1,X_2]$ so that
\begin{equation}
	[x_1,X_2]^\dagger A [x_1,X_2]
	=
	\begin{pmatrix}
	a_1 & 0 \\
	0 & A_2
	\end{pmatrix} ,
\end{equation}
where $A_2 \in \mc{M}_{d-1}$.
In this specific setting, the separation function becomes~\citep{Stewartselected}
\begin{equation}\label{eq:scalar_sep}
\mathrm{sep}(a_1,A_2) = \norm{\left(a_1 \1-A_2\right)^{-1}}^{-1}.
\end{equation}
From Theorem~\ref{thm:subspace_pert}, we then have \newtext{the following}.
\begin{corollary}[Perturbation of a 1-dim simple subspace]\label{thm:eigenvector_perturbation}
	The left and right perturbed eigenvectors originated from $x_1$ are
	\begin{equation}
	r_1 \approx x_1 +  X_2\, \left(a_1 \1-A_2\right)^{-1} \, X_2^\dagger E x_1 \qquad \text{and} \qquad
	\ell_1^\dagger \approx x_1^\dagger+ x_1^\dagger E X_2\, \left(a_1 \1-A_2\right)^{-1} \, X_2^\dagger ,
	\end{equation}
	where we neglect terms $O(\norm{E}^2)$.
\end{corollary}

Finally, to analyze perturbations of eigenvalues, we will make use of the Bauer-Fike Theorem~\cite[theorem 1.6]{StewartSun}: let $A$ be diagonalizable such that $S^{-1}AS= \mathrm{diag}  ((a_j)_j)$ and let $E$ be an arbitrary operator of same dimension. Then, for any eigenvalue $\tilde a$ of $A+E$, the bound
\begin{equation}
\abs{\tilde a - a_j} \leq \normn{S}\normn{S^{-1}}\normn{E}
\end{equation}
is satisfied for some eigenvalue $a_j$ in any vector-induced norm. This implies that, if $A$ is Hermitian, then
\begin{equation}\label{bound:eigval_perturbation}
\abs{\tilde a - a_j} \leq\normn{E}.
\end{equation}
%----------------------------------------------------------------------
\section{The randomized benchmarking protocol}\label{sec:rand_bench}

The name \acf{RB} is conventionally given to a class of methods that assess the quality of a set of quantum gates. These methods are probabilistic, and can be seen as constructing an estimator for a quantity that captures some notion of gate quality. In this section, we will make an attempt at \emph{defining} \acl{RB}. By this we mean that we will attempt to organize and make explicit various ideas that have been present in the literature. We begin (in section \ref{subsec:structure}) by dividing \ac{RB} into two parts: a data collection phase and a data processing phase. These correspond roughly to the parts of \ac{RB} performed on a quantum computer and on a classical computer, respectively. Within this division we focus first on the data collection phase. In sections \ref{subsec:RB_input} and \ref{subsec:rb_protocol}, we give a general protocol for the data collection phase of \ac{RB}. This general protocol depends on a number of input parameters, and we can obtain every known \ac{RB} protocol from a choice of these input parameters. We complement this protocol with a classification of \ac{RB} protocols into a few types in section \ref{subsec:rb_typology}. This classification, which pertains only to the data collection phase of \ac{RB} is largely a formalization of knowledge implicit in the literature but we will see that it is a useful organizing tool when proving theorems about the data generated by \ac{RB}. This data we discuss in section \ref{subsec:output_data}.

We note that the output of \ac{RB} data is assumed to be of a very particular form, namely that of a
linear combination of (matrix) exponential decays. However, this form is incumbent upon assumptions on the quantum computer on which (the data collection phase of ) \ac{RB} is implemented. We discuss what assumptions have been made before in the literature and propose our own set of assumptions, which we justify later in the text.

\subsection{The data collection and data processing phases}\label{subsec:structure}

 {\ac{RB} is composed of two major parts}, a {\it data collection phase} and a {\it data processing phase}. The data collection phase consists of what one typically thinks of as \ac{RB}: one randomly selects a sequence of quantum gates and applies them to a quantum state together with a global inverse, and measures the resulting state. Averaging over many random choices of these gates one obtains \ac{RB} output data that depends on the length of the random sequence in a controlled way. This vague description can be made more precise in many different ways and we will provide a general framework for this procedure in the next few subsections.

 The data processing phase on the other hand consists of what one then does with the data given by a \ac{RB} experiment. This can be as simple as fitting the data to an exponential decay, but in many cases also involves more sophisticated processing techniques. The key feature of the \ac{RB} protocol that allows for a structured approach to data processing is the fact that the \ac{RB} output data has a very controlled form. We will discuss this form in section\ \ref{subsec:output_data} after more formally discussing the data collection phase of \ac{RB}.

\subsection{Input parameters}\label{subsec:RB_input}
The data collection phase of a \ac{RB} procedure is characterized by a set of input parameters. These input parameters fully define a protocol (which we write down in section \ref{subsec:rb_protocol}) that can be executed on a quantum computer, yielding probabilistic data that can then be interpreted.  Below is a list of all input parameters to \ac{RB}, together with an explanation and examples of choices for these parameters that correspond to versions of \ac{RB} present in the literature.

\begin{enumerate}
	\item {\bf A gate set/group:} A finite set of unitaries (quantum gates) on $\mathbb{C}^d$. In (almost) all \ac{RB} protocols this gate set is also a finite subgroup $\gr\subset U(d)$ of the unitary group. In a large section of the \ac{RB} literature the group considered is the $q$-qubit Clifford group $\mathbb{C}_q$, but a range of other choices (such as the Pauli group $\mathbb{P}_q$ \citep{CycleBenchmarking}, the real Clifford group \citep{hashagen2018real} or the $\mathrm{CNOT}$-dihedral group \citep{carignan2015characterizing,cross2016scalable}) are possible. Choosing a group fixes what gates \ac{RB} assesses the quality of and partially determines the structure of the output data. {In generator-style \ac{RB}~\cite{proctor2019direct,francca2018approximate} this group is defined implicitly by the set of generators.}
	\item {\bf A reference implementation/representation:} A map $\phi_{r}$ from the gate set/group $\gr$ to the $d$-dimensional super-operators that specifies how the gates in $\gr$ should be implemented in the quantum computer. This map takes into account aspects of the specific \ac{RB} protocol but also how gates are composed of elementary gates and other implementation details. In uniform \ac{RB} the map $\phi_r$ is a representation of the group $\gr$ on $\mc{S}_{d}$. {The prototypical example is the action on
the space of Hermitian matrices $\rho$ by conjugation, i.e., $\phi_r(g) (\rho) = \omega(g)(\rho) = U_g\rho U_g\ct$.} In general, however, the reference implementation $\phi_r$ is not a representation, though we will see that for any known \ac{RB} procedure the reference implementation can be written as $\phi_r(g) = \mc{A}\omega(g)\mc{B}$, where $\mc{A},\mc{B}$ are (unitary) quantum channels. We will refer to $\omega$ as \emph{the reference representation}.
	
	\item {\bf An ending gate:} A group element $\gend$ that dictates the global action of an \ac{RB} sequence. For most proposals this gate is simply the identity, but in other proposals non-trivial choices for $\gend$ (such as choosing it uniformly at random~\citep{helsen2019new,carignan2015characterizing,CycleBenchmarking,flammia2019efficient}) play an essential role in data-processing schemes. This ending gate also allows us to include \ac{RB} schemes that do not involve an inversion gate \citep{dirkse2019efficient,wallman2015estimating,WallmanEtAl:2016:Leakage,wood2018quantification}. { We emphasize that it is not necessary to implement this gate physically, but rather it arises from compilation.}
	\item {\bf A set of sequence lengths:}  A set of integers $\md{M}$ denoting the length of the random sequences of gates implemented in a \ac{RB} experiment. We will denote elements of this set by $m$ and the largest element of this set by $M$.
	\item {\bf An input state:} A state $\rho_0$ that is prepared at the beginning of an \ac{RB} experiment. This state will typically be a pure state (such as the $\ket{0,\ldots, 0}$ state vector), but is chosen mixed in some versions of \ac{RB}~\citep{helsen2019multiqubit}.
	\item {\bf An output POVM:} A POVM that is measured at the end of an \ac{RB} experiment. We will denote this POVM as $\{
	\Pi_i\}_{i\in I}$ with some index set $I$. In many cases this is a two-component POVM , but some \ac{RB} procedures explicitly call for more complex measurements (such as a computational basis measurement \citep{arute2019quantum}).
	\item {\bf A set of sampling distributions:} A set of probability distributions $\nu_i$ for $i \in \{1,\ldots , M\}$ over the group $\gr$ that govern the random sampling of group elements in \ac{RB}. We will often consider the scenario where all these probability distributions are the same, in which case we will drop the subscript $i$ and just write $\nu$ for \emph{the} probability distribution . Moreover, in almost all instances in the literature this distribution is uniform, i.e., $\nu(g) = 1/ |\gr|$, and unless stated explicitly we will always assume this to be the case.
\end{enumerate}

\subsection{The data collection protocol}\label{subsec:rb_protocol}
Given the input parameters discussed above we can write down a formal procedure for the data collection phase of \ac{RB}. It has as output an estimator $\hat{p}(i,m)$ of a probability $p(i,m)$ for each POVM element $\Pi_i$ for $i\in I$ and  each sequence length $m\in \md{M}$.

\begin{center}
{\setlength{\fboxsep}{5pt}

\framebox{
\begin{minipage}[t]{0.9\columnwidth}
\centering
\begin{algorithm}[H]\label{alg:dataCollectionPhase}
\SetAlgoLined
\For{$m\in \md{M}$}{
	Prepare the initial state $\rho_0$\\
	\For{$i \in \{1,\ldots, m\}$}{
 	Choose $g_i$ at random from $\gr$ according to the measure $\nu_i$\\
 	Apply $\phi_r(g_i)$ to the state\\
 	}
 Compute the global inverse $g_{\mathrm{inv}} = (g_m\ldots g_1)^{-1}$\\
 Apply $\phi_r(\gend g_{\mathrm{inv}})$ to the state\\
 Measure the state in the POVM $\{\Pi_i\}_{i\in I}$\\
 Repeat the above many times to obtain estimators $\hat{p}(i ; g_1, \ldots , g_m)$ for the probabilities
 \begin{equation*}
 p(i ; g_1, \ldots , g_m) = \tr \big( \Pi_i \phi_r(\gend g_{\mathrm{inv}})\phi_r(g_m)\cdots \phi_r(g_1)(\rho_0)\big)
 \end{equation*}\\
 Repeat for many random $g_1, \ldots, g_m$ and average to obtain estimators $\hat{p}(i,m)$ for the probabilities
 \begin{equation*}
 p(i , m) = \sum_{g_1,\ldots, g_m\in \gr}\nu_1(g_1)\ldots \nu_m(g_m) p(i ; g_1, \ldots , g_m)
 \end{equation*}
 }
 Output the estimators $\hat{p}(i,m)$ for all $i\in I, m\in \md{M}$\\
 \caption{\Acl{RB} (data collection phase)}\label{prot:rand_bench}
\end{algorithm}
\end{minipage}
}}
\end{center}

Note that the probabilities $p(i,m)$ depend in a non-trivial manner on the initial state $\rho_0$, the POVM $\{\Pi_i\}_{i\in I}$ and the ending gate $\gend$. We will, however,
suppress this dependence unless it is explicitly necessary to refer to it.

\subsection{A typology of \acl{RB} protocols}\label{subsec:rb_typology}

Given protocol\ alg.~\ref{prot:rand_bench}, different choices of the parameters discussed in section\ \ref{subsec:RB_input} give rise to different \ac{RB} procedures. More strongly, (the data collection phases of) all variants of \ac{RB} currently in the literature can be expressed by choosing these input parameters correctly. Surveying the literature we can distinguish $3$ major types that are differentiated by their reference implementations and sampling distributions. The output data associated with these classes of protocols has varying behaviour and we will treat each class separately in section\ \ref{sec:output_data}. All protocols included in these classes can be found in fig.\ \ref{fig:results overview} (here we will only give illustrative examples).
\begin{enumerate}
	\item {\bf Uniform \acl{RB}: } This is the basic type of \ac{RB}. It is characterized by the fact that the probability distributions $\nu_i$ are the uniform distribution for all $i\in \{1, \ldots, m_\mathrm{max}\}$, and that the reference implementation map $\phi_r$ is exactly a representation $\omega$, usually the standard action by conjugation given by $\omega(g) = \phi_r(g)(\rho) = U_g\rho U_g\ct$ for unitaries $U_g$ (other choices have been made in~\citep{wood2018quantification,wallman2016robust}). Randomized benchmarking proposals of this type are mainly distinguished by what group $\gr$ they consider as a gate-set (at least when it comes to the data collection phase, different proposals in this class might have radically different data processing procedures.) Protocols of this type include the original \ac{RB} proposals \citep{FirstRB,magesan2011gate} and many others.
	\item {\bf Non-uniform \acl{RB}: } The defining feature of this class is that the sampling distributions $\nu_i$ are not the uniform distribution. It comes in two flavours, which we will discuss separately:
	\begin{enumerate}
	 	\item {\bf Subset \ac{RB}:} Here, the distributions $\nu_i$ are far from uniform (and typically only have support on a small subset of the group $\gr$). Examples from the literature are refs.\ \citep{proctor2019direct,francca2018approximate,helsen2019new,ryan2009randomized}.
	 	\item {\bf Approximate \ac{RB}:} Here the $\nu_i$ are close to uniform. This latter class will turn out to be essentially the same as uniform \ac{RB}. This class has been discussed  in ref.\ \cite{francca2018approximate} and also arises in the original `NIST' \ac{RB} proposal~\citep{knill2008randomized} (as per the analysis of ref.\ \cite{boone2019randomized}).
	\end{enumerate}
	In all works of this type so far the reference implementations are representations (akin to uniform \ac{RB}).
	\item {\bf Interleaved \acl{RB}: }
	This class of \ac{RB} protocols is characterized by the addition of an extra `interleaving gate' in the \ac{RB} procedure. This is a class that is somewhat idiosyncratic, having one standard subtype and a collection of `non-standard'
	protocols:
	\begin{enumerate}
	\item {\bf Standard interleaved \acl{RB}:} In this class the interleaving gate is an element of the benchmarked group $\gr$. In this case we find that it is most useful to interpret interleaved \ac{RB} as uniform \ac{RB}, with the reference implementation a representation $\omega$, but with the probability distributions $\nu_i$ uniform for even $i$ and peaked on a single group element (the interleaving gate) for odd $i$. We will consider this in more detail in section\ \ref{subsec:interleaved}. The paradigmatic example is ref.\ \citep{magesan2012efficient}, but nearly all uniform \ac{RB} protocols have an interleaved version.
	\item {\bf Non-standard interleaved \acl{RB}:} These protocols are characterized by the addition of interleaving gates that are not part of the group $\gr$ as well as non-uniform sampling distributions. We will discuss these protocols on a more case by case basis in section\ \ref{subsec:interleaved}.
	\end{enumerate}
\end{enumerate}

\subsubsection{Protocols without inversion gates}
A number of \ac{RB} protocols have been developed that do not feature an inversion gate $g_{\mathrm{inv}}$. These protocols are indicated with a $**$ in fig.\ \ref{fig:results overview}. While not immediately obvious, these protocols are actually covered by the general procedure written down in alg.~\ref{prot:rand_bench}. We can think of these protocols as choosing the ending gate $\gend$ at random for each experimental run and averaging over the results. Because of the invariance of uniform group averages this is equivalent to not including an inversion gate and ending the protocol on a random group element. In section \ref{sec:new_crb} we will see that protocols without inversion gate can be seen as a special case of a general post-processing scheme for \ac{RB} data.

\subsection{Output data}\label{subsec:output_data}
There is a folkloric notion that the output data of \ac{RB} has an exponential dependence on the sequence length, with the rate of decay dependent only on the implementation $\phi$ of the gates in $\gr$. This was first established to be true for uniform \ac{RB} (in our typology) with the unitary and Clifford groups where, under certain assumptions (see section \ref{subsec:assumptions}) on the quantum computer implementing operations, one can prove that $p(i,m) = Af^m + B$ where $f$ only depends on the implementation map $\phi$ and $A,B$ are constants depending on SPAM. However, if the group $\gr$ was not the Clifford group it was found that the \ac{RB} output data did not follow a single exponential decay but rather was of the form $p(i,m) = \sum_\lambda A_\lambda f_\lambda^m$ with the decay constants $f_\lambda$ depending only on the implementation of the quantum operations and associated with the irreducible sub-representations of the reference representation $\omega$.

However, this functional form is only valid if the reference representation $\omega$ has no multiplicities (no irreducible sub-representation occurs more than once), and hence does not describe all possible \ac{RB} experiments. In this paper we will argue that for a general reference representation of the form $\omega = \bigoplus_{\lambda\in \Lambda}\sigma_\lambda^{\oplus n_\lambda}$ for $\Lambda \subset \Irr$ \ac{RB} data takes the form
\begin{equation}\label{eq:rb_data_decay}
p(i,m) \approx \sum_{\lambda\in \Lambda} \tr(A_{\lambda} M^m_\lambda)
\end{equation}
where $M_\lambda$ is an $n_\lambda\times n_\lambda$ real matrix that only depends on the implementation $\phi$ and $A_{\lambda}$ is an $n_\lambda\times n_\lambda$ matrix encoding SPAM behaviour. Note that the matrices $M_\lambda$ are not required to be normal, or even diagonalizable. This means that $p(i,m)$ can appear to be strikingly non-exponential (at least if $m$ is fairly small) unless $\omega$ is known to be multiplicity-free. We will discuss this in greater detail in section\ \ref{sec:data_processing} when we discuss general fitting procedures.

\subsection{Assumptions}\label{subsec:assumptions}
The functional form of \ac{RB} output data given in eq.\  (\ref{eq:rb_data_decay}) does not immediately follow from the specification of the protocol in alg.\ \ref{prot:rand_bench}. Rather it must be derived based on assumptions on the behaviour of the operations being performed inside the quantum computer. Here we give a run-down of assumptions that are made throughout the literature, and which we will make in order to derive eq.~\eqref{eq:rb_data_decay}. The assumptions we will make are not the most general possible that still lead to eq.\ (\ref{eq:rb_data_decay}), but we attempted to strike a balance between generality and operational motivation. In the list we will point out where assumptions can be generalized and refer to work where this is done (for some versions of \ac{RB}).

\begin{itemize}
	\item {\bf State preparation and measurement consistency: } We assume that the initial state $\rho_0$ and the measurement POVM $\{\Pi_i\}_{i\in I}$ are always prepared in the same manner, independently of the gates being implemented. Slightly stronger, we will assume the existence of quantum channels $\mc{E}_{\mathrm{SP}}$ and $\mc{E}_{\mathrm{M}}$ such that the implemented initial state is given by $\mc{E}_{\mathrm{SP}}(\rho_0)$ and the elements of the implemented measurement POVM are given by $\mc{E}_{\mathrm{M}}(\Pi_i)$. This assumption is made throughout the \ac{RB} literature.

	\item {\bf Markovianity and time-independence:} We assume that the implementation of a gate $g\in \gr$ is always the same, independently of when it is performed in the \ac{RB} protocol and independently of its context (the gates being performed before and after). This assumption leads to the concept of an implementation map $\phi:\gr \to \mc{S}_d$ which assigns to each group element $g$ a completely positive super-operator $\phi(g)$ modelling the actual implementation of the gate.

	\begin{itemize}

		\item This assumption is not always justified, as the implementation of a gate can in principle depend on, e.g., the gates being implemented before it or the amount of time elapsed in the protocol. It can also depend on external uncontrolled variables (either deterministic or random). In  ref.\  \cite{wallman2014randomized},
		a model of time dependence has been considered and in refs.\  \cite{epstein2014investigating,fong2017randomized,fogarty2015nonexponential} the effect of gate-correlations and certain uncontrolled variables such as quasi-static noise were investigated. In all of these scenarios, however,
		the exponential behaviour of eq.\ (\ref{eq:rb_data_decay}) breaks down. {It might be  possible to derive assumptions beyond the setting of Markovian time-independence that lead to output data of the correct form, but we will not pursue this here. }
	\end{itemize}

	\item {\bf Closeness to reference implementation:} In order to derive eq.\ (\ref{eq:rb_data_decay})
	we must make additional assumptions on the implementation map $\phi$. We will assume that
	\begin{equation}
	\frac{1}{|\gr|}\sum_{g\in \gr} \dnorm{\phi_r(g) - \phi(g)} \leq \delta,
	\end{equation}	
	for sufficiently small $\delta>0$. The appearance of the diamond distance might strike one as overly pessimistic, however, we will show that it is in fact required in section\ \ref{sec:fidelity_interpretation}. It is also not the most general possible assumption that still guarantees eq.\ (\ref{eq:rb_data_decay}) (see below), but it has the advantage of making reference only to physical quantities and being operationally interpretable.
	\begin{itemize}

		\item In early works on \ac{RB} the standard assumption was that of gate-independent noise. This means the implementation map $\phi$ is of the form $\phi(g) = \mc{A}\phi_r(g)$ for all $g$ with some fixed quantum channel $\mc{A}$. This is not a very realistic assumption and several attempts were made to replace it with a weaker assumption. In  ref.\  \cite{magesan2011gate} it has been proposed to consider a perturbation $\phi(g) = \mc{A}\phi_r(g) + \mc{A}_g\phi_r(g) $. In  ref.\  \cite{Proctor17}, however,
		this analysis was shown to not be strong enough to actually justify behaviour of the form eq.\ (\ref{eq:rb_data_decay}). {Here, an analysis of uniform Clifford randomized benchmarking as a power iteration of a matrix was proposed (see also early work in this direction by \cite{chasseur2015complete}), justifying the exponential decay model (but with non-optimal correction).
		Subsequently,} in ref.\  \cite{wallman2018randomized} eq.\ (\ref{eq:rb_data_decay})
		was derived (with an exponentially small correction) for uniform \ac{RB} with the multi-qubit Clifford group under the assumption that there exist super-operators $\mc{R},\mc{L}$ such that
	\begin{equation}
	\frac{1}{|\gr|} \sum_{g\in \gr} \dnorm{\phi(g) - \mc{R}\omega(g)\mc{L}} \leq \delta
	\end{equation}
	for small enough $\delta$. This assumption is quite general, but has as its main drawback that the operators $\mc{R},\mc{L}$ are not guaranteed to be completely positive, complicating the interpretation of this assumption as being a belief on physical quantities.
	{Finally,  ref.\  \cite{Merkel18} derives (introducing the Fourier analysis also used here) eq.\ (\ref{eq:rb_data_decay}) (up to an exponentially small correction)} for uniform \ac{RB} with the multi-qubit Clifford group under an assumption on the fidelity of the implementation map $\phi$ w.r.t.~its reference implementation,
	\begin{equation}
	\frac{1}{|\gr|} \sum_{g\in \gr} F(\omega(g),\phi(g)) \geq 1- \delta .
	\end{equation}
	This assumption has the advantage of making reference to physical objects only, but suffers from the drawback that $\delta$ must grow inversely proportional to the underlying Hilbert space dimension for the argument in  ref.\ \cite{Merkel18} to hold. We discuss this further in section \ref{sec:RBdiamond}.
	\end{itemize}

\end{itemize}

\section{The randomized benchmarking fitting model}\label{sec:output_data}
In this section, we will prove a general theorem about the behaviour of \ac{RB} output data, i.e.,
the probabilities $p(i,m)$ associated with an \ac{RB} experiment with its input parameters specified as in section\ \ref{subsec:RB_input} and described in protocol\ alg.~\ref{prot:rand_bench}. We will argue that for a broad variety of choices for reference implementations and probability distributions this data is well described by a linear combination of exponential (matrix) decays (as in eq.\ (\ref{eq:rb_data_decay})), as long as the {physical implementation} $\phi$ is close to its ideal version: the reference implementation $\phi_r$. By close we mean that the diamond distance between reference and ideal implementations, averaged over the group, has to be bounded
as
\begin{equation}
\frac{1}{|\gr|}\sum_{g\in \gr} \dnorm{\phi_r(g) - \phi(g)} \leq \delta.
\end{equation}
One can think of the above equation as a relatively weak initial belief one must hold about one's quantum computer (instantiated in $\phi$) before one can trust the outcome of \ac{RB}.

For the rest of the work we will adopt the transfer matrix framework (discussed in section \ref{sec:prelim_one}) for describing the action of super-operators. We also explicitly write implementation noise on the initial state $\rho_0$ and output POVM $\{\Pi_i\}_{i\in I}$ through quantum channels $\mc{E}_{\mathrm{SP}}$ (state preparation) and $\mc{E}_{\mathrm{M}}$ (measurement). This is notationally somewhat clumsy, but it makes explicit one of the assumptions underlying \ac{RB}, namely that SPAM noise is independent of sequence length.

The theorems we present in this section are generalizations of the theorems given  in ref.\ \cite{wallman2018randomized}, encompassing almost all known \ac{RB} procedures, but the techniques used are based on the cleaner conceptual framework of matrix valued Fourier transforms provided by ref.\ \cite{Merkel18}, which we reviewed in section \ref{subsec:matrix_fourier}. The central observation of ref.\ \cite{Merkel18} is that the data collection phase of uniform \ac{RB} can be seen as evaluating an $m$-fold convolutions of the implementation map $\phi$. This observation generalizes beyond uniform \ac{RB} to arbitrary implementation maps, and, in particular, we see that
\begin{equation}
p(i,m) = %\frac{1}{|\gr|^m}\!\!
\sum_{g_1,\ldots, g_m\in \gr}\!\!\!\! \braa{\mc{E}_{\mathrm{M}}(\Pi_i)}\nu(g_1)\ldots \nu_m(g_m) \phi(\gend g_1^{-1}\ldots g_m^{-1})\phi(g_m)\cdots \phi(g_1)\kett{\mc{E}_{\mathrm{SP}}(\rho_0)}
\end{equation}
can be rewritten, using the invariance of the uniform sum over $\gr$ under changes of variables, as
\begin{align}\label{eq:rb_as_conv}
p(i,m)  &= %\frac{1}{|\gr|^m}\!\!
\sum_{g_1,\ldots, g_m\in \gr}\!\!\!\! \braa{\mc{E}_{\mathrm{M}}(\Pi_i)} \phi(\gend g_m^{-1})\nu_m(g_mg_{m-1}^{-1})\phi(g_m g_{m-1}^{-1}) \cdots \nu_1(g_1)\phi(g_1) \kett{\mc{E}_{\mathrm{SP}}(\rho_0)}\\
&= \braa{\mc{E}_{\mathrm{M}}(\Pi_i)} \big(\phi*(\nu_m\phi)* \cdots *(\nu_1\phi)\big)(\gend)\kett{\mc{E}_{\mathrm{SP}}(\rho_0)}
\end{align}
where we have used the definition of convolution of implementation maps given in eq.\ (\ref{eq:convolution}) and where $(\nu_i\phi)(g) = \nu_i(g)\phi(g)$. We will see that often the convolution product map $\phi*(\nu_m\phi)* \cdots *(\nu_1\phi)$ can be written exactly as an $m$-fold convolution $\phi'^{*m}$ (for some $\phi'$ that is not necessarily the same as $\phi$).

We will begin in section\ \ref{subsec:mother_standard} with discussing the case of uniform \ac{RB} (as per the \ac{RB} typology in section\ \ref{subsec:rb_typology}). This is the easiest case, but the results derived there will go a long way in analyzing the other two types (non-uniform and interleaved \ac{RB}).

\subsection{Uniform \acl{RB}}\label{subsec:mother_standard}

Here we discuss the behaviour of \ac{RB} output data given by a uniform \ac{RB} scheme (as defined in section\ \ref{subsec:rb_typology}). We will prove that this data behaves as expected (i.e., a controlled linear combination of exponential decays), as long as the implementation map $\phi$ is close enough to its reference implementation $\phi_r$. As we saw in section\ \ref{subsec:RB_input}, for uniform \ac{RB} protocols this reference implementation is exactly a representation, which we denote by $\omega$. We can always decompose $\omega$ into a direct sum of irreducible representations. We write this as $\omega = \bigoplus_{\lambda\in\Lambda} \sigma_\lambda^{\oplus n_\lambda}$ with $\Lambda$ some index set and $\sigma_\lambda$ irreducible sub-representations appearing with multiplicity $n_\lambda$. As discussed in section\ \ref{sec:rand_bench}, we expect the \ac{RB} output data to be approximately well described by a linear combination of the form
\begin{equation}
p(i,m) \approx \sum_{\lambda\in \Lambda}\tr(A_\lambda M_\lambda^m)
\end{equation}
where $M_\lambda$ is an  $n_\lambda\times n_\lambda$ matrix depending only on the actual implementation $\phi$. In particular $M_\lambda$ is given by the projection of the Fourier mode $\mc{F}(\phi)[\sigma_\lambda]$ onto the subspace associated with its $n_\lambda$ largest (in absolute value) eigenvalues. This is the content of theorem\ \ref{thm:mother}. The essential idea in theorem\ \ref{thm:mother} is the fact that convolutions correspond to matrix multiplication in Fourier space, together with a careful use of the subspace perturbation techniques discussed in section\ \ref{sec:prelim_one}.

\begin{theorem}[Output data of uniform \acl{RB}]\label{thm:mother}
	Let $p(i,m)$ be the outcome probability associated with a uniform \ac{RB} experiment with group $\gr$, initial state $\rho_0$, reference representation $\omega = \bigoplus_{\lambda\in \Lambda} \sigma_\lambda^{\oplus n_\lambda}$, and ending gate $\gend$, for a specific sequence length $m\in \md{M}$ and POVM element $\Pi_i$ in the POVM $\{\Pi_i\}_i$ (as described in protocol\ alg.~\ref{prot:rand_bench}). Let $\phi$ be the implementation map describing the actually implemented operations.
	Moreover, assume that there exists a  $\delta>0$ such that
	\begin{equation}\label{eq:norm_assump}
	\frac{1}{|\gr|}\sum_{g\in \gr} \dnorm{\omega(g) - \phi(g)} \leq \delta \leq 1/9,
	\end{equation}
	The \ac{RB} output probability $p(i,m)$ is well approximated as
	\begin{equation}\label{eq:mother}
	|p(i,m) -\sum_{\lambda \in \Lambda} \tr(A_\lambda M_\lambda^m)| \leq 8\left(\delta\left[1+ \frac{2\delta}{1-5\delta}\right] \right)^{m}\, ,
	\end{equation}
	where $M_\lambda, A_\lambda$ are $n_\lambda\times n_\lambda$ real matrices and $M_\lambda$ only depends on the implementation $\phi$.

\end{theorem}

\begin{proof}
	Note from eq.\ (\ref{eq:rb_as_conv}) with $\nu_i$ the uniform probability distribution for all $i\in \{1, \ldots , m\}$ that
	\begin{equation}\label{eq:rb_as_conv_stand}
	p(i,m) = \braa{\mc{E}_{\mathrm{M}}(\Pi_i)}(\phi*\phi^{*m})(\gend)\kett{\mc{E}_{\mathrm{SP}}(\rho_0)}.
	\end{equation}
	Inserting the Fourier transform of $\phi$, we get
  \begin{align}
	p(i,m) &= \sum_{\lambda \in \Irr} d_{\lambda}\braa{\mc{E}_{\mathrm{M}}(\Pi_i)}\tr_{V_\lambda}(\mc{F}(\phi)^{m+1}[\sigma_{\lambda}]\overline{\sigma_{\lambda}}(\gend^{-1})\otimes \1 )\kett{\mc{E}_{\mathrm{SP}}(\rho_0)} \label{eq:RB_outcome_simple_form}\\
	&=\braa{\mc{E}_{\mathrm{M}}(\Pi_i)} \tr_{V_{\omega_\gr}} \left[\bigg[\frac{1}{|\gr|}\sum_{g\in \gr} \overline{\omega}_\gr(g)\otimes \phi(g)\bigg]^{m+1} (\mc{D}_\gr \overline{\omega}_\gr(\gend^{-1}))\otimes \1\right] \kett{\mc{E}_{\mathrm{SP}}(\rho_0)}
	\end{align}
	where $\omega_\gr(g) = \oplus_{\lambda\in \Irr}\sigma_\lambda(g)$ is the direct sum of all  irreducible representations of $\gr$ and $\mc{D}_\gr = \oplus_{\lambda\in \Irr} d_\lambda \1_{\lambda}$ accounts for the dimensional factor in the inverse Fourier transform. Now we can consider the Fourier operator $F(\phi)$ (as defined in eq.\ (\ref{eq:fourier_operator})) associated with $\phi$ as a perturbation of its ideal version $F(\omega)$. From our discussion of Fourier transforms and Fourier operators we know that $F(\omega) =\frac{1}{|\gr|}\sum_{g\in \gr}\bigoplus_{\lambda\in \Lambda}\overline{\sigma}_\lambda(g)\otimes \omega(g)$ is an orthogonal projection, with rank given by the number of irreducible sub-representations of $\omega$ ($\mathrm{Rk}(F(\omega)) = \sum_{\lambda\in \Lambda} n_\lambda$). Recall also that there is a natural matrix norm $\norm{\cdot}_{\mathrm{m}}$ on the space of Fourier operators and that
	\begin{equation}
	\norm{F(\phi-\omega)}_{\mathrm{m}} = \frac{1}{|\gr|}\sum_{g\in \gr}\dnorm{\tr_{V_{\omega_\gr}}\left[ F(\phi-\omega) \mc{D}_{\gr}\overline{\omega}_\gr(g^{-1})\otimes \1\right] }= \frac{1}{|\gr|}\sum_{g \in \gr} \dnorm{\phi(g)-\omega(g)}.
	\end{equation}
	The plan is now to use the perturbation theorem (theorem\ \ref{thm:subspace_pert}) to split the above into dominant and sub-dominant invariant subspaces. To do this note that $F(\omega)$ is a projector so we trivially get a spectral resolution with $X_1 = F(\omega)$, $X_2 = \1 - F(\omega)$ with $F(\omega)$ acting as the identity on the column and row space of $X_1 $ and as the zero operator on the column and row space of $X_2$. Thinking of $F(\phi-\omega)$ as a perturbation to $F(\omega)$ we need to ensure the conditions in eq.\ (\ref{eq:pert_props}) are satisfied with respect to the norm $\norm{\cdot}_{\mathrm{m}}$.  Using the sub-multiplicativity of this norm and the fact that $\norm{X_1}_{\mathrm{m}}=1$ by construction together with the triangle inequality, we get the following sufficient condition for the applicability of theorem\ \ref{thm:subspace_pert}:
	\begin{equation}
	\frac{\norm{X_1\ct F(\phi-\omega)X_2}_m\norm{X_2\ct F(\phi-\omega)X_1}_m}{(\mathrm{sep}(\1,0)-\norm{X_1\ct F(\phi-\omega)X_1}_m - \norm{X_2\ct F(\phi-\omega)X_2}_m)^2 } \leq \frac{(2\norm{F(\phi -\omega)}_{\mathrm{m}})^2}{(1- 5\norm{F(\phi -\omega)}_{\mathrm{m}})^2 }< \frac{1}{4}
	\end{equation}
	where we also used that $\mathrm{sep}(\1,0)=1$, which is easy to see from the definition of $\mathrm{sep}$ (see section \ref{subsec:pert_theort}). Working out, we see that the above is satisfied if eq.\ (\ref{eq:norm_assump}) is true, which it is by assumption. Hence we can use theorem\ \ref{thm:subspace_pert} to conclude the existence of operators $R = [R_1,R_2],L= [L_1,L_2]$ with $L\ct = R^{-1}$ and $P_1$ such that

	\begin{equation}
	F(\phi) = R_1\big[X_1\ct F(\omega)X_1 + X_1\ct F(\phi-\omega) (X_1+ X_2 P_1)\big]L_1\ct +  R_2\big[X_2\ct F(\omega)X_2 +(X_2\ct-P_1X_1\ct) F(\phi-\omega)X_2 \big]L_2\ct .
	\end{equation}
	Using the fact that $L\ct= R^{-1}$ (and thus that $L_2\ct R_1 = L_1\ct R_2 = 0$) we can now write $p(m, \gend, \Pi)$ as a sum of two terms corresponding to the above spectral resolution:
	\begin{align}\label{eq:mother_inter}
	p(i,m) &= \braa{\mc{E}_{\mathrm{M}}(\Pi_i)}\tr_{V_{\omega_\gr}}(\mc{D}_{\gr}\overline{\omega}_{\gr}(\gend^{-1})\otimes \1) F(\phi)\left[R_1\big[X_1\ct F(\omega)X_1 +  X_1\ct F(\phi-\omega)(X_1+X_2P_1)\big]L_1\ct\right]^m\kett{\mc{E}_{\mathrm{SP}}(\rho_0)} \notag\\&\hspace{2em}+  \braa{\mc{E}_{\mathrm{M}}(\Pi_i)} \tr_{V_{\omega_\gr}}(\mc{D}_{\gr} \overline{\omega}_{\gr}(\gend^{-1})\otimes \1) \left[ R_2\big[X_2\ct F(\omega)X_2 + (X_2\ct-P_1X_1\ct) F(\phi-\omega) X_2\big]^{m+1} L_2\ct \right]\kett{\mc{E}_{\mathrm{SP}}(\rho_0)}.
	\end{align}
	We will consider both of these terms separately. We will deal first with the second term.
	Note that, using the definitions of $R,L$ from theorem \ref{thm:subspace_pert}, we have
\begin{align}
(2)  &\leq \dnorm{\tr_{V_{\omega_{\gr}}}(\mc{D}_{\gr} \overline{\omega}_{\gr}(\gend^{-1})\otimes \1) \left[ R_2\big[X_2\ct F(\omega)X_2 + (X_2\ct-P_1X_1\ct) F(\phi-\omega) X_2\big]^{m+1} L_2\ct \right]}\\
&\leq \norm{\left[ (X_2 + X_1P_2 + X_2P_1P_2)\big[X_2\ct F(\omega)X_2 + (X_2\ct-P_1X_1\ct) F(\phi-\omega) X_2\big]^m (X_2-P_1X_1)\ct \right]}_{\mathrm{max}}
\end{align}
which is just a statement about the max-norm of a Fourier-operator. Note that $X_2\ct F(\omega)X_2 = 0$ by construction so the above only depends on $F(\phi-\omega)$. Now using the max-mean norm inequality in
eq.\ (\ref{eq:norm_id}) several times and the fact that $X_2 = \1 - X_1$, we can upper bound this as
\begin{align}
(2) &\leq \norm{(X_2 + X_1P_2+ X_2P_1P_2) (X_2-P_1X_1)\ct F(\phi-\omega)X_2(X_2-P_1X_1) }_{\mathrm{max}} \norm{(F(\phi-\omega) X_2(X_2-P_1 X_1)\ct)^{m}}_{\mathrm{m}}\\
&\leq 2\big(\norm{F(\phi-\omega)}_{\mathrm{max}}(1+\norm{P_2}_{\mathrm{m}})(1+\norm{P_1}_{\mathrm{m}}) +\norm{P_1}_{\mathrm{m}}^2 \norm{P_2}_{\mathrm{m}}(3+ \norm{P_1}_{\mathrm{m}}) \big)\left(\norm{F(\phi-\omega)}_{\mathrm{m}}(1+\norm{P_1}_{\mathrm{m}}) \right)^{m}.
\end{align}
Now we use from theorem\ \ref{thm:subspace_pert}, the upper bounds on
\begin{align}
\norm{P_1}_{\mathrm{m}}\leq\frac{\norm{X_2\ct F(\phi-\omega)X_1}_{\mathrm{m}}}{1 - \norm{X_1\ct F(\phi-\omega)X_1}_{\mathrm{m}} -\norm{X_1\ct F(\phi-\omega)X_1}_{\mathrm{m}} }\leq \frac{2\norm{F(\phi-\omega)}_{\mathrm{m}}}{1 -5\norm{F(\phi-\omega)}_{\mathrm{m}} }\leq \frac{2\delta}{1-5\delta}
\end{align}
and
\begin{align}
\norm{P_2}_{\mathrm{m}} &\leq \frac{2\norm{F(\phi-\omega)}_{\mathrm{m}}}{1 - 5\norm{F(\phi-\omega)}_{\mathrm{m}} - 2\norm{P_1}_{\mathrm{m}}\norm{X_1\ct F(\phi-\omega) X_2 } }\\& \leq \frac{2\norm{F(\phi-\omega)}_{\mathrm{m}}}{1- 5\norm{F(\phi-\omega)}_{\mathrm{m}} - \frac{8\norm{F(\phi-\omega)}_{\mathrm{m}}^2}{1 -5\norm{F(\phi-\omega)}_{\mathrm{m}} } }\\ &\leq \frac{2\delta(1-5\delta)}{1- 8 \delta^2}\\
&\leq\frac{\delta}{1-  \delta} ,
\end{align}
where we have exploited the assumption $\delta\leq 1/9$ in the last line.
Inserting these bounds into the main expression we get
\begin{align}\label{eq:upper_bound}
(2) &\leq 4\bigg(\!\!\left[1+ \frac{2\delta}{1-5\delta}\right]  \!\! \left[1+\frac{\delta}{1-\delta}\right]  \!    +  \!   \left[\frac{2\delta}{1-5\delta}\right]^2  \! \!   \left[\frac{\delta}{1-\delta^2}\right]\!  \!   \left[3 +\frac{2\delta}{1-5\delta}\right] \!\!  \bigg) \!\! \left(\delta\left[1+ \frac{2\delta}{1-5\delta}\right] \right)^{m} \\
&\leq \frac{115}{16}\left(\delta\left[1+ \frac{2\delta}{1-5\delta}\right] \right)^{m}\\&\leq 8\left(\delta\left[1+ \frac{2\delta}{1-5\delta}\right] \right)^{m}
\end{align}
where we have used that $\norm{F(\phi-\omega)}_{\mathrm{max}}\leq 2$ and $\delta\leq 1/9$.
Next we consider the first term in eq.\ (\ref{eq:mother_inter}). For this term we desire an exact expression. We begin by noting that both $F(\omega)$ and $F(\phi)$ are block diagonal with respect to the decomposition of $\omega_{\gr}$ into irreducible representations.
 This implies that the matrices $R,L$ are block diagonal w.r.t.\ this decomposition as well, and that, moreover, we can take the matrices $P_1,P_2$ to be block diagonal with the blocks labeled by the irreducible sub-representations present in $\omega = \bigoplus_{\lambda \in \Lambda}\sigma_\lambda^{\oplus n_\lambda}$. Writing $P = \oplus_{\lambda\in \Lambda} P^\lambda$, and similarly for other operators we can write the first term of eq.\ (\ref{eq:mother_inter}) as
\begin{align}
(1) &= \sum_{\lambda\in \Lambda} d_{\sigma_\lambda}\braa{\mc{E}_{\mathrm{M}}(\Pi_i)}\tr_{V_{\sigma_\lambda}}(\overline{\sigma}_\lambda(\gend^{-1})\otimes \1)\left[\mc{F}(\phi)[\sigma_\lambda] R_1^\lambda\big[({X_1^\lambda}\ct \mc{F}(\phi)[\sigma_\lambda] (X_1^\lambda +{X_2^\lambda}P_1^\lambda)]^m {L_1^\lambda}\ct\right]\kett{\mc{E}_{\mathrm{SP}}(\rho_0)},
\end{align}
where we have also used that $F(\omega)X_2 = 0$ by construction. To continue further we need to pick a convenient basis to express $X_1^\lambda,R_1^\lambda$.

For this note that we can specify rank $1$ Fourier operators in $L(V_{\omega_{\gr}})\otimes \mc{S}_d$ by specifying pairs of super-operators $\mc{A},\mc{B}$ and looking at Fourier operators of the form $F(\mc{A}\omega\mc{B})$ (It is useful to think of the Fourier operator $F(\omega)$ as a vectorization operation on $\mc{A},\mc{B}$). We can express $X_1 = F(\omega)$ in this way by considering the operators $F(\mc{P}^j_\lambda \omega \mc{P}^{j}_\lambda)$ where $\mc{P}_\lambda^j$ is the (super-operator) projector onto the $j$th copy of the representation $\sigma_\lambda$ in $\omega =\bigoplus_{\lambda\in \Lambda}\sigma_\lambda^{\oplus n_\lambda})$. Note that these operators are rank one orthogonal projectors and moreover that
\begin{equation}
\sum_{j_\lambda=1}^{n_\lambda}F(\mc{P}_\lambda^{j_\lambda} \omega\mc{P}_\lambda^{j_\lambda})  = \sum_{j_\lambda=1}^{n_\lambda} = F(\sigma_\lambda^{\oplus n_\lambda}) = X_1^{\lambda}
\end{equation}
holds true. Now noting that $(X_1^\lambda + {X_2^\lambda}P_1^\lambda) = R^\lambda_1$ is a rank $n_\lambda$ matrix with $X_1^\lambda R^\lambda_1 = X_1^\lambda$, we can similarly find $n_\lambda$ super-operators $\mc{R}_\lambda^{j_\lambda}$ ($j_\lambda\in 1, \ldots ,n_\lambda$) such that
\begin{equation}
R_1^\lambda = \sum_{j_\lambda=1}^{n_\lambda} F(\mc{R}_\lambda^{j_\lambda} \omega\mc{P}_\lambda^{j_\lambda}),
\end{equation}
where $F(\mc{P}_\lambda^{j_\lambda} \omega\mc{R}_\lambda^{j_\lambda})$ is again of rank one (but no longer orthogonal).
Note that $X_1^\lambda R_1^\lambda = X_1^\lambda$ gives rise to the orthogonality property
\begin{equation}
 F(\mc{P}_\lambda^{j_\lambda} \omega\mc{P}_\lambda^{j_\lambda})F(\mc{R}_\lambda^{j_\lambda'} \omega\mc{P}_\lambda^{j_\lambda'}) =  \delta_{j_\lambda,j'_\lambda}F(\mc{P}_\lambda^{j_\lambda} \omega\mc{P}_\lambda^{j_\lambda}) .
\end{equation}

Using these resolutions of $X_1^\lambda, R_1^\lambda$  and the orthogonality property we can express the first term in eq.\ (\ref{eq:mother_inter}) further as
\begin{align}
(1) &= \sum_{\lambda\in \Lambda} \sum_{j_\lambda^1,\ldots,j_\lambda^{2m} =1}^{n_\lambda} d_{\sigma_\lambda}\braa{\mc{E}_{\mathrm{M}}(\Pi_i)}\tr_{V_{\sigma_\lambda}}\bigg[(\overline{\sigma}_\lambda(\gend^{-1})\otimes \1)\mc{F}(\phi)[\sigma_\lambda] R^\lambda_1\\
&\hspace{7em}\times \big[ F(\mc{P}_\lambda^{j_\lambda^1} \omega\mc{P}_\lambda^{j_\lambda^1}) \mc{F}(\phi)[\sigma_\lambda] F(\mc{R}_\lambda^{j_\lambda^2} \omega\mc{P}_\lambda^{j_\lambda^2}) \cdots \mc{F}(\phi)[\sigma_\lambda] F(\mc{R}_\lambda^{j_\lambda^{2m}} \omega\mc{P}_\lambda^{j_\lambda^{2m}})\big]{L_1^\lambda}\ct \bigg]\kett{\mc{E}_{\mathrm{SP}}(\rho_0)} \\
=&\sum_{\lambda\in \Lambda} \sum_{j_\lambda^1,j^{2m}_\lambda =1}^{n_\lambda} d_{\sigma_\lambda}\braa{\mc{E}_{\mathrm{M}}(\Pi_i)}\tr_{V_{\sigma_\lambda}}\left[(\overline{\sigma}_\lambda(\gend^{-1})\otimes \1)\mc{F}(\phi)[\sigma_\lambda] R^\lambda_1 F(\mc{P}_\lambda^{j_\lambda} \omega\mc{P}_\lambda^{j'_\lambda}) {L_1^\lambda}\ct \right]\kett{\mc{E}_{\mathrm{SP}}(\rho_0)} [M_\lambda^m]_{j_\lambda,j_\lambda^{2m}}
\end{align}

with
\begin{equation}
[M_\lambda]_{j_\lambda,j'_\lambda} =\tr \left(F(\mc{P}_\lambda^{j_\lambda}\omega\mc{P}_\lambda^{j_\lambda})  F(\phi)F(\mc{R}_\lambda^{j_\lambda'}\omega\mc{P}_\lambda^{j'_\lambda})\right) =   \tr \left(F(\phi)F(\mc{R}_\lambda^{j_\lambda} \omega\mc{P}_\lambda^{j'_\lambda}) \right),
\end{equation}
by the fact that $ F(\mc{P}_\lambda^{j_\lambda} \omega\mc{P}_\lambda^{j_\lambda}),  F(\mc{R}_\lambda^{j_\lambda} \omega\mc{P}_\lambda^{j_\lambda})$ are of rank one.
 Now writing
	\begin{equation}\label{eq:spamfactor}
	[A_\lambda]_{j_\lambda,j_\lambda} = d_{\sigma_\lambda}\braa{\mc{E}_{\mathrm{M}}(\Pi_i)}\tr_{V_{\sigma_\lambda}}\left[(\overline{\sigma}_\lambda(\gend^{-1})\otimes \1)\mc{F}(\phi)[\sigma_\lambda] R^\lambda_1 F(\mc{P}_\lambda^{j_\lambda} \omega\mc{P}_\lambda^{j_\lambda}) {L_1^\lambda}\ct \right]\kett{\mc{E}_{\mathrm{SP}}(\rho_0)}
	\end{equation}
we can combine the two terms in eq.\  (\ref{eq:mother_inter}) to get
\begin{equation}
\bigg|p(i,m) - \sum_{\lambda\in \Lambda} \tr(A_\lambda M_\lambda^m)\bigg|\leq 8\left(\delta\left[1+ \frac{2\delta}{1-5\delta}\right] \right)^{m}.
\end{equation}
\end{proof}

\subsection{Randomized benchmarking with non-uniform sampling}\label{subsec:mother_non_uniform}

Several works \citep{francca2018approximate,proctor2019direct,knill2008randomized,boone2019randomized,helsen2019new} discuss adaptations of \ac{RB} where the elements of the group are no longer sampled exactly at random, but are instead sampled according to (1) a distribution close to uniform~\citep{francca2018approximate,knill2008randomized,boone2019randomized} (which we call `approximate \ac{RB}' in section \ref{subsec:RB_input}, following ref.~\citep{francca2018approximate}), or (2) a distribution that only has support on a small subset of the group; group generators in the case of ref.\ \cite{francca2018approximate} (see also early work on the Clifford group by \cite{ryan2009randomized}), subgroup cosets in the case of ref.\ \cite{helsen2019new}, and constant depth circuits (layers) in the case of ref.\ \cite{proctor2019direct}. In section \ref{subsec:RB_input}, we called these approaches `subset RB'.

We begin by treating the case of approximate \ac{RB}. This corresponds to performing \ac{RB} as described in protocol alg.~\ref{prot:rand_bench} but instead of sampling group elements from the group $\gr$ uniformly at random one samples group elements according to some prescribed probability distributions $\nu_i:\gr\to [0,1]$ (with $i$ indicating the time at which the gate is applied). In ref.\ \citep{francca2018approximate} it has been  argued that as long as the distributions $\nu_i$ are all close to the uniform distribution in the $l_1$-norm, then the output data of approximate \ac{RB} is close to the output data of exact \ac{RB}.

As a corollary of theorem\ \ref{thm:mother} we obtain a similar result. Our result is somewhat less general than the one given in theorem 17 of ref.\ \citep{francca2018approximate}. In particular, we will  assume that all distributions $\nu_i$ are equal to a fixed distribution $\nu$. In return for this restriction we will be able to make a much stronger statement on the behaviour of the \ac{RB} output data.
Moreover, our approach does not require the gate-independent noise assumption (replacing it with the more general diamond norm assumption of eq.\ (\ref{eq:norm_assump})). We have the following statement.

\begin{theorem}[Randomized benchmarking data with non-uniform sampling]\label{thm:non_uniform}
	Let $\nu$ be a probability distribution on $\gr$ and $p_\nu(i,m)$ be the outcome probability associated with a non-uniform \ac{RB} experiment with implementation map $\phi$ and reference representation $\omega(g) = \bigoplus_{\lambda\in \Lambda}\sigma_\lambda^{\oplus n_\lambda}$.
	Moreover, assume that there exists  $\delta,\delta'>0$ such that
	\begin{align}\label{eq:norm_assump_approx}
	\frac{1}{|\gr|}\sum_{g\in \gr} \dnorm{\omega(g) - \phi(g)} &\leq \delta,\\
	\sum_{g\in \gr} |\nu(g) - \frac{1}{|\gr|}| \leq \delta',
	\end{align}
	with $\delta+ \delta'\leq 1/9$.
	Now $p_\nu(i,m)$ is well approximated as
	\begin{equation}\label{eq:non-uniform decay}
	|p_\nu(i,m) -\sum_{\lambda \in \Lambda} \tr(A_\lambda (M^\nu_\lambda)^m)| \leq 8\left((\delta+\delta')\left[1+ \frac{2(\delta+\delta')}{1-5(\delta+\delta')}\right] \right)^{m}
	\end{equation}
	where $M^\nu_\lambda, A_\lambda$ are $n_\lambda\times n_\lambda$ real matrices, $M^\nu_\lambda$ depends on the implementation $\phi$ and the measure $\nu$.
\end{theorem}
\begin{proof}
Consider the map $\phi_\nu:\gr\to \mc{S}_d: g\to |\gr|\nu(g)\phi(g)$. Note that we can think of non-uniform \ac{RB} as being uniform \ac{RB} with this (not trace preserving but still completely positive) implementation map. In particular we have
\begin{equation}
p_\nu(i,\gend,m) = \braa{\mc{E}_{\mathrm{M}}(\Pi_i)}(\phi*\phi_\nu^{*m})(\gend)\kett{\mc{E}_{\mathrm{SP}}(\rho_0)}
\end{equation}
which is just (\ref{eq:rb_as_conv_stand})
but with the `effective implementation'  $\phi_\nu$.
From the assumptions of the theorem  we have
\begin{equation}
\gsum{g} \dnorm{\omega(g)- \phi_\nu(g)} \leq  \frac{1}{|\gr|}\sum_{g\in \gr} \dnorm{\omega(g) - \phi(g)} + \sum_{g\in \gr} |\nu(g) - \frac{1}{|\gr|}| \leq \frac{1}{9} .
\end{equation}
Hence, the proof of theorem \ref{thm:mother} immediately applies to $p_\nu(i,\gend,m)$, yielding (\ref{eq:non-uniform decay}).

\end{proof}
We note that in the case of NIST \ac{RB} \citep{boone2019randomized,boone2019randomized} the probability distribution over (a subgroup of) the single qubit Clifford group is not strictly speaking close enough to uniform to apply the above theorem. This can be easily solved by blocking a few gate applications together, defining a new effective implementation map $\phi' = (\nu\phi)*(\nu\phi)\cdots *(\nu\phi)$ which is close enough to uniformly distributed to apply theorem \ref{thm:non_uniform}.

The above approach fails utterly when applied to subset \ac{RB}.
In this scenario the distribution $\nu$ only has support on a small subset $A$ of $\gr$ and consequently $\sum_{g\in \gr}|\nu(g) - \frac{1}{|\gr|}| \approx 1$ in many cases.
This is not necessarily a weakness of theorem \ref{thm:mother} but rather a statement of the fact that strong deviations from exponential behaviour can be observed if one does not give the distribution $\nu$ time to converge to the uniform distribution through repeated convolution.
This was already noted more or less explicitly in previous papers on subset \ac{RB}.
There are two approaches to solving this problem. The first, followed in refs.\ \cite{proctor2019direct,francca2018approximate,helsen2019new,ryan2009randomized} is to
restrict the set of sequence lengths $\md{M}$ at which \ac{RB} data is
gathered to $m\geq m_{\mathrm{mix}}$ where $m_{\mathrm{mix}}$ is related to the mixing time of the distribution $\nu$.
Note that in the direct \ac{RB} proposal \citep{proctor2019direct},
this convergence time is instead enforced directly by applying a uniformly random gate before applying non-uniformly sampled gates.
The second approach is to take this deviation from uniform \ac{RB} behaviour at face value~\citep{CycleBenchmarking} and draw conclusions from the \ac{RB} output directly.
We believe this latter approach is more accurately classified as an interleaved benchmarking scheme and we will discuss it there.

With regards to the first approach we can make a statement akin to theorem \ref{thm:mother} about subset \ac{RB} procedures by making the (natural) assumption that upon equilibration of the distribution $\nu$ the quality of the total gates has not degraded too much. Intuitively, this means that the gates that have high weight in the initial distribution are of high enough quality to generate (by composition) good quality implementations of all gates in the group. Concretely,
we have the following theorem.

\begin{theorem}[Subset \acl{RB}]\label{thm:subset_mother}
Let $\nu$ be a probability distribution on $\gr$ and $p_\nu(i,m)$ be the outcome probability associated with a non-uniform \ac{RB} experiment with implementation map $\phi$ and reference representation $\omega(g) = \bigoplus_{\lambda\in \Lambda}\sigma_\lambda^{\oplus n_\lambda}$.
	Moreover, assume that there exists an integer $m_{\mathrm{mix}}$ and real numbers $\delta,\delta'>0$ such that
	\begin{align}\label{eq:norm_assump_subset}
		\sum_{g\in \gr} |\nu^{*m_\mathrm{mix}}(g) - \frac{1}{|\gr|}| \leq \delta',
		\\
	\sum_{g\in \gr} \nu(g)\dnorm{\omega(g) - \phi(g)} &\leq \frac{\delta}{m_{\mathrm{mix}}}
	\end{align}
	with $\delta+ \delta'\leq 1/9$.
	Now $p_\nu(i,m)$ is well approximated as
	\begin{equation}\label{eq:subset_decay}
	|p_\nu(i,m) -\sum_{\lambda \in \Lambda} \tr(A_\lambda M_\lambda^{m-m_{\mathrm{mix}}})| \leq \varepsilon
	\end{equation}
	with $M_\lambda$ the projection onto the $n_\lambda$ dimensional dominant invariant subspace of $\mc{F}(\nu\phi)[\sigma_\lambda]$ and where
  \begin{equation}
  \varepsilon \leq 2\delta''\bigg(\!\!\left[1+ \frac{2\delta''}{1-5\delta''}\right]  \!\! \left[1+\frac{\delta''}{1-\delta''}\right]  \!    +  \!   \left[\frac{2\delta''}{1-5\delta''}\right]^2  \! \!   \left[\frac{\delta''}{1-\delta''}\right]\!  \!   \left[3 +\frac{2\delta''}{1-5\delta''}\right] \!\!  \bigg) \leq 4 \delta''
  \end{equation}
  with $\delta'' = \delta + \delta'$.
\end{theorem}

Note that this theorem is qualitatively less strong than theorem \ref{thm:mother}. In particular, we can not guarantee that the distance between the output data of subset \ac{RB} and the exponential decays associated with the irreducible sub-representations of the reference representation closes exponentially fast with increasing sequence length. However, our bound on this distance is stronger than previous rigorous statements (theorem 20 in ref.\ \cite{francca2018approximate}) and works under weaker assumptions. The distance bound given in ref.\ \cite{helsen2019new} (theorem 3) does close exponentially but the proof relies critically on the fact that $\nu$ is uniformly non-zero on a (large) subgroup coset in $\gr$, and thus only applies to a far more restricted situation. {Note also that it does not directly apply to the approach taken in \cite{proctor2019direct}. However, we believe that with very minor alterations the reasoning below can be made to fit. }

\begin{proof}
Consider again the  map $\phi_\nu:\gr\to \mc{S}_d: g\to |\gr|\nu(g)\phi(g)$. We have
\begin{equation}
p_\nu(i,m) = \braa{\mc{E}_{\mathrm{M}}(\Pi_i)}(\phi*\phi_\nu^{*m})(\gend)\kett{\mc{E}_{\mathrm{SP}}(\rho_0)} .
\end{equation}
We now establish a bound on the quality of $\phi_\nu^{*m_{\mathrm{mix}}}$, namely we show that
\begin{equation}\label{eq:non_uniform_norm}
\gsum{g}\dnorm{\phi_\nu^{*m_{\mathrm{mix}}}(g) - \omega(g)} \leq \delta + \delta'\leq \frac{1}{9}.
\end{equation}
This can be seen as follows
\begin{equation}
\gsum{g}\dnorm{\phi_\nu^{*m_\mathrm{mix}}(g)- \omega(g)} \leq  \gsum{g} \dnorm{\omega_\nu^{*m_\mathrm{mix}}(g)- \omega(g)} + \gsum{g}\dnorm{\phi_\nu^{*m_\mathrm{mix}} - \omega_\nu^{*m_\mathrm{mix}}}
\end{equation}
with $\omega_\nu(g) = |\gr|\nu(g)\omega(g)$. Writing out the convolution in the first term and changing variables,
we get
\begin{align}
\gsum{g} \dnorm{\omega_\nu^{*m_\mathrm{mix}}(g)- \omega(g)}&= \gsum{g} \dnorm{\sum_{g_1,\ldots g_{m_{\mathrm{mix}}-1}\in \gr} |\gr|\nu(gg_{m_{\mathrm{mix}}-1}^{-1})\ldots \nu(g_1) \omega(gg_{m_{\mathrm{mix}}-1}^{-1}) \cdots \omega(g_1)- \omega(g)}\\
&\leq \gsum{g} \bigg|\sum_{g_1,\ldots g_{m_{\mathrm{mix}}-1}}|\gr|\nu(gg_{m_{\mathrm{mix}}-1}^{-1})\ldots \nu(g_1)-1\bigg| \dnorm{\omega(g)}\\
&= \sum_{g\in \gr} \bigg|\frac{1}{|\gr|} - \nu^{*{m_{\mathrm{mix}}}}(g)\bigg|
\end{align}
for the first term and
\begin{align}
\gsum{g}\dnorm{\phi_\nu^{*m_\mathrm{mix}} - \omega_\nu^{*m_\mathrm{mix}}}&= \gsum{g}\dnorm{\sum_{j=1}^{m_{\mathrm{mix}}}\phi_\nu^{*(m_\mathrm{mix}-j)}*(\phi_\nu - \omega_\nu)* \omega_\nu^{*(j-1)}(g) }\\
&\leq m_{\mathrm{mix}}\sum_{g\in \gr} \nu(g) \dnorm{\phi(g)-\omega(g)} ,
\end{align}
where we have used the telescoping series identity $A^{m} - B^{m} = \sum_{j=1}^{m}A^{m-j} (A-B)B^{j-1}$ which holds for any elements $A,B$ of an associative algebra (such as the implementation maps with convolution), the sub-multiplicativity of the diamond norm,  and the fact that $\dnorm{\phi(g)} = \dnorm{\omega(g)} = 1$ for all $g\in \gr$. Together with the theorem assumptions, this yields (\ref{eq:non_uniform_norm}).
Now as in theorem \ref{thm:mother}, we can write the \ac{RB} output data as
\begin{equation}
p_\nu(i,m) = \braa{\mc{E}_{\mathrm{M}}(\Pi_i)}\tr_{V_{\omega_\gr}}\bigg(\mc{D}_{\gr}(\overline{\omega}_{\gr}(\gend^{-1})\otimes \1) F(\phi) F(\phi_\nu)^{m'}F(\phi_\nu^{*m_\mathrm{mix}})\bigg) \kett{\mc{E}_{\mathrm{SP}}(\rho_0)},
\end{equation}
where $m' = m- m_{\mathrm{mix}}$.
We can again consider $F(\phi_\nu^{*m_\mathrm{mix}})$ as a perturbation of $F(\omega)$. Since $F(\omega)$ is a projector, the operator $F(\phi_\nu^{*m_\mathrm{mix}})$ will resolve into a dominant an sub-dominant invariant subspace (as in theorem \ref{thm:mother}). We have
\begin{align}\label{eq:subset_inter}
  p(i,m) &= \braa{\mc{E}_{\mathrm{M}}(\Pi_i)} \tr_{V_{\omega_\gr}}\bigg(\mc{D}_{\gr}(\overline{\omega}_{\gr}(\gend^{-1})\otimes \1) \Big[F(\phi)F(\phi_\nu)^{m-m_{\mathrm{mix}}}R_1\big[X_1\ct F(\omega)X_1\notag \\&\hspace{15em}+ (X_1\ct-P_1X_2\ct) F(\phi_\nu^{*m_\mathrm{mix}}-\omega)X_1\big]L_1\ct \Big]\bigg)\kett{\mc{E}_{\mathrm{SP}}(\rho_0)} \notag\\&\hspace{5em}+  \braa{\mc{E}_{\mathrm{M}}(\Pi_i)} \tr_{V_{\omega_\gr}}\bigg(\mc{D}_{\gr}(\overline{\omega}_{\gr}(\gend^{-1})\otimes \1) \Big[F(\phi)F(\phi_\nu)^{m-m_{\mathrm{mix}}} R_2\big[X_2\ct F(\omega)X_2 \notag\\&\hspace{17em}+ (X_2\ct-P_1X_1\ct) F(\phi_\nu^{*m_\mathrm{mix}}-\omega) X_2\big] L_2\ct\Big]\bigg)\kett{\mc{E}_{\mathrm{SP}}(\rho_0)}.
\end{align}

Now note that $F(\phi_\nu^{*m_\mathrm{mix}})$ and $F(\phi_\nu)$ commute, and hence share invariant subspaces. This means we can write the first term in eq.\
(\ref{eq:subset_inter}) as
\begin{equation}
(1) = \sum_{\lambda\in \Lambda} \tr\big(A_\lambda M_\lambda^{m-m_\mathrm{mix}}).
\end{equation}
Finally, we can bound the second term in eq.\ (\ref{eq:subset_inter}) as
\begin{align}
|\;(2)\;|&\leq \norm{F(\phi)F(\phi_\nu)^{m-m_{\mathrm{mix}}}R_2\big[X_2\ct F(\omega)X_2 + (X_2\ct-P_1X_1\ct) F(\phi_\nu^{*m_\mathrm{mix}}-\omega) X_2\big] L_2\ct  }_{\mathrm{max}}\\
&\leq \norm{R_2\big[X_2\ct F(\omega)X_2 + (X_2\ct-P_1X_1\ct) F(\phi_\nu^{*m_\mathrm{mix}}-\omega) X_2\big] L_2\ct}_{m} \norm{F(\phi_\nu)^{m-m_{\mathrm{mix}} -1} }_{\mathrm{m}}\norm{F(\phi) }_{\mathrm{max}}
\end{align}
using the max-mean inequality of the norms on Fourier operators. Note now that
\begin{equation}
\norm{F(\phi_\nu)^{m-m_{\mathrm{mix}}-1} }_{\mathrm{m}} \leq \left[\sum_{g\in \gr} \nu(g) \dnorm{\phi(g)}\right]^{m-m_{\mathrm{mix}}-1}\leq 1
\end{equation}
where we have used that $\nu$ is a probability distribution and that $\dnorm{\phi}\leq 1$. Moreover, we have that $\norm{F(\phi) }_{\mathrm{max}} \leq 1$.
Using this and the reasoning from theorem \ref{thm:mother} we can thus bound the second term as
\begin{equation}
|\;(2)\;|\leq 2\norm{F(\phi_\nu^{*{m_\mathrm{mix}}} - \omega)}_{m} \bigg(\!\!\left[1+ \frac{2\delta''}{1-5\delta''}\right]  \!\! \left[1+\frac{\delta''}{1-\delta''}\right]  \!    +  \!   \left[\frac{2\delta''}{1-5\delta''}\right]^2  \! \!   \left[\frac{\delta''}{1-\delta''}\right]\!  \!   \left[3 +\frac{2\delta''}{1-5\delta''}\right] \!\!  \bigg)
\end{equation}
with $\delta'' = \delta+ \delta'$. Inserting the assumption that $\norm{F(\phi^{*{m_\mathrm{mix}}} - \omega)}_{m}\leq \delta''$ we obtain the statement of the theorem.
\end{proof}

\subsection{Interleaved \acl{RB}}\label{subsec:interleaved}
As discussed in \ref{subsec:rb_typology}, a common variant of \ac{RB} is interleaved randomized benchmarking (IRB). IRB is performed like uniform \ac{RB}, as
formulated in alg.~\ref{prot:rand_bench}, but the reference implementation is not a representation. Instead a fixed operation $C$ is being interleaved between the application of randomly selected group elements. The outcome of this experiment is then compared to the same \ac{RB} experiment without the interleaving gate to infer the quality of the interleaved gate $C$. The literature splits into two sections, standard interleaved \ac{RB}~\citep{magesan2012efficient,sheldon2016characterizing} and non-standard interleaved \ac{RB}~\citep{harper2017estimating,PhysRevLett.123.060501}. {We emphasize here that we discuss the so-called `interleaved step of the interleaved \ac{RB} protocol, and do not interpret the resulting decay rate (for a thorough discussion of the relationship of interleaved RB decay rates and their interpretation see \cite{carignan2019bounding}).}

\subsubsection{Standard interleaved \acl{RB}} In the standard protocol the interleaved operation $C$ is applied after every randomly selected gate and is also a part of the group $\gr$. Hence at the end of a random sequence, the inversion step can be performed inside the group. An IRB output data is thus of the form
\begin{equation}
p_{\mathrm{IRB}}(i,\gend, m) = \frac{1}{|\gr|^m} \sum_{g_1,\ldots ,g_m \in \gr}\braa{\mc{E}_{\mathrm{M}}(\Pi_i)} \phi(\gend (g_1C\ldots g_m C)^{-1})\phi(C)\phi(g_m)\cdots \phi(C)\phi(g_1)\kett{\mc{E}_{\mathrm{SP}}(\rho_0)}
\end{equation}
for a POVM element $\Pi_i$, an ending gate $\gend$, a sequence length $m$, an implementation map $\phi$ and an initial state $\rho_0$. It is interesting to interpret this procedure in the light of the protocol given in section \ref{subsec:rb_protocol}. Namely we can think of defining a probability distribution $\nu_C$ over $\gr$, that takes the value $1$ for $g = C$ and $0$ for all other group elements. With this probability distribution, we can reconsider the above as an \ac{RB} experiment according to the protocol written in alg.~\ref{prot:rand_bench}, we have
 \begin{align}
p_{\mathrm{IRB}}(i,\gend, m) = p(i,\gend,2m)= \!\!\!\!\sum_{g_1,\ldots ,g_{2m} \in \gr}&\!\!\!\!\braa{\mc{E}_{\mathrm{M}}(\Pi_i)} \phi(\gend (g_1g_2\ldots g_m)^{-1})\nu_C(g_{2m})\phi(g_{2m})\notag \\&\!\!\!\!\!\!\!\times \mu(g_{2m-1})\phi(g_{2m-1})\cdots \nu_C(g_2)\phi(g_2)\mu(g_1)\phi(g_1)\kett{\mc{E}_{\mathrm{SP}}(\rho_0)}
\end{align}
where $\mu$ is the uniform distribution on $\gr$. Hence, we can think of standard IRB as being a \ac{RB} experiment with a particular choice of sampling distributions. In this picture, it becomes trivial to extend theorem \ref{thm:mother} to standard interleaved \ac{RB} by considering the map $\phi_C = (\nu_C\phi)*\phi$.
By the standard change of variables we can see
\begin{equation}
p_{\mathrm{IRB}}(i,\gend, m) =\braa{\mc{E}_{\mathrm{M}}(\Pi_i)}\phi*\phi_C^{*m}(\gend)\kett{\mc{E}_{\mathrm{SP}}(\rho_0)}
\end{equation}
and hence interleaved \ac{RB} is just uniform \ac{RB} with the implementation map $\phi_C$. If $\phi(C)$ is close enough to its reference representation element $\omega(C)$ the assumption eq.\ (\ref{eq:norm_assump}) is reasonable for $\phi_C$ as well. Hence, theorem \ref{thm:mother} holds equally well for interleaved \ac{RB}.

Non-standard interleaved RB protocols \citep{harper2017estimating,PhysRevLett.123.060501,CycleBenchmarking,kimmel2014robust} depart from the above framework by including interleaved gates that are not part of the group $\gr$, (the Pauli group in the case of ref.\ \cite{CycleBenchmarking} and the Clifford group in the case of ref.\ \cite{harper2017estimating}) and sampling from the group in a non-uniform manner.  These are somewhat idiosyncratic so we will treat them separately. We will see that the protocols of refs.\ \cite{harper2017estimating,PhysRevLett.123.060501} are covered by theorem \ref{thm:mother}, while the protocols of ref.\ \cite{CycleBenchmarking} and ref.\
\cite{kimmel2014robust} are not covered. We expect that it is possible to make guarantees on the output data of these protocols with suitable adaptations to theorem \ref{thm:mother} but we do not pursue this here.

\subsubsection{Interleaved T-gate \acl{RB}}
In ref.\ \cite{harper2017estimating} the quality of a $T$-gate (with ideal implementation $\mc{T}$), with an associated noisy implementation $\widetilde{\mc{T}}$ is assessed by estimating the following quantity
\begin{equation}
p_T(m) =\frac{1}{|\mathbb{P}_q|^m |\mathbb{C}_q|^m} \sum_{\substack{p_1,\ldots p_m\in \mathbb{P}_q\\g_1,\ldots g_m\in \mathbb{C}_q}}\braa{\mc{E}_{\mathrm{M}}(\Pi_i)}  \phi((g_m t(p_m) \ldots g_1 t(p_1))^{-1})\phi(g_m)\widetilde{\mc{T}} \phi(p_m)\widetilde{\mc{T}} \cdots \phi(g_1) \widetilde{\mc{T}} \phi(p_1)\widetilde{\mc{T}}\kett{\mc{E}_{\mathrm{SP}}(\rho_0)}
\end{equation}
with $\mathbb{C}_q$ the $q$-qubit Clifford group, $\mathbb{P}_q\subset \mathbb{C}_q$ the Pauli group and $\phi:\mathbb{C}\to \mc{S}_d$ an implementation of the Clifford group (and the Pauli group) and $t(p):\mathbb{P}_q\to \mathbb{C}_q$ is an injective map mapping Pauli elements $p$ to $TpT\ct$. Because $T$ is in the third level of the Clifford hierarchy we have $TpT\ct \in \mathbb{C}_q$ for all $p\in \mathbb{P}_q$ making the above well defined. By defining the map $\phi_T(g):\mathbb{C}_q\to \mc{S}_d:g \mapsto \nu_T(g)\mc{T}\ct \phi(t^{-1}(g))\mc{T}$ with
\begin{equation}
\nu_T(g) = \frac{|\mathbb{C}_q|}{|\mathbb{P}_q|}I(g \in \mathrm{Im}(t))
\end{equation}
a probability distribution on $\mathbb{C}_q$ taking non-zero value only on the image of the map $t$ (strictly speaking $t^{-1}(g)$ is not defined for $g\not \in \mathrm{Im}(t)$, but $\nu_T$ is zero there anyway).
With these definitions we can rewrite the output probability as
\begin{equation}
p_T(m) = \braa{\mc{E}_{\mathrm{M}}(\Pi_i)}  (\phi*(\phi*\phi_T)^{*m})(e)\kett{\mc{E}_{\mathrm{SP}}(\rho_0)} .
\end{equation}
Hence, theorem \ref{thm:mother} generalizes to $p_T(m)$ as long as \eqref{eq:norm_assump} is satisfied for the convoluted map $\phi*\phi_T$. In the ideal case of $\phi=\omega$ (the reference representation) and $\mc{T} = T$ we see that $\phi*\phi_T(g) = \omega(g)$. Hence this is a reasonable assumption to make, and theorem \ref{thm:mother} thus covers the protocol presented  in ref.\
\cite{harper2017estimating}.

\subsubsection{Individual gate benchmarking}
Individual \ac{RB}, as proposed  in ref.\ \cite{PhysRevLett.123.060501}, is an interleaved \ac{RB} protocol characterized by uniform probability distributions and, interestingly, a reference implementation $\phi_r$ that is not a representation. Rather, the reference implementation is of the form $\phi_r(g) = \mc{U} \omega(g)$ where $\omega(g)$ is the standard action by conjugation, i.e. $\omega(g)(\rho)=U_g\rho U_g\ct$, and $\mc{U}(\rho) = U\rho U\ct$ is a fixed unitary gate (that is not a part of the group $\gr$). Moreover, $\mc{U}$ is assumed to commute with the representation $\omega(g)$.
The output \ac{RB} data $p(i,m)$ associated with this procedure is of the form
(\ref{eq:rb_as_conv_stand}),
however, the central assumption  (eq. \ref{eq:norm_assump}) of theorem \ref{thm:mother} is generally far from satisfied (unless $\mc{U}$ is the identity). However,
we can make the alternative assumption that
\begin{equation}\label{eq:ind_gate_assump}
\gsum{g}\dnorm{\mc{U}\omega(g) - \mc{\widetilde{U}} \phi(g)}\leq\delta
\end{equation}
where $\mc{\widetilde{U}}$ is the noisy implementation of the unitary $\mc{U}$ and $\phi$ is the implementation of the reference representation $\omega(g)$. This is a reasonable assumption to make since
\begin{align}
\gsum{g}\dnorm{\mc{U}\omega(g) - \mc{\widetilde{U}} \phi(g)}&\leq \gsum{g}\dnorm{\mc{U}\omega(g) - \mc{U} \phi(g)} + \dnorm{\mc{U}\phi(g) - \mc{\widetilde{U}} \phi(g)}\\
&\leq\dnorm{\mc{U}- \mc{\widetilde{U}}} + \gsum{g}\dnorm{\omega(g) - \phi(g)}
\end{align}
so as long as the implementation of the interleaving unitary $\mc{U}$ is of sufficient quality eq.\ (\ref{eq:ind_gate_assump})
is reasonable. Furthermore we note that due to the commutation assumption $[\omega(g),\mc{U}]=0$ the Fourier operator $F(\mc{U}\omega)$ has the same dominant invariant subspace as $F(\omega)$ (since $F(\mc{U}\omega) = \1_{\gr}\otimes \mc{U} F(\omega) = F(\omega)\1_{\gr}\otimes \mc{U}$). Hence the proof of theorem \ref{thm:mother} goes through for individual gate benchmarking as well, replacing the assumption eq.\ (\ref{eq:norm_assump}) with eq.\  (\ref{eq:ind_gate_assump}).

\subsubsection{Cycle benchmarking}
{

Cycle benchmarking~\cite{CycleBenchmarking} is a recently developed \ac{RB} protocol that can also be subsumed under the framework of Theorem \ref{thm:mother}, albeit after some non-trivial considerations we will discuss in this section.

The data collection phase of cycle benchmarking can be seen as interleaved \ac{RB} over the Pauli group with the interleaving gate $C$ being a (non-Pauli) Clifford gate. In particular, cycle benchmarking implements sequences $C, g_m, \ldots,C, g_1$ where $g$ is drawn uniformly at random from the Pauli group $\mathbb{P}_q$ and $C$ is a Clifford gate.

A key aspect of cycle benchmarking is the cycle length, i.e. an integer $c$ s.t. $C^c = e$ (note that for any Clifford gate such a cycle length exists). In cycle benchmarking the number of random Pauli elements implemented is always a multiple of the cycle length. Writing $\phi(g)$ for the noisy implementation of the standard conjugation representation of the Pauli group, and $\widetilde{\mathcal{C}}$ for the noisy implementation of the Clifford gate $C$ we can define the cycle implementation map (on the Pauli group):
\begin{equation}
\phi_c(g)  = \frac{1}{|\mathbb{P}_q|^{c-1}} \sum_{\substack{g_1,\ldots g_c \in \mathbb{P}_q\\Cg_c \cdots Cg_1 = g}} \widetilde{\mathcal{C}} \phi(g_c) \ldots \widetilde{\mathcal{C}}\phi(g_1).
\end{equation}
Note that because the Clifford group contains the Pauli group the equation $Cg_c \cdots Cg_1 = g$ makes sense. Now because of the cycle property
\begin{equation}
Cg_c \ldots Cg_1 = (C^{-(c-1)}g_c C^{(c-1)}) \ldots C^{-1} g_2 C g_1 =  g'_c \cdots g'_1
\end{equation}
since $C^{-1}g C$ is always a Pauli element. Hence the equation has exactly $|\mathbb{P}_1^{(c-1)}|$ solutions. Furthermore we have that
\begin{equation}
\frac{1}{|\mathbb{P}_q|}\sum_{g\in\mathbb{P}_q} \phi_c(g)
=\frac{1}{|\mathbb{P}_q|^{c}} \sum_{g \in \mathbb{P}_q}\sum_{\substack{g_1,\ldots g_c \in \mathbb{P}_q\\Cg_c \cdots Cg_1 = g}} \widetilde{\mathcal{C}} \phi(g_c) \ldots \widetilde{\mathcal{C}}\phi(g_1)
= \frac{1}{|\mathbb{P}|^{c}} \sum_{g'_1,\ldots g'_c \in \mathbb{P}_q} \widetilde{\mathcal{C}} \phi(g'_c) \cdots \widetilde{\mathcal{C}} \phi(g'_1)
\end{equation}
and thus that
\begin{equation}
\frac{1}{|\mathbb{P}_q|^{mc}}\sum_{g_{1,1},\ldots g_{m,c} \in \mathbb{P}_q}\widetilde{\mathcal{C}} \phi(g_{mc}) \cdots \widetilde{\mathcal{C}} \phi(g_1) = \phi_c^{*m}(e)
\end{equation}
which means cycle benchmarking can be framed as \ac{RB} with the implementation map $\phi_c$. Moreover, since in the limit of perfect gates we have, if $Cg_c \cdots Cg_1 = g$, that
\begin{equation}
\mathcal{C}\omega(g_c) \cdots \mathcal{C}\omega(g_1) = \omega(g)
 \end{equation}
 we can reasonably make the assumption that $\phi_c$ is close to its reference implementation (i.e. \ref{eq:norm_assump}). Hence the behaviour of cycle benchmarking data is covered by Theorem \ref{thm:mother}. What is less clear is how to interpret the resulting exponential decays (especially in terms of the implementations $\phi$ and $\widetilde{\mathcal{C}}$). This requires a more sophisticated analysis, which is done in \cite{CycleBenchmarking}.
}

\subsubsection{Robust benchmarking tomography}
In \emph{robust benchmarking tomography} \citep{kimmel2014robust} one uses a \ac{RB} protocol as a subroutine to extract tomographic information from a super-operator (not necessarily a unitary) $\mc{E}$. This is done by estimating the  probability
\begin{equation}
p(i,m) = \frac{1}{|\gr|^m} \sum_{g_1,\ldots ,g_m \in \gr}\braa{\mc{E}_{\mathrm{M}}(\Pi_i)} \phi( g'(g_1\ldots g_m)^{-1})\mc{E}\phi(g' g_m)\cdots \mc{E}\phi(g' g_1)\kett{\mc{E}_{\mathrm{SP}}(\rho_0)},
\end{equation}
where $g'$ is a fixed element of the group $\gr$ and $\phi$ is the implementation of a reference representation $\omega$ (the goal is to estimate correlations between $\omega(g')$ and $\mc{E}$). We can consider this as an interleaved \ac{RB} scheme with reference implementation $\phi_{\mathrm{tom}}(g) = \omega(g' g)$ (thinking of $\mc{E}$ as a noisy implementation of the identity gate). However, this reference implementation is not close to a representation (unless $g'=e$), which means that theorem \ref{thm:mother} does not apply. This is not an artifact of the proof technique but rather a reflection of the fact that robust benchmarking tomography features extremely rapid exponential decays. In the gate-independent noise case the decay rate is set by the average fidelity $F(\omega(g'),\mc{E})$ which can be very small. In the language of matrix Fourier theory this means that the dominant eigenvalues of the Fourier operator $F(\phi_{\mathrm{tom}})$ will be small even in the ideal case. Hence, we do not expect an assumption of the form\ (\ref{eq:norm_assump})
to be strong enough to guarantee exponential behaviour of the \ac{RB} output data in this scenario.

\section{Data processing and sample complexity}\label{sec:data_processing}

As discussed before the \acf{RB} protocol can be divided into data collection and post-processing
phases.
The data collection protocol is summarized in algorithm~\ref{alg:dataCollectionPhase}.
The outputs of the data collection phase are mean estimators $\hat p(i, m, \gend)$ that
estimate the
average over all sequences of length $m$ according to the measures $\nu_i$ and the quantum measurement statistics, simultaneously.
The main theorems of the data collection phase (theorems \ref{thm:mother} - \ref{thm:subset_mother}) state that the expectation value, again both over
the measurement statistics and the random sequences, is well-approximated by a linear combination of (matrix) exponentials in $m$.

The figures of merit that \ac{RB} experiments report are the decay parameters associated with the
linear combination of (matrix) exponentials.
Extracting these decay parameters is the objective
of the data-processing phase that is the focus of the current section.
For gate-independent noise and reference representations without multiplicities
the decay parameters
can be directly connected to the average gate fidelity of the noise.
In the more general case, the interpretation of the decay parameters in terms of
other operational measures of quality can be more complicated. We will consider the connection between the decay parameters and the average gate-set fidelity in
section~\ref{sec:fidelity_interpretation}.

Here we want to take a more pragmatic approach for the post-processing phase.
The deviation of the decay parameters from unity can directly be regarded as a measure of
quality that captures the deviation of the actually implemented gates from an ideal implementation.
In principle, the set of decay parameters itself provides
a refined image of the quality of the implementation, as compared to the average gate fidelity.
This motivates us to limit the post-processing phase to the extraction of the decay parameters.
The estimation of other measures of quality from the decay parameters is then left to an optional subsequent processing phase.

In the simplest \ac{RB} setting (e.g., uniform \ac{RB} with the Clifford group), featuring a single noise-affected representation, the data processing phase only involves fitting a single exponential decay curve.
The analysis of \ac{RB} data arising in more general settings, however, requires a considerably more flexible approach for the data processing.

Extracting multiple decay coefficients, or poles, from a discrete series of data points is a well-studied problem in signal processing that arises in many different disciplines.
For this reason, this section includes a review of modern approaches to this fitting problem that not only have been generalized to the fitting of matrix exponentials but also come with theoretical performance guarantees and bounds.
The pole-finding algorithms we review (MUSIC and ESPRIT) come with multiple merits:
(1) they are easily and efficiently implementable, (2) they are flexible enough to in principle analyze any \ac{RB} signal of the general form \eqref{eq:rb_data_decay}, (3) they come with in-built de-noising and super-resolution capabilities, (4) they feature theoretical bounds that can (4.a) inform the design of experimental parameters, and (4.b) -very importantly- can be used to identify parameter regimes where distinguishing the different decay parameters becomes infeasible in practice.

Following this review we combine analytical guarantees and numerical simulations to
evaluate the performance of these algorithmic approaches for the processing of \ac{RB} data.
In particular, we discuss the effect of
the configuration of the decay parameters, such as their number and spacings, on the overall number of required measurements and the maximal sequence length in the experiment.
We thereby provide theoretical guiding principles for designing
\ac{RB} experiments and explicitly work out
limitations where the experimental
precision required in order to separate
multiple decays become impractical.

These fundamental limitations in analyzing \ac{RB} data have previously motivated a variety of more resource-intensive data-gathering protocols that take further data from which one can isolate different decay curves in the classical post-processing phase.
We will turn our attention to devising a novel general method for isolating matrix exponentials in section~\ref{sec:new_crb}.
We begin by a detailed description of the data processing problem.

\subsection{The \acl{RB} data processing phase}

The theorems on the data collection phase, morally summarized by eq.\ \eqref{eq:rb_data_decay}, state that in expectation \ac{RB} output data is well-approximated by a linear combination of (matrix) exponentials in $m$.
Every matrix $M_\lambda \in \CC^{n_\lambda \times n_\lambda}$ in the expansion is associated with an irreducible representation $\lambda$ of the reference representation $\omega$ and $n_\lambda$ is the multiplicity of $\sigma_\lambda$ in the decomposition of $\omega$.
From the collected data, a \ac{RB} protocol subsequently
extracts decay parameters that describe the exponential decay.
The decay parameters associated with a matrix $M_\lambda$ are its eigenvalues $\operatorname{spec}(M_\lambda) = \{z^{(\lambda)}_i\}_{i=1}^{n_\lambda}$.
If $M_\lambda$ is diagonalizable, then
\begin{equation}
\Tr(A_\lambda M^m_\lambda) = \sum^{n_\lambda}_i a^{(\lambda)}_i (z^{(\lambda)}_i)^{m}
\end{equation}
with coefficients $a^{(\lambda)}_i$ depending on the overlap of $A_\lambda$ with the eigenspaces.
More generally, let $M_\lambda = S^{-1} J S$ be the Jordan normal decomposition of $M_\lambda$ with Jordan blocks $J = \diag(J_1, J_2, \ldots)$, $J_i \in \mathbb R^{\mu_i \times \mu_i}$ and $\{z^{(\lambda)}_i\}$ being
the corresponding eigenvalues.
For $m \geq \mu_i$, the $j$-th diagonal of the $m$-th power of the $i$-th Jordan block contains the entry $\binom{m}{j} z_i^{m-j}$.
Therefore, the matrix exponential takes the form
\begin{equation}\label{eq:RBdata:signalmodel}
  \Tr(A_\lambda M_{\lambda}^m) = \sum_{i} \sum_{j \in [\mu_i]} a^{(\lambda, j)}_i \binom{m}{j} (z^{(\lambda)}_i)^{m-j}
\end{equation}
with real coefficients $a_i^{(\lambda, j)}$.
Note that $\binom{m}{j}$ are falling polynomials in $m$. Thus, the function space of $\Tr(A_\lambda M_{\lambda}^m)$ is in general spanned by exponential function parametrized by the eigenvalues modulated by falling polynomials.
With the pole-finding techniques, which we discuss in the next section,
one can extract the set of all poles
\begin{equation}
	Q = \bigcup_{\lambda \in \Lambda} \{ z_i^{(\lambda)} | i \in [n_\lambda] \}
\end{equation}
 from \ac{RB} output data.
Thus, the general post-processing task of \ac{RB} is the following: \emph{Given a data-series $\hat p(m)$ that is approximately described by linear combinations of polynomial modulated decays, extract the set $Q$ of all poles}.

Loosely speaking, estimating $Q$ is typically possible, provided that
the coefficients of all representations are sufficiently large and the poles are sufficiently spaced.
In the remainder of this section, we assess this statement quantitatively using analytical and numerical methods.

{In practice, one might operate under additional assumptions and does not need to extract all poles individually.
For example, if one expects multiple poles in the data series that are all more or less aligned, the data processing problem becomes equivalent to extracting a single pole.
The general form of the data-processing task however, stays the same, namely extracting the poles in the data series.}

Without additional assumptions or post-processing,
the resulting poles are \emph{unlabeled}, in the
sense that one does not know which pole is associated with which irreducible presentation.
This issue will be addressed when we turn our attention to techniques that filter the \ac{RB} data
for specific representations in section~\ref{sec:new_crb}.

\subsection{Data processing algorithms and guarantees}
\subsubsection{Fitting single decays}
Many proposals for \ac{RB} derive a data model that is well-approximated by a single decay curve. This is for example the case when the group is a unitary $2$-design, the reference representation $\omega$ is the adjoint representation and the actual implementation is close to being trace-preserving~\cite{wallman2018randomized}.
The adjoint representation of a unitary $2$-group acts irreducible on the space of traceless matrices and yields a single dominant decay curve.

A single dominant decay parameter can be extracted using non-linear least-square fitting algorithms such as Levenberg-Marquardt, see, e.g.,~ref.~\cite[Chapter~3.2]{Kelley:1999:optimization}.
In ref.~\citep{helsen2019multiqubit} it has been shown that in \ac{RB} for the Clifford group the variance of the data points is expected to strongly vary with the sequence length $m$.
This observed heteroskedasticity motivates to use iteratively re-weighted variants of least square fitting algorithms.

Ref.~\citep{Harper:2019:Stats} analyses a simplified fitting procedure that estimates the decay parameter from the ratio of the data for two sufficiently separated sequence lengths.
In the regime of high fidelity, it establishes a multiplicative error in the deviation of the decay parameter from one from an efficient number of samples. Relatedly, ref.\ \citep{flammia2019efficient} gives an estimation scheme for a \ac{RB} procedure that estimates, in parallel, multiple single exponential decays with multiplicative accuracy. This scheme makes use of post-processing techniques to guarantee the `single-exponential' shape of the data. We will discuss this more in section \ref{sec:new_crb}.

\subsubsection{Fitting multiple decay with pole-finding algorithms: MUSIC and ESPRIT}
Algorithms for simultaneously identifying multiple poles (frequencies and decay parameters) from a discrete series of data points date back to at least the work of Prony \cite{Prony:1795}.
A zoo of modern algorithmic approaches has been developed in the context of direction-of-angle estimation in array signaling.
In principle, these techniques can extract poles that are closer together than the grid spacing defined by the finite sampling rate, a phenomenon dubbed \emph{super resolution}.
The theoretical framework to derive guarantees for these algorithms that go beyond a perturbative analysis of special noise models or very simple configurations, was only developed recently \cite{CandesFernandez-Granda:2013, CandesFernandez-Granda:2014}, first focusing on convex optimization.

Here, we will analyze the performance of the MUSIC algorithm \cite{Schmidt:1986} and the ESPRIT \cite{RoyPaulrajKailath:1986} algorithm on \ac{RB} data.
Performance guarantees for these two subspace algorithms were derived in refs.~\cite{LiaoFannjiang:2014:MUSIC,Fannjiang:2016:ESPRIT,LiLiao:2017,LiLiaoFannjiang:2019:Super-resolution} for the multiplicity-free case.
Furthermore, the ESPRIT algorithm was extended to polynomially modulated exponentials of the type we encounter in \ac{RB} data with multiplicities in refs.~\cite{BadeauDavidRichard:2006,BadeauRichardDavid:2008}.
We will summarize the required modification in section~\ref{sec:generalESPRIT}.
For the sake of clarity, we now start reviewing the algorithms for identifying multiple poles without polynomial modulation. This corresponds to the case of \ac{RB} with a multiplicity-free reference representation.
For the rest of this section we will denote the output data as $y_m$ instead of $\hat{p}(m)$, in keeping with the signal processing literature.
We will also assume equidistant spacing of the available sequence lengths $m$.
As we point out in section~\ref{sec:low-rankCompletion}, this requirement can be relaxed by running a low-rank
completion algorithm on incomplete data and thereby infer equidistantly spaced data $y_m$. When clear from the context,
we will write the data series simply as a vector $y$, dropping the explicit dependence on $m$.

The strategy of both algorithms, MUSIC and ESPRIT, is to identify the range of the subspaces associated with the dominant singular values of the Hankel matrix of the data series $\{y_m\}_m$. The crucial observation is that from this subspace the poles can be extracted.
Let $y \in \RR^M$ be the \ac{RB} data with $M$ the maximal sequence length. The \emph{Hankel matrix} for $1 \leq L < M$ is given by
\begin{equation}
  \operatorname{Hankel}_L(y) = %
    \begin{pmatrix}
      y_0 & y_1 & \cdots & y_{M-L} \\
      y_1 & y_2 & \cdots & y_{M-L + 1} \\
      \vdots & \vdots & \ddots & \vdots \\
      y_L & y_{L+1} & \cdots & y_M
    \end{pmatrix}.
\end{equation}
We denote the Vandermonde matrix of size $n \times M$ for poles $z = (z_1, \ldots, z_n)$ by
\begin{equation}
  W_M (z) = W_M(z_1, \ldots, z_n) = \begin{pmatrix}
    1 & z_1 & z_1^2 & \ldots & z_1^{M-1} \\
    1 & z_2 & z_2^2 & \ldots & z_2^{M-1} \\
    \vdots & \vdots & \vdots & \vdots & \vdots \\
    1 & z_n & z_n^2 & \ldots & z_n^{M-1}
  \end{pmatrix}.
\end{equation}
If $n=1$, and thus $z \in \mathbb C$ we will refer to  $W_M(z)$ as the Vandermonde vector of length $M$ and pole $z$.

With this notation, the data vector $y$, without noise, is in the range of $W_M(z)^T$.
Furthermore, cyclically shifting the entries of $y$ amounts to multiplication of the summands with the respective poles.
In effect, the Hankel matrix has a Vandermonde decomposition
\begin{equation}\label{eq:vandermonde_decomposition}
  \operatorname{Hankel}_L(y) = W_L^T(z) \operatorname{diag}(a) W_{M-L}(z) + \operatorname{Hankel}_L(\alpha),
\end{equation}
where we have denoted by $\alpha$ the deviation of $y$ from an ideal linear combination of exponentials due to the
perturbative error $\epsilon(m)$ and finite statistics and where $a$ is the vector of pre-factors given in eq.~(\ref{eq:RBdata:signalmodel}).

To identify the signal subspace and distinguish it from the noise subspace, the MUSIC and ESPRIT algorithms employ an SVD decomposition of the Hankel matrix, $\operatorname{Hankel}_L(y) = U \Sigma V^T$.
In the absence of noise and perturbation, i.e., $\alpha = 0$,
$\operatorname{Hankel}_L(y)$ has $n$ non-vanishing singular values and the corresponding singular vectors form an orthonormal basis of the signal space $\operatorname{span} W^T_M(z)$.
Let $U_\text{signal}$ be the matrix consisting of the singular vectors of the non-trivial singular values as columns and let $U_\text{noise}$ be the matrix consisting of an orthonormal basis of the complement.
In the presence of noise, analogously choosing the singular vectors of the $n$ largest singular values yields an estimate of the signal space.

From the \emph{noise space projector} $P_\text{noise} = U_\text{noise} U^\dagger_\text{noise}$, the MUSIC algorithm defines the inverse noise-space correlation function
$R^{-1}: \mathbb C \to \mathbb R$,
\begin{equation}\label{eq:MUSIC:correlator}
  R^{-1}(z) =
    \frac
    {\|W_L(z)\|_2}
    {\|P_\text{noise}W_L(z)\|_2}.
\end{equation}
The poles $z$ can then be identified as the peaks of $R^{-1}(z)$. These can be found by a continuous scan of the values of $R^{-1}(z)$, which can be done numerically.

A slightly different approach that avoids the continuous search for poles is taken by the ESPRIT algorithm.
The ESPRIT algorithm exploits a so-called `rotational invariance' property. To this end, let $W^{\downarrow}_L(z)$ and $W^{\uparrow}_L(z)$ be the sub-matrices of the Vandermonde matrix $W^T_L(z)$ that omit the last and first column, respectively.
These sub-matrices are related via
\begin{equation}
  W^{\downarrow}_L(z) = W^\uparrow_L(z) \operatorname{diag}(z).
\end{equation}
This rotational invariance property is inherited by $U_\text{signal}$.
In consequence, let $H^\downarrow$ and $H^\uparrow$ be the sub-matrix of the Hankel matrix $H$ that omits the last and first rows, respectively.
Then, in the noiseless case, a solution matrix $\Psi$ of the equation
\begin{equation}
  H^{\downarrow} = H^\uparrow \Psi
\end{equation}
has non-zero eigenvalues $z$, which are the poles contained in the data. It is given explicitly by the pseudo-inverse of $H^\uparrow$ applied to $H^\downarrow$.
Again noisy signals can be considerably de-noised by projecting $H^\uparrow$ to the signal space before inversion.
Altogether we find the algorithmic strategy of ESPRIT to be
(i) calculate the SVD of the Hankel matrix of $y$ and determine $P_\text{signal} = U_\text{signal} U^\dagger_\text{signal}$,
(ii) calculate $\Psi =  (P_\text{signal} H^\uparrow)^+ H^\downarrow$ and
(iii) determine $z$ as the eigenvalues of $\Psi$.

\subsubsection{Performance guarantees}

Non-perturbative analysis of the performance of MUSIC
 has been conducted in refs.~\cite{LiaoFannjiang:2014:MUSIC,LiLiao:2017}.
Therein, the following bound for the deviation of the noise-correlation function $R(z)$ from the ideal noiseless counter-part $R_\text{signal}(z)$
has been derived for poles $z$ of unit absolute value (sinusoids).
{The argument, however, holds verbatim for all $z \in \CC^n$.}
\begin{theorem}[Noise-correlation function bound \citep{LiLiao:2017}, Proposition~4.2]\label{thm:MUSICperformance}
  Let $E = \operatorname{Hankel}_L(\alpha)$ denote the Hankel matrix of the perturbation/noise of the signal vector $y$.
  Let $\varepsilon_{\min}$ be the smallest singular value of the Hankel matrix of the noise-free signal.
  Suppose $L \geq n, M - L + 1 \geq n$ and $2\snorm{E} < \varepsilon_{\min}$. Then
  \begin{equation}
    |R(z) - R_\text{\rm signal}(z)| \leq  \frac{2\snorm{E}}{\varepsilon_{\min}}
  \end{equation}
  for all $z\in \CC$.
\end{theorem}
We observe that the bound on $R(z)$ is proportional to the spectral-norm of the noise in the signal but in addition is decorated by a \emph{noise-enhancing} factor inversely proportional to the smallest singular value $\varepsilon_{\min}$ of the Hankel matrix.
The bound on $R(z)$ can thus not be directly translated into a bound on the precision in recovering the poles $z$ without
further assumptions, see ref.~\cite[theorem 4]{LiaoFannjiang:2014:MUSIC} in this context.
Nonetheless,
the peaks of $R^{-1}(z)$ are typically very sharp, and
the bound on $R(z)$ indicates a regime where one can typically expect MUSIC to accurately work.
For the ESPRIT algorithm, similar bounds  can be found in refs.~\cite{Fannjiang:2016:ESPRIT,LiLiaoFannjiang:2019:Super-resolution}.
The bounds for ESPRIT additionally involve the minimum singular value of the truncation (as defined above) of the Hankel matrix.

\subsubsection{Conditioning of Vandermonde matrices}
The performance guarantees for MUSIC (and ESPRIT) show a noise-enhancement inversely proportional to
the minimum singular value $\varepsilon_{\min}$ of the Hankel matrix of the ideal signal.
The minimum singular value $\varepsilon_{\min}$ in turn can be regarded as a measure for the conditioning
of the Vandermonde matrices into which the Hankel matrix decomposes.  This conditioning depends on the system parameters and on the configuration of poles.
Given expected values for the poles and the maximal sequence length, it is straight-forward to calculate the minimum singular value numerically.
This can provide valuable information in the design of \ac{RB} experiments.

More systematically,
it is informative to understand the scaling behaviour of the conditioning of the Vandermonde matrices with the help of theoretical bounds. One such bound that allows us to study its asymptotic behaviour is briefly reviewed in this section.
A lot of work has been devoted to study the often surprisingly favorable conditioning of Vandermonde matrices for poles on the unit circle, which describe sinusoidal oscillations,
see, e.g., ref.~\cite{LiLiaoFannjiang:2019:Super-resolution}
and references therein for a discussion of the phenomenon
of super resolution.

In the context of \ac{RB}, we are conversely interested in poles that are on the real-line.
% Add something about direct evaluation
A more general characterization of the conditioning of Vandermonde matrices with poles inside the unit circle (allowing for decays beyond oscillations) has been studied in ref.~\cite{Bazan}.
The conditioning obviously depends on the set of poles $z$ and the size $M$ of the Vandermonde matrix.
To state the result given in ref.~\cite{Bazan} we define several quantities. To the set of poles $z = (z_1, \ldots, z_n)$, we associate
$\check z := \max_j |z_j|$,
$\hat z := \min_j |z_j|$ and
$\ddot z := \min_{j \neq k} |z_j - z_k|$.
Furthermore, let us define
\begin{equation}
Q_M(z) = (W_M(z) W^\dagger_M(z))^{-1/2}.
\end{equation}
Note that $W_M(z)W^\dagger_M(z) $ is the frame operator of the frame defined by the rows of the Vandermonde matrix and $Q_M(z)$ is the orthogonalizing matrix arising in symmetric orthogonalization.
With the help of $Q_M(z)$,
we define the matrix
\begin{equation}
  F_M(z) := Q_M(z)\diag(z) Q^{-1}_M(z)
\end{equation}
that will play a prominent role for analyzing the Vandermonde conditioning. In particular,
its departure from normality as measured by
$D^2(F_M(z)) = \|F_M(z)\|_F^2 - \|z\|^2_{\ell_2}$ will appear.

In ref.\ \citep{Bazan}
%Bazan (2000)
a bound is derived for the $2$-norm condition number $\kappa_2(W_M) = \snorm{W_M} \snorm{W_M^+}$ through the bounding of the Frobenius norm condition number $\kappa_F(W_M) =  \fnorm{W_M} \fnorm{W_M^+}$. Here $X^+$ denotes the (Moore-Penrose) pseudo inverse of a matrix $X$.
The condition number of a linear map $A$ gives a worst-case bound on the relative reconstruction error in $\ell_2$-norm induced by an additive error in $\ell_2$-norm for a linear inverse problem.
But here we are more concerned with how it enters into the accuracy of identifying poles in the MUSIC and ESPRIT algorithms.
For the analysis of the MUSIC and ESPRIT algorithm, we want to upper bound the minimum singular value $\varepsilon^{-1}_{\min}$.
By means of the Vandermonde decomposition \eqref{eq:vandermonde_decomposition} and the sub-multiplicativity of the spectral norm,
we have $\varepsilon^{-1}_{\min} \leq \snorm{W^+_{M-L}}\snorm{W^+_L} \hat z^{-1}$.
Since $\snorm{W_{M}} \geq 1$, we conclude that
\begin{equation}\label{eq:sigmin_as_cond}
  \varepsilon^{-1}_{\min} \leq \kappa_2(W_{M-L}) \kappa_2(W_L) \hat z^{-1}\, .
\end{equation}

For the condition number the following bound holds.
\begin{theorem}[Conditioning of Vandermonde matrices \citep{Bazan}, theorem~6]
\label{thm:bazanbound}
  For $M > n \geq 2$, for a Vandermonde matrix $W_M(z)$, it holds that
  \begin{equation}
    \frac{\varepsilon_1(F_M(z))}{\check z} \leq \kappa_2(W_M(z)) \leq \frac12\left(\rho + \sqrt{\rho^2 - 4}\right)
  \end{equation}
  with
  \begin{equation}
    \rho = n \left[1 + \frac{D^2(F_M(z))}{(n-1) \ddot z^2} \right]^\frac{n-1}2 \frac{\lTwoNorm{\phi_L(\check z)}}{\lTwoNorm{\phi_L(\hat z)}} - n + 2 \, .
  \end{equation}
\end{theorem}

Most interesting in our context is the asymptotic scaling in the limit of large maximal sequence length $M$, for poles inside the unit disc $|z_i| < 1$ for all $i$. In this limit, the above bounds become tight and the following
 holds true.

\begin{lemma}[Asymptotics of condition number \citep{Bazan}, lemma~8] \label{lem:asymptotic_conditioning}
  Let $z = (z_1, \ldots, z_n) \in \CC^n$ with $|z_i| < 1$ for all $i \in [n]$.
  Define $C(z) \in \CC^{n\times n}$ as the matrix with entries
  \begin{equation}
  C_{i,j}(z) = \frac1{1 - z_i \bar z_j}.
  \end{equation}
  Then,
  \begin{equation}
    \lim_{M \to \infty} \kappa_2(W_{M}(z)) = \sqrt{\kappa_2(C(z))}\, .
  \end{equation}
\end{lemma}
Later in this section we will use this bound to perform numerical investigations of the resolving power of the MUSIC and ESPRIT algorithms and to give a sampling complexity bound for general \ac{RB}.
\subsubsection{Extensions of the algorithms}

\paragraph*{Incomplete data or logarithmic grids.} \label{sec:low-rankCompletion}
So far the presented algorithms and analysis relied on having an equidistant grid of sequence-length.
It is well-known that a low-rank matrix can under fairly general assumptions be completed from the
knowledge of just a subset of their entries \cite{NguyenKimShim:2019:Low-rank}.
Thus, given only data $y_m$ for values $m$ on an irregular subset regular grid, one can attempt at
completing the Hankel matrix for the regular grid using a low-rank matrix completion algorithm.
This pre-processing step can be combined with MUSIC or ESPRIT to arrive at pole-finding algorithms
that do not rely on complete data from an equidistant grid \cite{LiaoFannjiang:2014:MUSIC}.
In particular, we suspect that for exponential decays a logarithmic grid can potentially
yield improved recovery similar to the multiplicative error bounds for the fitting of single
exponentials derived in ref.~\cite{Harper:2019:Stats}, but we leave formally verifying this to future work.

\paragraph*{Generalization of ESPRIT to matrix exponentials.}\label{sec:generalESPRIT}
Refs.~\cite{BadeauDavidRichard:2006,BadeauRichardDavid:2008} have
generalized the ESPRIT algorithm to signal spaces spanned by products of falling polynomials and exponentials.
This is exactly the signal model \eqref{eq:RBdata:signalmodel} that we encountered for \ac{RB} output data, when the reference representation has multiplicities.
The key insight in this generalization is that the Hankel matrix of such signals admits a decomposition analogous to the Vandermonde decomposition \eqref{eq:vandermonde_decomposition} in terms of \emph{Pascal-Vandermonde matrices}.
These Pascal-Vandermonde matrices feature the same rotational invariance property underlying the ESPRIT algorithm.
Thus, one can show that when applying the standard ESPRIT algorithm to data of this form, the vector of eigenvalues of the matrix $\Psi$ is still the vector of poles $z$ with the eigenvalues appearing in multiplicities according to the maximal degree of the associated falling polynomial. Hence, ESPRIT can be directly applied to estimate matrix-exponential data series.  Noise in the signal will generically break the degeneracy of the eigenvalue spectrum,
corresponding to the fact that a generic matrix has non-degenerate eigenvalues.
Searching for regular polygons of poles allows for matching groups of perturbed poles corresponding to the same unperturbed pole.
We refer to refs.~\cite{BadeauDavidRichard:2006,BadeauRichardDavid:2008} for further details.

\subsection{Randomized benchmarking sampling complexity -- estimation of the Hankel matrix}

The performance bounds on the pole-finding algorithms, such as theorem~\ref{thm:MUSICperformance}, depend on the deviation of the Hankel matrix from ideal data in spectral norm.
In \ac{RB} protocols this error has two contributions:
\begin{enumerate}
	\item The finite sampling statistics of the measurements, which yields a statistical error of the mean estimator $\hat p(m)$.
\item The perturbative error that comes from neglecting sub-dominant eigenvalues, which is controlled by our theorems~\ref{thm:mother}, \ref{thm:non_uniform}, and \ref{thm:subset_mother}.
\end{enumerate}

For the finite sampling error, we provide the following bound. To this end, we model the
individual measurement performed during the \ac{RB} protocol by a random variable $\hat Y_m$.
To simplify the notation in the proof, we assume that the number of different sequence lengths is even and use a square Hankel matrix.

\begin{lemma}[Statistical estimation]\label{lem:statistical_estimation}
  Let $M$ be even and $L = M / 2$.
  For $m \in [M]$, let $\hat Y_m$ be a random variable taking values in $[0,1]$ with $\Var[\hat Y_k] \leq \varepsilon^2$.
  Furthermore, let $\hat p(m) = \frac 1 N \sum_{i=1}^N \hat Y^{(i)}_m$ the corresponding mean estimator of $N$ i.i.d.\ copies $\hat Y^{(i)}_m$ of $\hat Y_m$. We denote with $\operatorname{Hankel}_{L}(\hat p)$ the Hankel matrix of the vector $\hat p = (\hat p(m))_{m \in [M]} \in \RR^{M}$.
  Then,
    \begin{equation}
  \snorm{\operatorname{Hankel}_{L}(\hat p) - \EE \operatorname{Hankel}_{L}(\hat p)} \leq \epsilon
    \end{equation}
    with probability $1 - \delta$
  provided that
  \begin{equation}
    N \geq 4 \max \left\{ \frac{M\varepsilon^2}{\epsilon^2}, \frac 2{3 \epsilon} \right\} \log\frac M \delta.
  \end{equation}
\end{lemma}

Combining lemma~\ref{lem:statistical_estimation} with the performance bound for MUSIC, theorem~\ref{thm:MUSICperformance}, and \eqref{eq:sigmin_as_cond} we can state the following result for the overall sampling complexity of random benchmarking experiments.

\begin{corollary}[Sampling complexity]\label{cor:sampling_complexity}
    Let $M$ be even and $L = M / 2$. And $z = (z_i)_{i=1}^n$ be a set of poles.
    For $m \in [M]$ let $\hat p(m)$ be the mean estimator of i.i.d.\ copies of random variables with variance bounded by $\varepsilon^2$.
    Choose $\tilde \epsilon, \delta > 0$, provided that the total number of random trials is
    \begin{equation}
       N_\text{total} \geq 8 \kappa^4_2(W_{M/2}(z)) \hat z^{-2} \frac{M \varepsilon^2}{\tilde\epsilon^2} \log\frac M \delta
    \end{equation}
    and
    \begin{equation}
        N_\text{total} \geq \frac{16}3\, \kappa^2_2(W_{M/2}(z)) \hat z^{-1} \frac1{\tilde\epsilon} \log\frac M \delta
    \end{equation}
    for the noise space correlation function \eqref{eq:MUSIC:correlator} defined by the MUSIC algorithm with input data $\hat p$
    it holds that $ | R(z') - R_\text{\rm signal}(z') | \leq \tilde\epsilon$ with probability $\delta$.
\end{corollary}

We state this bound in terms of the condition number of the Vandermonde matrix, which allows us to make analytic claims about the behaviour of the sampling complexity in various regimes. However, one can state an equivalent bound in terms of the smallest singular value, which will often be significantly smaller. It is, however, difficult to work with analytically.

For the application of corollary \ref{cor:sampling_complexity} to \ac{RB} data processing, one has to additionally control the perturbative error appearing in theorems~\ref{thm:mother}, \ref{thm:subset_mother}, \ref{thm:non_uniform}. The perturbative error $\epsilon$ per \ac{RB} data point, see e.g.,~\eqref{eq:mother}, yields an additive error in the noise correlation function of order of $M\hat z^{-1} \kappa^2_2(W_{M/2}(z)) \epsilon$.
The scaling with $M\epsilon$ originate from the spectral norm of the Hankel matrix and the factor of $z^{-1}\kappa^2_2(W_{M/2}(z))$ captures the noise-enhancement.

Lemma \ref{lem:statistical_estimation} follows from the Matrix Bernstein bound~\citep{TroppBernstein, AW02}
that requires us to control
the spectral norm and matrix variance statistics in order to provide a tail bound for sums of matrices.
We follow the same strategy as presented in ref.~\citep{TroppBernstein} for Toeplitz matrices.

\begin{proof}[Proof of lemma~\ref{lem:statistical_estimation}]
  With the help of the $L \times L$ exchange matrix
  \begin{equation}
    J_{i,j} = \begin{cases} 1 & j = L - i + 1 \\ 0 & \text{else} \end{cases}
  \end{equation}
  and the $L \times L$ (non-cyclic) shift matrix $X$ that has ones its first upper off-diagonal and zeros everywhere else
  we can write
  \begin{equation}
    \operatorname{Hankel}_L(\hat p) = \sum_{k=-L+1}^{L-1} \hat p_k X^k J,
  \end{equation}
  where we identify the elements of $p$ cyclically.
  We define
    \begin{equation}
    S^{(i)}_k := \frac1N(\hat Y^{(i)}_k - \EE[\hat Y^{(i)}_k]) X^k J
      \end{equation}
      such that $\operatorname{Hankel}_L(\hat p) - \EE\operatorname{Hankel}_L(\hat P) = \sum_{i=1}^N\sum_{k=-L+1}^{L-1} {S_k}$ is the sum of the random matrices $S^{(i)}_k$.
  Since     \begin{equation}
  \snorm{X} = \snorm{J} = 1
   \end{equation} and $\hat Y_k$ takes values in $[0,1]$, we have that
     \begin{equation}
   \snorm{S^{(i)}_k} \leq 2/N
      \end{equation}
      for all $i, k$.
  For the matrix variance we calculate that
  \begin{equation}
     \sum_{k=-L+1}^{L-1} \EE[S^{(i)}_k(S^{(i)}_k)^\dagger] = \frac1{N^2}\sum_{k=-L+1}^{L-1} \Var[\hat Y_k] X^k X^{-k} = \frac{1}{N^2}\sum_{k=-L+1}^{L-1} \Var[\hat Y_k] P_k,
   \end{equation}
   with $P_k$ a diagonal projector having $k$ ones on the diagonal and zeros everywhere else.
   One finds the same structure for $\sum_{k=-L+1}^{L-1} \EE[(S^{(i)}_k)^\dagger S^{(i)}_k]$ analogously.
   By the assumption of the lemma $\Var[\hat Y_k] \leq \varepsilon^2 $.
   Therefore, matrix variance statistics is dominated as
   \begin{equation}
      \max \left\{ \snorm{\sum_{i=1}^N\sum_{k=-L+1}^{L-1} \Ev(S_kS\ad_k)},\snorm{\sum_{i=1}^N\sum_{k=-L+1}^{L-1} \Ev(S_k\ad S_k)} \right\} \leq \frac{M \varepsilon^2} N.
   \end{equation}

   The matrix Bernstein inequality \cite{TroppBernstein} yields
   \begin{equation}
     \PP\left[ \snorm{\sum_{i=1}^N \sum_{k=-L+1}^{L-1} S_k} \geq \epsilon \right] \leq M \exp\left(
        - \min \left\{
            \frac{\epsilon^2 N}{4M\varepsilon^2}, \frac{3\epsilon N} 8 \right\} \right).
   \end{equation}
   Requiring the right hand side to be dominated by $\delta$ and solving for $N$ yields the lemma's assertion.
  \end{proof}

\subsection{Vandermonde conditioning for \acl{RB} decays}

The noise-enhancements factor in the performance guarantee for the tone-finding algorithms MUSIC and ESPRIT is given
by the inverse of the minimum singular value $\varepsilon^{-1}_{\min}$ of the Hankel matrix of the ideal, noise-free signal.
This minimum singular value \eqref{eq:sigmin_as_cond} is in turn controlled by the minimal absolute value of the poles and the conditioning of the Vandermonde matrix $W_L(z)$ associated with the poles and the signal length. Here we numerically investigate this conditioning in various scenarios relevant to \ac{RB}.
We express all data in terms of the dimension of the Hankel matrix $L$, which one can generally take as being about half of the maximal sequence length $M$.

\begin{figure}[t]
  \includegraphics[width=.8\textwidth]{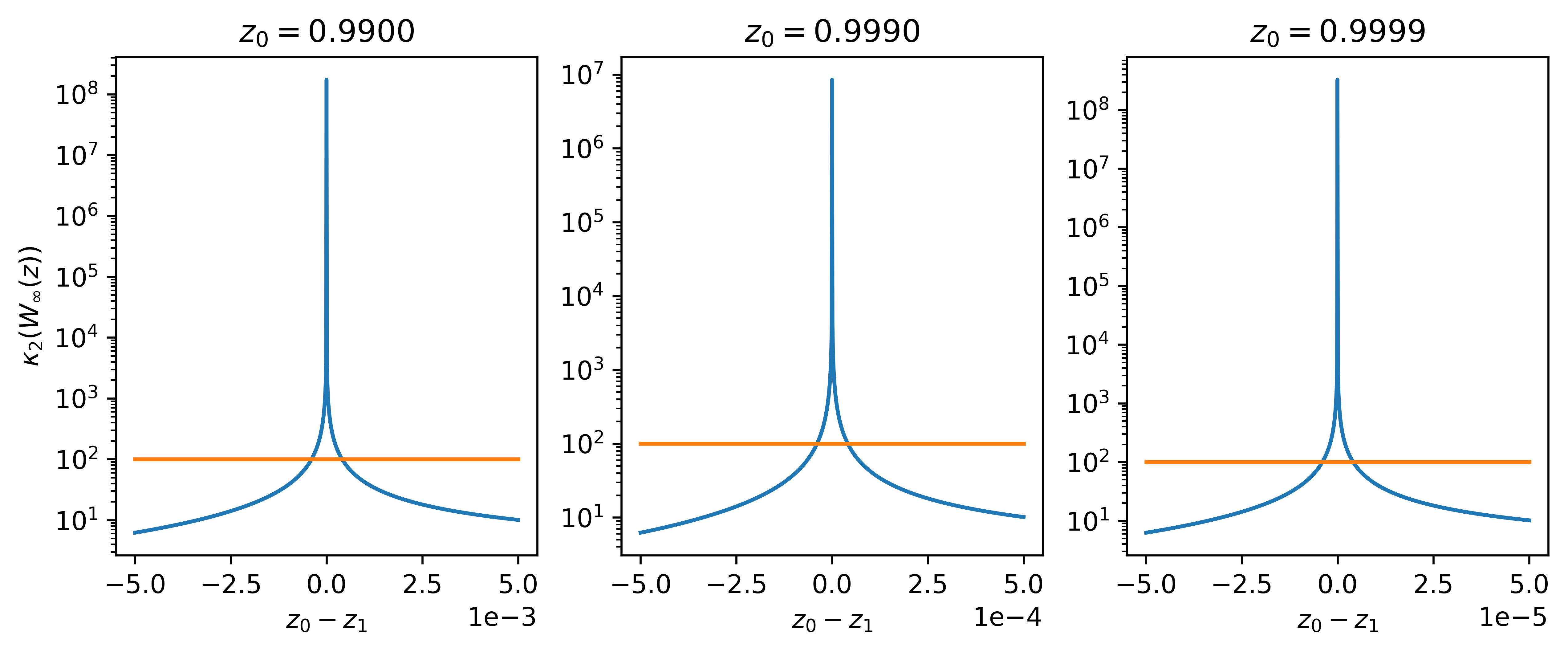}
  \caption{\label{fig:plotAsymptoticConditioningForPairs} Here we show the dependency of the conditioning number of the Vandermonde matrix on the spacing of two poles $z_0,z_1$, for infinite sequence length. We see that the conditioning depends drastically on the distance between the two poles, but not on the absolute location of the poles on the real line. The orange line at $10^2$ is added for the purpose of comparison. }
\end{figure}
\begin{figure}[t]
  \includegraphics[width=.6\textwidth]{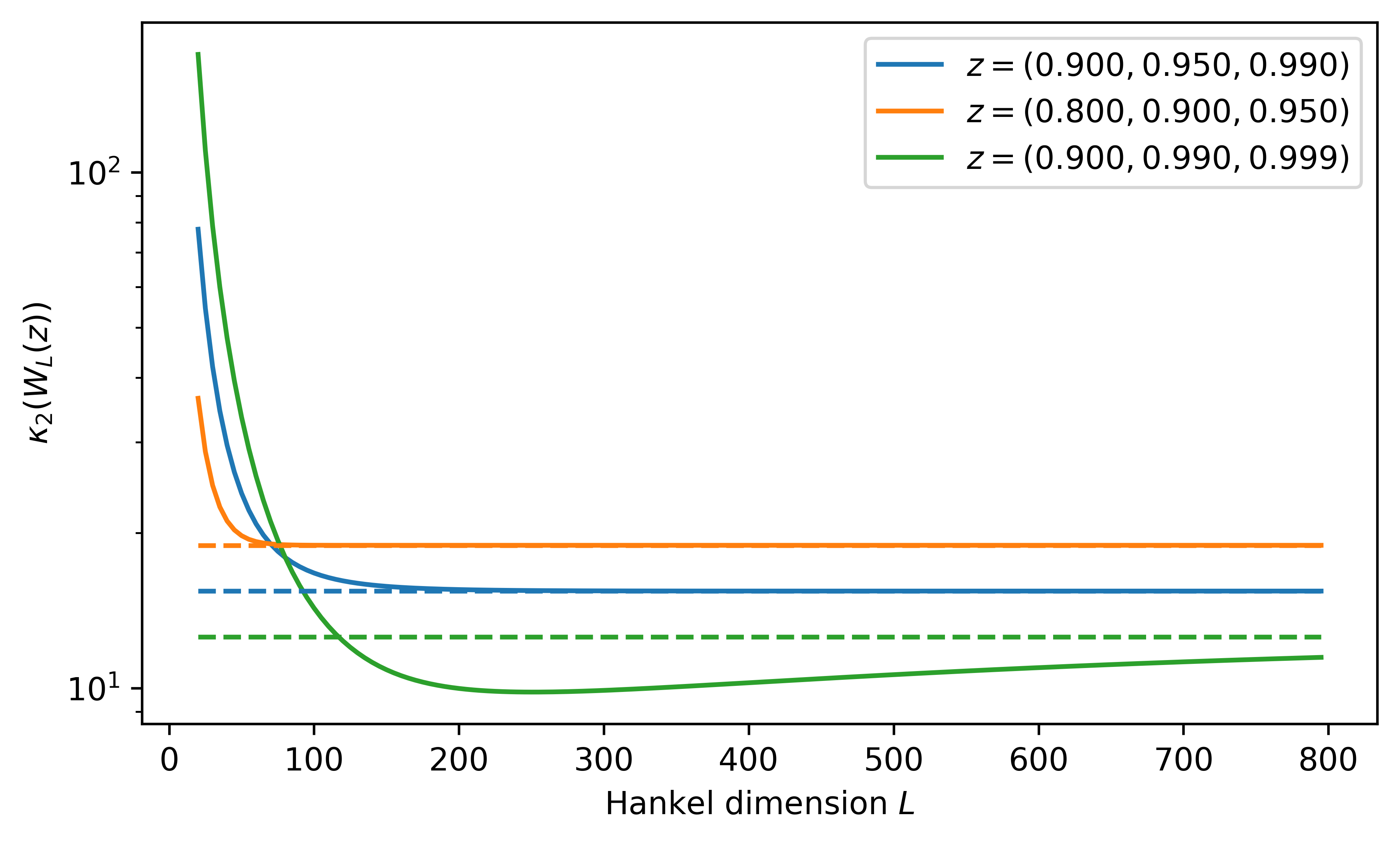}
  \caption{\label{fig:condVsSequenceLength} The condition number of the Vandermonde matrix for different Hankel matrix dimensions ($\propto$ \ac{RB} sequence length) for three different sets of poles.
  The dashed lines indicate the asymptotic expression of Lemma~\ref{lem:asymptotic_conditioning}.
  Note that the minimum for the green line is due to the scaling of the maximal eigenvalue of $W_{M-L}$. We observe (not depicted here) that the minimum singular value of Hankel matrix, as appearing in theorem~\ref{thm:MUSICperformance}, is monotonically increasing in $L$ for all three sets of poles.}
\end{figure}
When the \ac{RB} data model is described by many poles that are close in value
the noise-enhancement due to bad conditioning can be limiting factor rendering the extraction
of poles infeasible.

Increasing the sequence length improves the conditioning of $W_L(z)$, see figure~\ref{fig:condVsSequenceLength}.
But theorem \ref{thm:bazanbound} shows that the condition number of $W_L(z)$ is
even in the asymptotic limit $W_\infty(z)$ for large $L$ bounded away from zero.
Thus, increasing the length of observed \ac{RB} series only improves the conditioning up to a certain point.

The explicit expressions of the upper and lower bounds on the condition number in theorem~\ref{thm:bazanbound} have a rather complicated dependency on the geometrical constellation of the poles.
One can argue that for \ac{RB} data with poles on the real line there are roughly speaking two effects coming into play: (1) The spacing of the poles and (2) the number of poles.

To illustrate the dependency on the spacing of the poles, we have numerically evaluated the $\kappa_2(W_\infty(z))$ for different pairs of poles as they might appear in \ac{RB} data. The result is shown in fig.~\ref{fig:plotAsymptoticConditioningForPairs}.
The first pole is chosen to deviate from $1$ by a value $r \in \{10^{-2}, 10^{-3}, 10^{-4}\}$, the second pole is
chosen at different values around the first one.
We observe that as both poles move together the condition number diverges.
Importantly, the size of the interval in which the condition number grows over a certain threshold scales with $r$.
Correspondingly, we expect that poles closer to $1$ can be still resolved with a smaller spacing compared to poles
that deviate considerably from $1$.

\begin{figure}[t]
  \includegraphics[width=.6\textwidth]{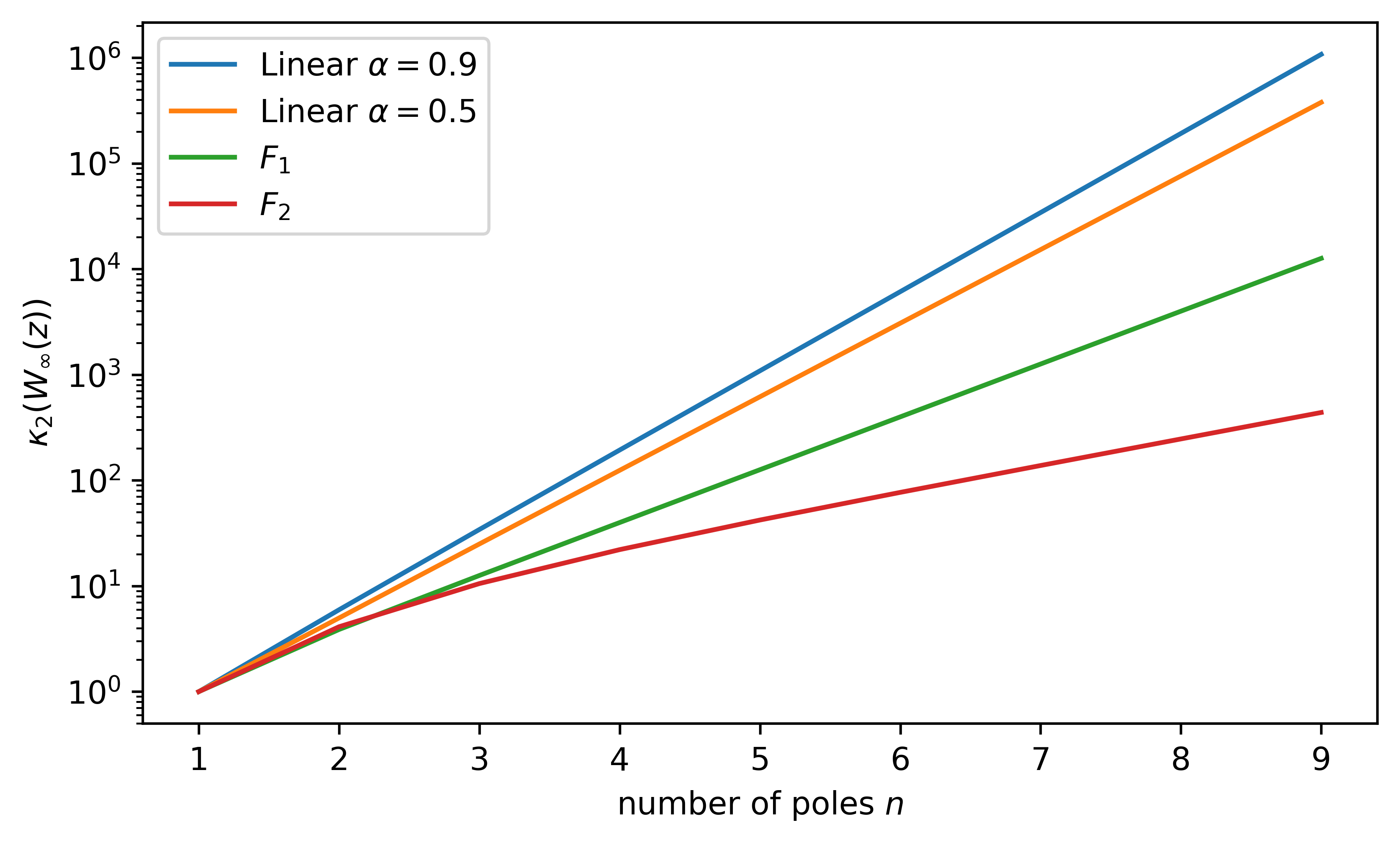}
  \caption{\label{fig:asympCondScalingInN} The dependency of the condition number in the limit of infinite sequence length on the number of poles for different families of poles. These families are defined in table \ref{tab:families}.}
\end{figure}

Secondly, even if the poles are spaced such that the ratio of the departure from normality and the minimum spacing are fixed the upper bound in theorem~\ref{thm:bazanbound} exhibits an exponential dependency on the number of poles.
We numerically evaluate this dependency for different families of poles that each defines a set of poles for every cardinality, see table~\ref{tab:families}.
These families include linearly spaced poles within the interval $(\alpha, 1)$ and the pole families
$F_a(n) = (z_i = 1 - 10^{-i/a} \mid i \in [n])$ for positive real $a$.
For example, $F_1(n) = (.9, .99, .999, \ldots)$ which can be regarded as featuring exponentially spaced `infidelities'.

\newcommand{\tabbodysize}{\scriptsize}

\begin{table}
\begin{tabular}{l >{\tabbodysize}c >{\tabbodysize}c >{\tabbodysize}c >{\tabbodysize}c}
  \toprule
  n &\normalsize 2 &\normalsize 4 &\normalsize 6 \\
  \midrule
  Lin. $\alpha = .9$ &
  $(0.9,  0.95)$ &
  $(0.9,   0.925, 0.95,  0.975)$ &
  $(0.9,        0.9167, 0.9333, 0.95,       0.9667, 0.9833)$
  \\
  Lin. $\alpha = .5$ &
  $(0.5,  0.75)$ &
  $(0.5,   0.625, 0.75, 0.875)$ &
  $(0.5,   0.5833, 0.6667, 0.75,  0.8333, 0.9167)$
  \\
  $F_1$ &
  $(.9, .99)$ &
  $(.9, .99, .999, .9999)$ &
  $(.9, .99, .999, .9999, .99999, .999999)$ %
  \\
  $F_2$ &
  $(0.9,     0.9684)$ &
  $(0.9,       0.9684, 0.99, 0.9968)$ &
  $(0.9,        0.9684, 0.99,  0.9968, 0.999, 0.9997)$
  \\
  \bottomrule
\end{tabular}
\caption{\label{tab:families} Examples of pole families for different numbers of poles $n$.}
\end{table}

Fig.~\ref{fig:asympCondScalingInN} depicts the dependency of $\kappa_2(W_\infty(z))$ on the number of poles $n$ for different families.
We find that due to a typically exponential dependency, the conditioning indicates that the reconstruction of multiple poles becomes demanding for already small numbers $n$.

Note that the conditioning is significantly improved if the poles are not exclusively on the real line but also have non-vanishing imaginary parts.
Such pole sets for example arise in the \ac{RB} variant of ref.~\cite{PhysRevLett.123.060501} focusing on individual gates.

\subsection{Performance evaluation}

\begin{figure}[t]
  \includegraphics[width=.5\textwidth]{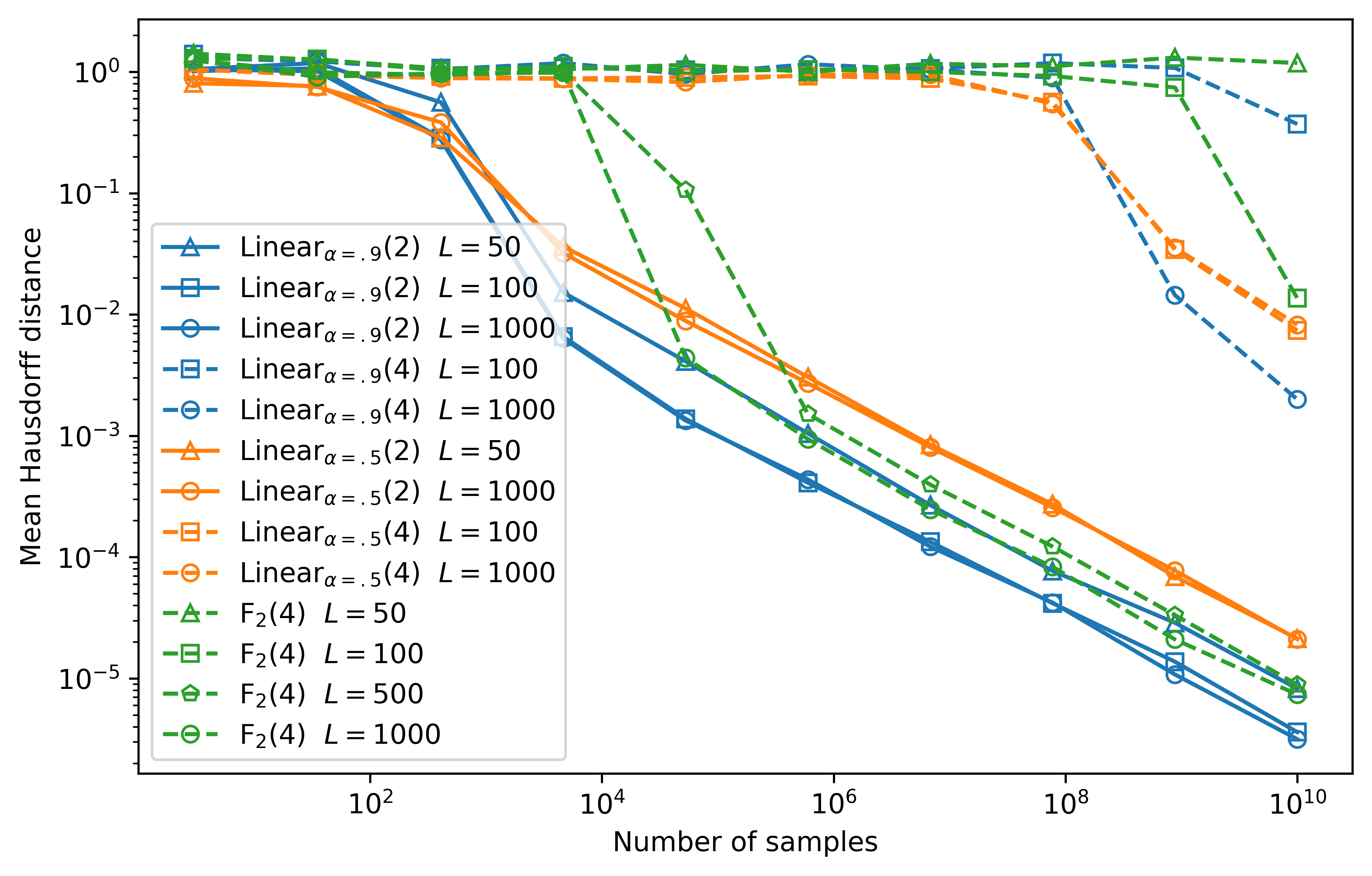}
  \caption{\label{fig:ESPRITHausdorffVsNos} Mean Hausdorff distance between the real set of poles and the reconstructed set of poles (via ESPRIT) for different families of poles (as defined in table \ref{tab:families}) and Hankel dimension $L$ ($\propto$ maximal \ac{RB} sequence length $M$) versus the number of samples used per expectation value estimation. Each data point is averaged over $100$ repetitions. For all families we see that the reconstruction essentially fails until a sampling threshold is reached, after this threshold the accuracy of the estimation increases rapidly with increased number of samples. This threshold increases strongly with the number of poles in the family across all families and also depends on the maximal sequence length. This latter dependence is mediated by the actual locations of the poles in the complex plane, which is as expected. }
\end{figure}
\begin{figure}[t]
  \includegraphics[width=.5\textwidth]{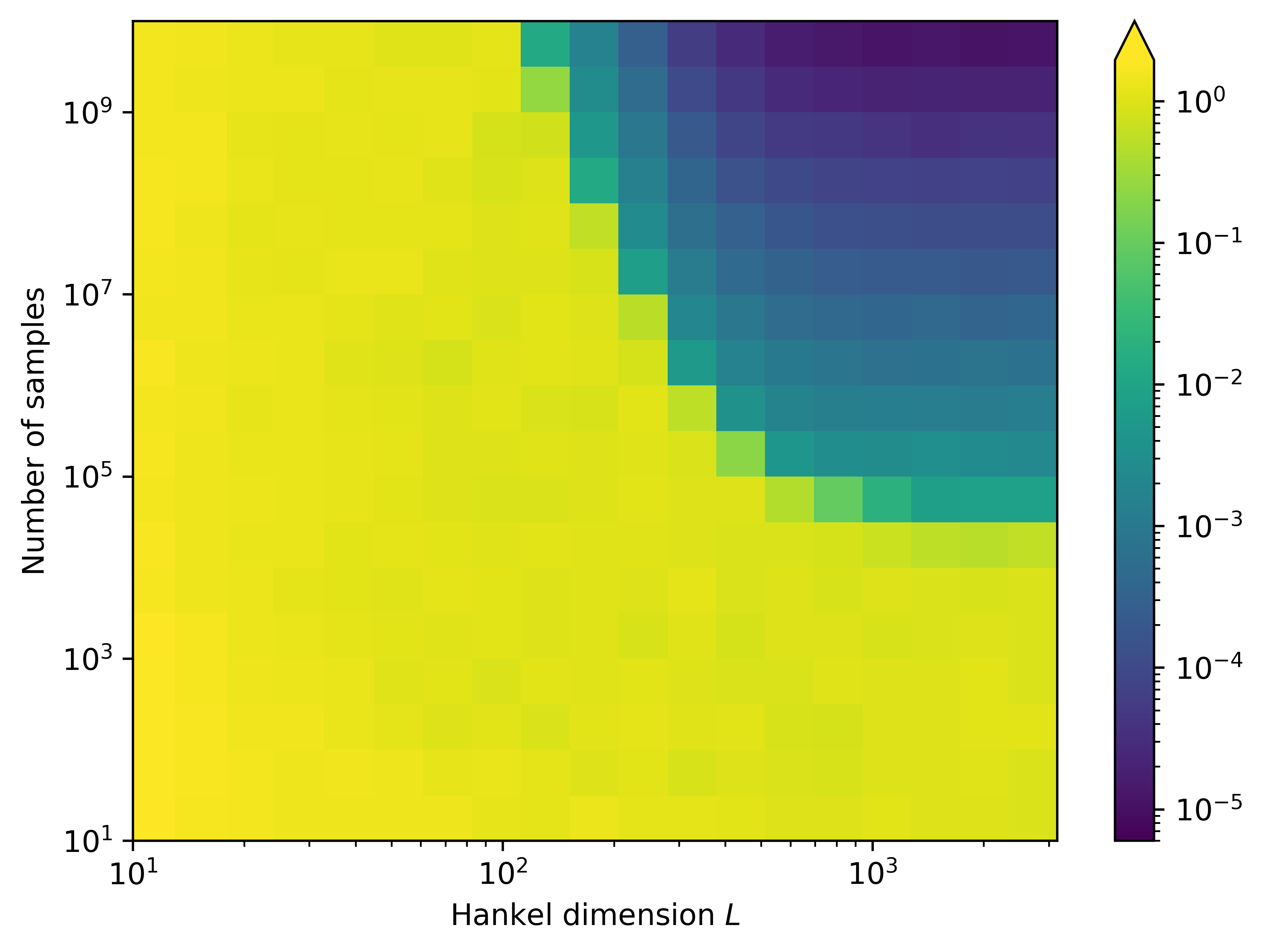}
  \caption{\label{fig:ESPRITHausdorffVsNosAndL} Mean Hausdorff distance of the reconstruction (via ESPRIT) for poles $z = F_2(4) = (0.9 ,       0.968 ,0.99  ,     0.997)$ for different number of samples and Hankel dimension $L$ ($\propto$ maximal \ac{RB} sequence length $M$). Each data point is averaged over $100$ repetitions. We see again that reconstruction essentially fails, until a threshold is reached both in number of samples and in maximal sequence length after which the accuracy of reconstruction increases with increasing number of samples and in $L$.}
\end{figure}

After collecting evidence that the reconstruction of multiple poles quickly becomes a demanding task.
We here show that for moderate configurations (i.e., not too many poles, not too close together) the ESPRIT algorithm is suitable for the post-processing of
\ac{RB} data. To this end, we implemented the ESPRIT algorithm in Python. For a fixed set of poles the ideal data series (constructed from the poles and a fixed identical pre-factor)
is made noisy by randomly sampling binomial distributions. This simulates the random noise due to finite statistics for a certain number of samples per sequence length.
Subsequently, the set of poles is reconstructed from the noisy data using the ESPRIT algorithms.
We compare the reconstructed set of poles with the ideal set of poles using the
symmetric Hausdorff distance.
Let $z \in \CC^n$ and $z' \in \CC^{n'}$
\begin{equation}
  d_H(z, z') = \max \{d_{dH}(z' ; z), d_{dH}(z ; z')\}\, , \qquad d_{dH}(z , z') = \max_{k \in [n]} \min_{k' \in [n']} |z_k - z'_{k'}|\, .
\end{equation}
Fig.~\ref{fig:ESPRITHausdorffVsNos} displays the mean Hausdorff distance for different
number of samples. Each data point is averaged over $100$ repetitions.
Fig.~\ref{fig:ESPRITHausdorffVsNosAndL} depicts the mean Hausdorff distance for different numbers of samples and maximal sequence lengths. In both of these plots we note a threshold effect where the reconstruction of the poles essentially fails until a threshold of samples and maximal sequence length is reached, after which reconstruction accuracy increases with increasing number of samples. This phenomenon is observed for different families of poles and the location of the threshold depends strongly on the number of poles in the signal. It is interesting to note in fig.~\ref{fig:ESPRITHausdorffVsNosAndL} that the minimal number of samples needed for reconstruction is dependent on the maximal sequence length. Since increasing the maximal sequence length has an implicit sampling cost, this points to a non-trivial optimization problem in allocating resources. We leave further investigation of the optimal point for a family of poles for further research. The conclusion from these numerical investigations is that the \ac{RB} decay rate recovery problem is feasible using modern methods when the number of poles is small but rapidly becomes impractical as the number of poles grows.

\section{Isolating matrix exponentials associated with a representation}\label{sec:new_crb}

We have seen in section \ref{sec:output_data} that for uniform \acf{RB} the output data is well described by a linear combination of (matrix) exponential decays associated with irreducible sub-representations of a reference representation. The decay rates can in principle be extracted by the methods described in section \ref{sec:data_processing}. However, two issues crop up here: (1) the sample complexity of extraction is strongly dependent on the number of decays present in the \ac{RB} output data, limiting \ac{RB} to groups with reference representations containing at most a few irreducible sub-representations, and (2) upon successful extraction of decay constants, it is not clear a priori how they are related to the different irreducible sub-representations present, making it hard to relate the decay constants to the average fidelity.

A data processing technique that addresses this problem was proposed in various papers (marked with a $*$ in fig.\ \ref{fig:results overview}) such as the dihedral benchmarking scheme \citep{carignan2015characterizing} for the single qubit dihedral group, the character benchmarking scheme \citep{helsen2019new} which works for general groups (with some technical constraints on the reference representation) and the Pauli channel tomography scheme \citep{flammia2019efficient} and cycle benchmarking \citep{CycleBenchmarking} for the Pauli group (in ref.\  \citep{flammia2019efficient} multiple decays are actually estimated in parallel).  The unifying theme in all of these procedures is that one estimates \ac{RB} output data $p(i,m, \gend)$ for different {ending gates} $\gend\in \gr$, and then correlates the resulting vector of signals $[p(i,m,\gend)]_{\gend}$ with a scalar function $f_\lambda(\gend)$ (which can be thought of as a dual vector) that depends on an irreducible sub-representation $\sigma_\lambda$ of the reference representation $\omega$.

In this section we will take this idea and generalize it as far as possible. In particular we will propose a post-processing method that, for any group $\gr$ and reference representation $\omega$, takes in \ac{RB} output data $p(i, m, \gend)$ (for all $\gend \in \gr$) and an irreducible sub-representation $\sigma_\lambda$ of the reference representation $\omega$, and outputs post-processed data $k_\lambda(m)$ that only depends on the (matrix) exponential decay associated with $\sigma_\lambda$. We will state theorems for uniform \ac{RB}, but the discussion below generalizes to the other types of \ac{RB}.

We note that all examples of \ac{RB} schemes without inversion gates (marked with a $**$ in fig.\ \ref{fig:results overview}) can be seen as special cases of the procedure given below, where the output data $[p(i,m,\gend)]_{\gend}$ is simply averaged over $\gend$.
{We would also like to note that the procedure defined here obviates the need for explicitly implementing the inversion gate (as it can be simply absorbed by redefining $\gend$). This makes the protocol more experimentally practical.}

\subsection{The post-processing procedure}

We begin by defining \emph{filter functions} $\alpha_\lambda$ (associated with a representation $\sigma_\lambda$)
\begin{equation}
\alpha_\lambda: \gr\times I \to \mathbb{C} :g,i \mapsto\braa{\Pi_i} \mc{P}_\lambda \overline{\omega}(g)\kett{\rho_0}
\end{equation}
where $\mc{P}_\lambda:\mc{S}_d\to \mc{S}_d$ is the projection onto the sub-representation $\sigma_\lambda^{\oplus n_\lambda}$ of the reference representation $\omega$. This is (up to normalization) the matrix element of the sub-representation $\sigma_\lambda^{\oplus n_\lambda}$ corresponding to the vectors $\kett{\rho_0}$ and $\braa{\Pi_i}$.
From the \ac{RB} data and the above matrix element function we can now compute the following quantity we call the $\lambda$-filtered \ac{RB} output data:
\begin{equation}\label{eq:lambda_correlation}
k_\lambda(m)  = \gsum{\gend}\sum_{i\in I} N_\lambda^{-1} \alpha_\lambda(i, \gend) p(i,m,\gend)
\end{equation}
where the normalization constant is given by
\begin{equation}
N_\lambda = \gsum{g}\sum_{i\in I}
\alpha_\lambda(g,i )\braa{\Pi_i} {\omega}(g)\kett{\rho_0}.
\end{equation}
One can think of this quantity as measuring the presence of the sub-representation $\sigma_\lambda$ in the data $p(i, \gend, m)$. We will make this more precise in the following theorem:
\begin{theorem}[Measuring sub-representations in the data]\label{thm:char_avg}
Let $\gr$ be a finite group and $\omega:\gr\to \mc{S}_d$ a reference representation of $\gr$ with decomposition $\omega = \bigoplus_{\lambda'\in \Lambda} \sigma_{\lambda'}^{\oplus n_{\lambda'}}$. Moreover, let $\phi$ be an implementation of $\omega$ for which Theorem \ref{thm:mother} holds.
For a fixed $\lambda\in \Lambda$ consider the $\lambda$-filtered data $k_\lambda(m)$ as defined in \eqref{eq:lambda_correlation}. As a function of $m$ we now have that
\begin{equation}
|k_\lambda(m) - \tr(B_\lambda M_\lambda^m)|\leq  8K\left(\delta\left[1+ \frac{2\delta}{1-5\delta}\right] \right)^{m}
\end{equation}
where $B_\lambda$ is an $n_\lambda\times n_\lambda$ matrix encoding SPAM terms, $M_\lambda$ is given by the projection onto the subspace associated with the $n_\lambda$ largest eigenvalues of $\mc{F}(\phi)[\sigma_\lambda]$ (as given in theorem \ref{thm:mother}), and $K$ is some constant independent of $m$.
\end{theorem}

\begin{proof}
We know from theorem \ref{thm:mother} that
\begin{equation}
|p(i,m, \gend) -  \sum_{\lambda'\in \Lambda}\tr(A_{\lambda'}M^m_{\lambda'})| \leq 8\left(\delta\left[1+ \frac{2\delta}{1-5\delta}\right] \right)^{m},
\end{equation}
with $A_{\lambda'}$ given in eq.\ (\ref{eq:spamfactor}).
From the definition of $k_\lambda(m)$,  we can thus compute
\begin{align}
k_\lambda(m) &= \frac{1}{|\gr|}\sum_{\gend\in \gr}\sum_{i\in I} N_\lambda^{-1} \alpha_{\lambda}(i,\gend) \sum_{\lambda'\in \Lambda}\tr(A_{\lambda'}M^m_{\lambda'})  \\&\hspace{5em}+\frac{1}{|\gr|}\sum_{\gend\in \gr}\sum_{i\in I} N_\lambda^{-1} \alpha_{\lambda}(i,\gend) \left(%
p(i, \gend, m)%
- \sum_{\lambda'\in \Lambda}\tr(A_{\lambda'}M^m_{\lambda'})%
\right) .
\end{align}
Considering only the first term, and inserting the definition of $\alpha_{\lambda}(i,\gend)$ we are interested in the SPAM operator quantity
\begin{equation}
B_{\lambda,\lambda'} = \frac{1}{|\gr|}\sum_{\gend\in \gr}\sum_{i\in I}\braa{\Pi_i}\mc{P}_\lambda \overline{\omega}(\gend)\kett{\rho_0}A_{\lambda'}
\end{equation}
for $\lambda'\in \Lambda$.
From the proof of theorem \ref{thm:mother} (eq.\ (\ref{eq:spamfactor})), 
we can recover an expression for the $n_{\lambda'}\times n_{\lambda'}$ matrix $A_{\lambda'}$:
\begin{equation}
[A_{\lambda'}]_{j,j'} = d_{\sigma_{\lambda'}}\braa{\mc{E}_{\mathrm{M}}(\Pi_i)}\tr_{V_{\sigma_{\lambda'}}}\left[(\overline{\sigma}_{\lambda'}(\gend^{-1})\otimes \1) R^{\lambda'}_1 F(\mc{P}_{\lambda'}^j \omega\mc{P}_{\lambda'}^{j'}) {L_1^{\lambda'}}\ct \right]\kett{\mc{E}_{\mathrm{SP}}(\rho_0)}
\end{equation}
where $\mc{P}^j_{\lambda'}$ is the projector onto the $j$'th copy of $\sigma_{\lambda'}$ in the reference representation $\omega$ and $R^{\lambda'},L_1^{\lambda'}$ encode the deviation of $\phi$ from $\omega$ (their precise shape is not relevant for our argument).  By linearity,
we can now consider
\begin{align}
&[B_{\lambda,\lambda'}]_{j,j'}= \\
& \sum_{i\in I}d_{\sigma_{\lambda'}} \braa{\Pi_i\otimes \mc{E}_{\mathrm{M}}(\Pi_i)}\gsum{\gend} (\mc{P}_\lambda \overline{\omega}(\gend^{-1}))\otimes \tr_{V_{\sigma_{\lambda'}}}\left[(\overline{\sigma}_{\lambda'}(\gend^{-1})\otimes \1) R^{\lambda'}_1 F(\mc{P}_{\lambda'}^j \omega\mc{P}_{\lambda'}^{j'}) {L_1^{\lambda'}}\ct \right] \kett{\rho_0\otimes\mc{E}_{\mathrm{SP}}(\rho_0)}\\
&= \sum_{i\in I} d_{\sigma_{\lambda'}}\braa{\Pi_i\otimes\mc{E}_{\mathrm{M}}(\Pi_i)} \tr_{V_{\sigma_{\lambda'}}}\left[\delta_{\lambda,\lambda'}(\mc{F}(\omega)[\sigma_\lambda]\otimes \1) \1\otimes ( R^{\lambda'}_1 F(\mc{P}_{\lambda'}^j \omega\mc{P}_{\lambda'}^{j'}) {L_1^{\lambda'}}\ct) \right]\kett{\rho_0\otimes\mc{E}_{\mathrm{SP}}(\rho_0)},
\end{align}
where we have used that
\begin{equation}
\gsum{\gend} \mc{P}_\lambda\omega(\gend) \otimes \sigma_{\lambda'}(\gend) =\gsum{\gend} \sigma_\lambda^{\oplus n_\lambda}(\gend) \otimes \sigma_{\lambda'}(\gend) =  \delta_{\lambda,\lambda'} \mc{F}(\omega)[\sigma_{\lambda'}],
\end{equation}
which is the Fourier transform analog of the orthogonality of characters of irreducible representations. Hence $B_{\lambda,\lambda'} = \delta_{\lambda,\lambda'}B_{\lambda,\lambda'} := B_{\lambda}$.

Plugging this back into the expression for $k_\lambda$ we get
\begin{equation}
k_\lambda(m) =  \tr(B_\lambda M_{\lambda}^m) + \frac{1}{|\gr|}\sum_{\gend\in \gr}\sum_{i\in I} N_\lambda^{-1} \alpha_{\lambda}(i,\gend) \left(\sum_{\lambda'\in \Lambda}\tr(A(\Pi_i,\gend)_{\lambda'}M^m_{\lambda'}) - p(i, \gend, m)\right).
\end{equation}
We can thus upper bound the difference $k_\lambda(m) -  \tr(B_\lambda M_{\lambda})$ by considering the magnitude of the difference term. Note that we know from theorem\ (\ref{thm:mother}) that $\left(\sum_{\lambda'\in \Lambda}\tr(A(\Pi,\gend)_{\lambda'}M^m_{\lambda'}) - p(i,m \gend)\right) \leq O(\delta^m)$. It follows that there exists a $K$ such that
\begin{equation}
\bigg|\frac{1}{|\gr|}\sum_{\gend\in \gr}\sum_{i\in I} N_\lambda^{-1} \alpha_{\lambda}(i,\gend) \left(\sum_{\lambda'\in \Lambda}\tr(A(\Pi,\gend)_{\lambda'}M^m_{\lambda'}) - p(i,m, \gend)\right)\bigg| \leq 8K\left(\delta\left[1+ \frac{2\delta}{1-5\delta}\right] \right)^{m} .
\end{equation}

\end{proof}
Hence, the $\lambda$-filtered output data has essentially the same behaviour as regular \ac{RB} data, except that only the Fourier mode associated with $\sigma_\lambda$ is included in the signal. One can think of the $\lambda$ filter function $\alpha_\lambda$ as placing a delta-peak filter function centered on the `frequency' $\sigma_\lambda$. Note that by linearity we get essentially the same result if one defines a filter function associated with non-irreducible representations (via a direct sum of irreducible representations). This can be thought of as placing a frequency comb on the \ac{RB} data.
{
Finally, it is interesting to explicitly write down the form of the SPAM matrix $B_\lambda$ in the limit of no SPAM and perfect gates. In the case of a multiplicity-free reference representation $\omega$ we have
\begin{equation}
B_{\lambda} = N_\lambda^{-1}\sum_{i\in I} \braa{\Pi_i\tn{2}} \mc{F}(\omega)[\sigma_\lambda]\kett{\rho_0\tn{2}},
\end{equation}
which emphasizes the importance of the normalization constant (on which more later), but also the importance of choosing $\rho$ and $\{\Pi_i\}_{i \in I}$ such that $B_\lambda$ is non-zero.
}

\subsection{Statistical estimation}\label{subsec:correlator_estimator}
When computing the filtered output data $k_\lambda(m)$ in the previous section we assumed we had access to the \ac{RB} output data $p(i,\gend,m)$ for all $i \in I$ and $\gend\in \gr$. This is not realistic since both the size of the POVM $\{\Pi_i\}_{i\in I}$ and the size of the group $|\gr|$ can be exponential in the number of qubits.  In practice we will need to construct a statistical estimator $\hat{k}_\lambda$ for $k_\lambda$, and argue that $\hat{k}_\lambda$ is a good approximation for a reasonable number of samples. This we will do in this section.

Note that the normalization factor $N_\lambda$ is essential in lower bounding the magnitude of the filtered function $k_\lambda$ (i.e., making sure that the number $k_\lambda$ is not too small). However, this normalization factor can be proportional to the Hilbert space dimension $d$, making it tricky to set up an estimator for $k_\lambda$ that has a sampling complexity that does not grow with $d$ (which would make sampling practically impossible for more than a few qubits). This is the task we will turn to now. We can construct an estimator for $k_\lambda(m)$ essentially directly from its definition.

\begin{center}
\framebox{
\begin{minipage}[t]{0.95\columnwidth}
\centering
\begin{algorithm}[H]
\SetAlgoLined
\For{$l\in \{1, \ldots , L\}$}{
	Choose ${\gend}_l \in \gr $ uniformly at random\\
	Perform the \ac{RB} protocol alg.~\ref{prot:rand_bench} to obtain frequencies $\{f_i({\gend}_l)\}_{i\in I}$  from the distribution $\{p(i,{\gend}_l,m)\}_{i\in I}$\\
 	Compute $\lambda$-filter function values for the non-zero frequencies: $\{\alpha_\lambda(i,{\gend}_l)\;\;\|\;\; i\in I,\; f_i({\gend}_l)\}$\\
 }
 Compute the empirical weighted average
 \begin{equation}
 \hat{k}_\lambda(m) = \frac{1}{L} \sum_{l=1}^L N^{-1}_\lambda \alpha_\lambda(i,{\gend}_l) f_i({\gend}_l)
 \end{equation}
 \caption{An estimator for $k_\lambda(m)$}\label{prot:lambda_estimator}
\end{algorithm}
\end{minipage}
}
\end{center}

It is easy to see that the mean of this estimator is equal to the $\lambda$ filtered output data $k_\lambda(m)$. However, this does not mean that the associated estimation procedure is efficient. A priori the variance of the estimator could scale with Hilbert space dimension $d$, since the magnitude of the filter function $N_\lambda^{\-1}\alpha_\lambda$ does so in general. We can not prove that this estimator is efficient for all groups $\gr$ and POVMs $\{\Pi_i\}_{i\in I}$. We can, however, make some partial statements. In particular, we can prove that the estimator is efficient as long as the POVM $\{\Pi_i\}_{i\in I}$ is generated by a $3$-design. This is a restrictive condition, but not impossible to fulfill. We will discuss how to implement such a POVM after stating and proving the following theorem, which essentially states that under the $3$-design condition, the variance of the estimator $\hat{k}_\lambda(m)$ does not scale with the Hilbert space dimension $d$. This means that the sampling resources required by the protocol do not depend on the number of qubits in the system, making the post-processing step scalable (at least with respect to sampling). We note that this theorem gives an extremely crude bound on the variance, and the actual variance is liable to be substantially smaller. For simplicity,
we assume that there is no SPAM or gate noise, but the conclusions made here easily generalize.

\begin{theorem}[Efficient estimators]\label{thm:eff_estimator}
Consider a uniform \ac{RB} experiment of sequence length $m$, with group $\gr$, reference representation $\omega$, measurement POVM $\{\Pi_i\}_{i\in I}$ and initial state $\rho_0$, and further assume that the POVM $\{\Pi_i\}_{i\in I}$ is an (exact) $3$-design, that is $\Pi_i = \frac{d}{|I|}\dens{\chi_i}$ with states $\ket{\chi_i}$ and $\frac{1}{I}\sum_{i\in I}\dens{\chi_i}\tn{3} = \int d\psi \dens{\psi}\tn{3}$. Then for all $\lambda \in \Lambda$ the variance of the estimator $\hat{k}_\lambda(m)$ is asymptotically independent of the Hilbert space dimension $d$.
\end{theorem}
\begin{proof}
First we calculate the effect of the $3$-design condition on the normalization factor of the correlation function $\alpha(i,\cdot)$, by direct calculation we have
\begin{align}
N_\lambda &= \gsum{g}\sum_{i\in I} \alpha_\lambda(i,g)\braa{\Pi_i} \omega(g)\kett{\rho_0}\\
&= \frac{d^2}{|I|} \gsum{g}\int d\psi \braa{\psi\tn{2}}\omega(g)\tn{2} \mc{P}_\lambda\otimes\1 \kett{\rho_0\tn{2}}\\
& = \frac{d^2}{|I|}\left[\frac{1}{d^2-1}\tr[\rho_0\mc{P}_\lambda(\rho_0)] + \frac{\tr(\mc{P}_\lambda(\rho_0))\tr(\rho_0)}{d^2}\right]\\
&=\frac{1}{|I|}\left[\frac{d^2}{d^2-1}\tr[\rho_0\mc{P}_\lambda(\rho_0)] + \tr(\mc{P}_\lambda(\rho_0))\right]
\end{align}

where we have used the fact that the Haar measure is invariant under unitary action to absorb the $\omega(g)$ dependence, as well as a standard formula for the second moment of a Haar average over the unitary group, see e.g.~\cite[Proposition 37]{KlieschRoth:2020:Tutorial} or \cite{huang2020predicting} (and that $\tr(\rho_0)=1)$.
We can now calculate the variance.
We denote by $\hat{k}_\lambda(m,\gend)$ the estimator of $\sum_{i \in I} N_\lambda^{-1}\alpha(i,\gend) p(i,m,\gend ) $ for a fixed $\gend\in \gr$.
By the law of total variation we can write:
\begin{align}
\md{V}(\hat{k}_\lambda(m)) &= \gsum{\gend}\md{V}\big[\hat{k}_\lambda(m,\gend)\big] + \md{V}_{\gr}\left[\sum_{i\in I}\alpha(i,\gend) p(i,m,\gend ) \right]\\
&\leq \gsum{\gend} \sum_{i \in I } N_\lambda^{-2}\alpha(i,\gend)^2 p(i,m,{\gend}) + \gsum{\gend} \left[ \sum_{i\in I}N_\lambda^{-1}\alpha(i,\gend) p(i,m,\gend ) \right]^2,
\end{align}
by dropping the negative terms in the variances.
We will begin with calculating the second term. For this note that for all $\gend \in \gr$ (again using the invariance of the Haar measure):
\begin{align}
\sum_{i\in I}N_\lambda^{-1}\alpha(i,\gend) p(i,m,{\gend}) &= I\left[\frac{d^2}{d^2-1}\tr[\rho_0\mc{P}_\lambda(\rho_0)] + \tr(\mc{P}_\lambda(\rho_0))\right]^{-1}\\&\hspace{3em} \times \int d\psi \frac{d^2}{I} \braa{\psi\tn{2}}\big(\omega(g)\otimes (\mc{E}_{\mathrm{M}} \phi^{*m}(\gend) \mc{E}_{\mathrm{SP}} )\big)\big( \mc{P}_\lambda\otimes \1\big) \kett{\rho_0\tn{2}}\\
&= \left[\frac{d}{d^2-1}\tr[\rho_0\big(\mc{P}_\lambda\mc{E}_{\mathrm{M}} \phi^{*m}(\gend) \mc{E}_{\mathrm{SP}}\big)(\rho_0)] + \tr \big(\mc{P}_\lambda(\rho_0)\big)\right]\\
&\hspace{3em}\times \left[\frac{d^2}{d^2-1}\tr[\rho_0\mc{P}_\lambda(\rho_0)] + \tr(\mc{P}_\lambda(\rho_0))\right]^{-1},
\end{align}
where we have used the expression for $p(i,m,\gend)$ from eq.\ (\ref{eq:rb_as_conv}). {Note that this expression is asymptotically independent of the Hilbert space dimension (depending only on how well the initial state overlaps with the projector $\mc{P}_\lambda$).}
Next we discuss the first term, given by
\begin{align}
\gsum{\gend} \sum_{i \in I } N_\lambda^{-2}&\alpha(i,\gend)^2 p(i,m,\gend )\\ &= N_\lambda^{-2} \frac{d^3}{|I|^2} \gsum{\gend} \int d\psi \braa{\psi\tn{3}} \big(\omega(\gend)\tn{2}\otimes (\mc{E}_{\mathrm{M}} \phi^{*m}(\gend) \mc{E}_{\mathrm{SP}} )\big)\big(\mc{P}_\lambda\tn{2}\otimes \1\big) \kett{\rho_0\tn{3}}\\
&= N_\lambda^{-2} \frac{d^3}{|I|^2} \int d\psi \braa{\psi\tn{3}}\big(\mc{P}_\lambda\tn{2}\otimes (\mc{E}_{\mathrm{M}} \phi^{*m}(e) \mc{E}_{\mathrm{SP}}) \big)\kett{\rho_0\tn{3}}.
\end{align}
Here appears a third moment of a Haar average, which can evaluated using Weingarten calculus (see for instance equations $S35$ and $S36$  in ref.\ \cite{huang2020predicting}, ref.~\cite{KlieschRoth:2020:Tutorial} or
ref.\ \citep{ginory2019weingarten} more generally). In this particular instance, we get
\begin{align}
\int d\psi\bra{\psi}&\mc{P}_\lambda(\rho_0)\ket{\psi} \!\!\bra{\psi}\mc{P}_\lambda(\rho_0)\ket{\psi}\!\! \bra{\psi}\big(\mc{E}_{\mathrm{M}} \phi^{*m}(e) \mc{E}_{\mathrm{SP}}\big)\!(\rho_0)\!\ket{\psi}\\ &= \frac{\tr\big[\mc{P}_\lambda(\rho_0))|_{t}^2\big] + 2 \tr \big[\mc{P}_\lambda(\rho_0)|_{t}^2 \big(\mc{E}_{\mathrm{M}} \phi^{*m}(e) \mc{E}_{\mathrm{SP}}\big)\!(\rho_0)\big]}{(d+2)(d+1)d} \\&\hspace{3em}+ \frac{\tr\big[\mc{P}_\lambda(\rho_0))\tr (\mc{P}_\lambda(\rho_0)|_{t} \big(\mc{E}_{\mathrm{M}} \phi^{*m}(e) \mc{E}_{\mathrm{SP}}\big)\!(\rho_0)\big] }{d^2(d+1)} + \frac{\tr\big[\mc{P}_\lambda(\rho_0)\big]^2}{d^3} ,
\end{align}
where $A|_t = A-\tr(A) \1$ for matrices $A$.
By isolating a common $d^{-3}$ factor and plugging back in, we get
\begin{align}
\gsum{\gend} \sum_{i \in I } N_\lambda^{-2}&\alpha(i,g)^2 p(i,m,\gend )\\
&=\bigg[\frac{2\tr\big[\mc{P}_\lambda(\rho_0))|_{t}^2)\big]  \tr \big[\mc{P}_\lambda(\rho_0)|_{t}^2 \big(\mc{E}_{\mathrm{M}} \phi^{*m}(e) \mc{E}_{\mathrm{SP}}\big)\!(\rho_0)\big]}{(d+2)(d+1)d^{-2}} \\&\hspace{3em}+ \frac{\tr(\big[\mc{P}_\lambda(\rho_0)\big]\tr \big[\mc{P}_\lambda(\rho_0)|_{t} \big(\mc{E}_{\mathrm{M}} \phi^{*m}(e) \mc{E}_{\mathrm{SP}}\big)\!(\rho_0)\big] }{d^{-1}(d+1)} +\tr\big[\mc{P}_\lambda(\rho_0)\big]^2\bigg]\\
&\hspace{17em}\times\left[\frac{d^2}{d^2-1}\tr[\rho_0\mc{P}_\lambda(\rho_0)] + \tr(\mc{P}_\lambda(\rho_0))\right]^{-2},
\end{align}
which is again asymptotically independent of the Hilbert space dimension.
\end{proof}

Measurement POVMs that are proportional to $3$-designs are not very common. However, when considering a system of $q$ qubits it is possible to construct one by considering computational basis measurements conjugated by a random element of the $q$-qubit Clifford group $\mathbb{C}_q$. That is, we consider the POVM
\begin{equation}
 \{\Pi_{x,C}\} = \{ \frac{1}{|\mathbb{C}_q|} \,C\!\ket{x}\!\!\bra{x}\!C\ct \;\; \|\;\; x\in \{0,1\}^q,\;\;C\in \mathbb{C}_q \} .
\end{equation}
It is easy to see that this is a POVM
\begin{equation}
\sum_{C\in \mathbb{C}_q} \sum_{x\in\{0,1\}^q } \frac{1}{|\mathbb{C}_q|}\, C\!\ket{x}\!\!\bra{x}\!C\ct = \frac{1}{|\mathbb{C}_q|} \sum_{C\in \mathbb{C}_q} C C\ct = \mathbb{I}
\end{equation}
and it is also proportional to a $3$-design, because the multi-qubit Clifford group is a unitary $3$-design~\citep{webb2016clifford,zhu2017multiqubit}, and hence every orbit $\{ C\ket{x}\!\!\bra{x}C\ct\}_{C\in \mathbb{C}_q}$ is a state $3$-design (and thus so is the union over $x$).

We emphasize that the $3$-design condition is only a sufficient condition for a controlled variance of the estimator for the filtered output data, which works for any group $\gr$ and sub-representation $\sigma_\lambda$. For particular choices of $\gr$ and $\sigma_\lambda$ the estimator $\hat{k}_\lambda(m)$ might be efficient for other choices of the POVM $\{\Pi_i\}_{i\in I}$. It is for instance easy to see that the variance will also be controlled if the degree $d_\lambda$ of the irrep $\sigma_\lambda$ is small. This follows from the fact that the normalization factor $N_\lambda$ can be written as
\begin{align}
N_\lambda = \frac{1}{d_\lambda} \sum_{i\in I} \tr(\Pi_i \mc{P}_\lambda(\Pi_i))\tr(\rho_0\mc{P}_\lambda(\rho_0))
\end{align}
so assuming the POVM $\{\Pi_i\}_{i\in I}$ and the initial state $\rho_0$ can be chosen to have sufficient (larger than $1/d$) overlap with the sub-representation $\sigma_\lambda$ the magnitude of the inverse normalization factor $N_\lambda^{-1}$ , and hence the size of the support of the probability distribution $\{N_\lambda^{-1} \alpha(i,\gend )\}_{p(i,m,\gend)}$ is controlled by $1/d_\lambda$. Hence, if $d_\lambda$ is small, the estimator $\hat{k}_\lambda(m)$ is efficient. This follows because it is constructed by sampling from a ($O(1)$ in $d$) bounded random variable. Examples of this behaviour have been noted in the literature \citep{helsen2019new,carignan2015characterizing,flammia2019efficient}.

Alternatively, there are situations where the dimension of the representation $\sigma_\lambda$ scales with the total Hilbert space dimension $d$ but the estimator $\hat{k}_\lambda(m)$ is still efficient because the group $\gr$ under consideration is sufficiently randomizing (roughly, it spans its own $3$-design due to the randomization over the ending gate $\gend$). An example of this is the recently introduced \emph{linear cross entropy benchmarking procedure} which we will discuss in the next section.

{Finally ,we would like to add that if one re-uses the same experimental data $p(m,\gend)$ to estimate $k_\lambda(m)$ for different $\lambda$, the resulting estimates for $k_\lambda(m)$ (and consequently the associated decay rates) will be correlated. This must be taken into account when performing joint statistical inferences on estimates for several $M_\lambda$. This can of course be remedied by gathering new data for each representation label $\lambda$.

}

\subsection{Example: Linear cross-entropy benchmarking}\label{subsec:lin_cross_ent}

Recently, ref.\ \citep{arute2019quantum} has introduced a \ac{RB}-like protocol referred to as \emph{linear cross-entropy benchmarking}, in short \emph{XEB}. We will see in this section that this protocol falls into the framework of the benchmarking schemes
introduced here. In fact, it can be seen as uniform \ac{RB} with $\gr$ the full unitary group, together with a post-processing scheme that is a special case of the above filtering scheme. Let $\phi:U(2^q) \to \mc{S}_d$ be an implementation map of the unitary group, also let $\{\Pi_x\}_{x\in \{0,1\}^n}$ be the computational basis POVM, and $\rho_0 =\dens{0}$. The linear cross-entropy fidelity is now given by
\begin{equation}
\mc{F}_{\mathrm{XEB}} = d\int_{\mathrm{Haar}} \!\!\!\!dU  \sum_{x\in \{0,1\}^q} |\!\!\bra{x}\!U\!\ket{0}\!|^2 \braa{\Pi_x}\mc{E}_{\mathrm{M}} \phi(U)\mc{E}_{\mathrm{SP}}\kett{\rho_0}
\end{equation}
with $\mc{E}_{\mathrm{M}},\mc{E}_{\mathrm{SP}}$ being the usual SPAM error channels.  Setting $\alpha(x,U) = |\!\bra{x}\!U\!\ket{0}\!|^2 = \braa{\Pi_x} \omega(U) \kett{\rho_0}$ we see that $\mc{F}_{\mathrm{XEB}}$ can be interpreted as a \ac{RB} experiment of sequence length `$0$' with $\gend = U$ together with post-processing by correlation with the adjoint representation $\omega(U) = U\cdot U\ct$. Note that the dimensional factor almost precisely serves as the correct normalization factor for $\alpha(x,U)$, since
\begin{equation}
\int dU \sum_{x} |\!\bra{x}\!U\!\ket{0}\!|^4 = \frac{2d}{d+1}.
\end{equation}
We can extend this interpretation by considering the linear cross entropy of a sequence of $m$ random unitaries (this is done implicitly  in ref.\ \cite{arute2019quantum}).
This  gives
\begin{equation}
\mc{F}_{\mathrm{XEB},m} = d\int_{\mathrm{Haar}}\!\!\!\! dU_1\ldots U_m  \sum_{x\in \{0,1\}^q} |\!\!\bra{x}\!U_m\cdots U_1\!\ket{0}\!|^2 \braa{\Pi_x}\mc{E}_{\mathrm{M}} \phi(U_1\cdots U_m)\mc{E}_{\mathrm{SP}}\kett{\rho_0} .
\end{equation}
Using the invariance of the Haar measure and the linearity of the trace and the tensor product we can rewrite this as
\begin{align}
\mc{F}_{\mathrm{XEB},m} &= d\int_{\mathrm{Haar}} \!\!\!\!dU_m  \sum_{x\in \{0,1\}^q} |\!\!\bra{x}\!U_m\!\ket{0}\!|^2 \braa{\Pi_x}\mc{E}_{\mathrm{M}} \phi^{*m}(U_m)\mc{E}_{\mathrm{SP}}\kett{\rho_0}\\
 &= d\int_{\mathrm{Haar}}\!\!\!\! dU_m  \sum_{x\in \{0,1\}^q} |\!\!\bra{x}\!U_m\!\ket{0}\!|^2p(x, U_m, m)
\end{align}
with $p(x, U_m, m)$ the output probability of a regular \ac{RB} experiment. Now noting that $\omega(U)$ decomposes into the trivial representation (on the space $\{ a \kett{\1} \mid a\in \mathbb{C} \}$) and the adjoint representation (on the space $\{\;\kett{A} \mid \tr(A)=0\}$) we apply theorem \ref{thm:mother} to the above to get
\begin{equation}
\mc{F}_{\mathrm{XEB},m} = A_{\mathrm{tr}}s_{\mathrm{tr}}^m + A_{\mathrm{adj}}f_{\mathrm{adj}}^m
\end{equation}
up to a correction exponentially small in $m$, where $s_{\mathrm{tr}}$ ($f_{\mathrm{adj}}$) is the largest eigenvalue of the Fourier transform of $\phi$ evaluated at the trivial (adjoint) representation. Recall that $s_{\mathrm{tr}}=1$ if $\phi(U)$ is trace preserving for all $U$, and that we can moreover interpret $f_{\mathrm{adj}}$ as affinely related to the average fidelity (certainly in the gate independent noise setting). Hence, through theorem  \ref{thm:mother} and our general post-processing scheme the linear cross entropy benchmarking procedure inherits both the stability and interpretation of uniform \ac{RB}.

It is notable that the estimator $\hat{k}_\lambda(m)$, which in this case estimates the linear cross entropy fidelity $\mc{F}_{\mathrm{XEB},m}$ is actually efficient, in the sense of theorem \ref{thm:eff_estimator}. We can sketch an argument for this by directly estimating the variance of the estimator. For this argument we will assume gate-independent noise (i.e., $\phi(U) = \mc{A}\omega(U)$ for some completely positive $\mc{A}$). Following theorem
\ref{thm:eff_estimator}, we have
\begin{align}
\md{V}(\hat{k}_\lambda(m)) &\leq d^2\sum_{x \in \{0,1\}^q} \int_{\mathrm{Haar}}dU  |\!\!\bra{x}\!U\!\ket{0}\!|^4 \braa{\Pi_x}\mc{E}_{\mathrm{M}} \phi^{*m}(U)\mc{E}_{\mathrm{SP}}\kett{\rho_0}\\&\hspace{5em} + d^2\int_{\mathrm{Haar}}dU \left[ \sum_{x \in \{0,1\}^q} |\!\!\bra{x}\!U\!\ket{0}\!|^2 \braa{\Pi_x}\mc{E}_{\mathrm{M}} \phi^{*m}(U)\mc{E}_{\mathrm{SP}}\kett{\rho_0} \right]^2\\
&\leq d^3\max_{x \in \{0,1\}^q} \int_{\mathrm{Haar}}dU  |\!\!\bra{x}\!U\!\ket{0}\!|^4 \braa{\Pi_x}\mc{E}_{\mathrm{M}} \phi^{*m}(U)\mc{E}_{\mathrm{SP}}\kett{\rho_0} \\&\hspace{5em}+ d^4 \max_{x,x' \in \{0,1\}^q} \int_{\mathrm{Haar}}\!\!\!\!dU  |\!\!\bra{x}\!U\!\ket{0}\!|^2 |\!\!\bra{x'}\!U\!\ket{0}\!|^2\\&\hspace{10em}\!\times \braa{\Pi_x}\mc{E}_{\mathrm{M}}\phi^{*m}(U)\mc{E}_{\mathrm{SP}}\kett{\rho_0} \braa{\Pi_{x'}}\mc{E}_{\mathrm{M}} \phi^{*m}(U)\mc{E}_{\mathrm{SP}}\kett{\rho_0}.
\end{align}

Using the gate independent noise assumption and the fact hat $\omega(U)(\rho) = U\rho U\ct$, the RHS is a Haar integral of a degree-$3$ homogeneous polynomial in the entries of $U, \overline{U}$, and the second term is a Haar integral of a degree-$4$ homogeneous polynomial. The asymptotic behaviour of such integrals (in the limit of large $d$) is well known~\citep{ginory2019weingarten} and evaluates to $O(d^{-3})$ and $O(d^{-4})$, respectively. Hence, the overall variance is $O(1)$ in $d$. One could fill in the exact constants by evaluating the Haar integrals (like we did in theorem \ref{thm:eff_estimator}), but we do not pursue this here.

%===========================================================================
\section{Randomized benchmarking and average fidelity}\label{sec:fidelity_interpretation}
%===========================================================================
{
Up until now, we have treated the information extracted from RB procedures, and in particular the decay rates, as figures of merit in their own right, without establishing a direct connection to other well-know quantities such as the average gate fidelity. Indeed, this latter object is often portrayed as the conclusive result of an RB protocol.

In this section, we will provide a series of arguments to validate the interpretation of the RB parameters as standalone information, by showing that connecting RB decays to the average gate fidelity presents complications that are hard to overcome. The underlying reason for this incompatibility is due to the gauge-dependent nature of the average gate fidelity (as argued in \cite{Proctor17}) that cannot be established nor controlled under RB. More precisely, in subsection \ref{sec:counter_example_CPT_property} we provide an explicit example showing that adopting a gauge to match the average gate fidelity gives rise to a channel that is not physical. In subsection \ref{sec:fidelity_vs_RB}, we substantiate our argument with an analysis of the expression of the entanglement fidelity -- a quantity closely related to the average fidelity -- in terms of RB decay parameters and the adopted gauge. Observing this expression we conclude that \ac{RB} parameters and fidelity can be linked only if there is a close overlap between the dominant eigenvector of the ideal operator and the dominant, gauge-dependent left and right eigenvectors of its implemented version; the critical point is that ascertaining whether this requirement is met is not possible with an RB procedure.
We want to highlight that this intricacy in connecting RB to other well-established quantities does not mean RB protocols are inherently flawed, but only that the information they provide have to be regarded independently, with decay rates as the defining quantities to characterize the accuracy of experimentally implemented sets of gates.}

%----------------------------------------------------------------------------
\subsection{The depolarizing gauge and in-between noise average fidelity}\label{sec:counter_example_CPT_property}

In an attempt to resolve the apparent disconnect between fidelity and \ac{RB} decay parameters in the gate-dependent noise setting, in refs.\ \citep{wallman2018randomized} and \citep{Merkel18} proposals have been
made for the precise connection between \ac{RB} decay rates and average fidelity.
In ref.~\citep{wallman2018randomized}, it has been noted that the output data of Clifford \ac{RB} could be exactly fitted to a single exponential whose decay rates are exactly interpreted as the average fidelity of the `noise in between gates', a manifestly gauge invariant quantity.
Similarly in ref.~\citep{Merkel18}, it has been argued that the decay of Clifford \ac{RB} can be regarded as the average fidelity of the implementation w.r.t.~a particular gauge choice, namely the one in which the average implementation inverted with the reference representation is precisely a depolarizing channel.
We will show here that (1) both of these statements can be generalized to \ac{RB} with arbitrary groups, (2) both statements in fact say the exact same thing, and (3) both interpretations suffer from the same problem, namely that the channel of which the average fidelity is measured by \ac{RB} is not necessarily a CP map (i.e., physical), even if the implementation map $\phi$ is.

%\medskip

In ref.~\citep{wallman2018randomized}, the \ac{RB} decay rate is interpreted as measuring the fidelity of `the noise in between gates'.  (A general version of) this construction goes as follows. For an implementation $\phi$ of a group $\gr$, close to some reference representation $\omega = \bigoplus_{\lambda\in \Lambda}\sigma_\lambda$ we can pick the dominant eigenvectors $\mathrm{vec}(\mc{R}_\lambda)$ of the Fourier transform $\mc{F}(\phi)$ evaluated at the irreducible sub-representation $\sigma_\lambda\subset \omega$ (for now assuming no multiplicities, this easily generalizes). We can de-vectorize these eigenvectors and sum them up to create a super-operator $\mc{R}$ with the property
\begin{equation}
\gsum{g}\phi(g)\mc{R}\omega(g)\ct = \mc{R}\mathrm{Dep}
\end{equation}
where $\mathrm{Dep}$ is the generalized depolarizing channel $\mathrm{Dep}= \sum_{\lambda\in \Lambda} f_\lambda \mc{P}_\lambda$ with $f_\lambda$ the eigenvalue corresponding to $\mc{R}_\lambda$. W.l.o.g. we can assume that $\mc{R}$ is invertible (as a matrix). Note also that for any $\phi$ we can write $\phi(g) = \mc{R} \omega(g) \mc{L}(g)$ where $\mc{L}(g)$ is some implementation map (not necessarily completely positive).

With this parametrization the noise between two gates $g,g'$ (which in this parametrization only depends on $g$) is given by $\mc{L}(g)\mc{R}$.
The entanglement fidelity w.r.t.~the identity averaged over all $g \in \gr$ of this map is
\begin{equation}
\gsum{g} F_{\rm avg}(\mc{L}(g)\mc{R},\1) = \gsum{g}F_{\rm avg} (\mc{R}^{-1}\mc{R} \omega(g)\mc{L}(g)\mc{R}\omega(g)\ct,\1) = F_{\rm avg} (\mathrm{Dep},\1),
\end{equation}
where we have used the linearity and unitary invariance of the average fidelity. Note that $F_{\rm avg} (\mathrm{Dep},\1) = \frac{1}{d^2-1}\sum_{\lambda\in \Lambda} f_\lambda d_\lambda -1$ is precisely the average fidelity one would obtaining by plugging the \ac{RB} decay rates $f_\lambda$ into eq.~\eqref{eq:fid} .

On the other hand, ref.~\citep{Merkel18} connects the \ac{RB} decay rates to the average fidelity of the implementation map $\phi$ in a particular gauge, that is a particular choice of invertible super-operators such that
\begin{equation}
\gsum{g}F_{\rm avg}(S^{-1}\phi(g)S,\omega(g))  =F_{\rm avg} (\mathrm{Dep},\1).
\end{equation}
This map $\phi_{\rm dep} = S^{-1}\phi S$ is called the depolarizing gauge. According to
ref.\ \citep{Merkel18} the correct interpretation of the \ac{RB} decay rates is that they measure the fidelity of the implementation map $\phi $ in the depolarizing gauge with respect to the reference implementation $\omega$.
It turns out that the correct choice for $S$ is precisely the operator $R$ mentioned above, which can be easily seen by explicit computation
\begin{equation}
\gsum{g}F_{\rm avg}(\mc{R}^{-1}\phi(g)R,\omega(g)) = F_{\rm avg}\left(\mc{R}^{-1}\gsum{g}\phi(g)\mc{R}\omega(g)\ct, \1\right) = F_{\rm avg}\left(\mc{R}^{-1}\mc{R} \mathrm{Dep},\1\right).
\end{equation}
We can connect the above two interpretations by inserting the parametrization $\phi(g) = \mc{R}\omega(g)\mc{L}(g)$ into the expression for $\phi_{\rm dep}$ as
\begin{equation}
\phi_{\rm dep}(g) = \mc{R}^{-1}\mc{R}\omega(g)\mc{L}(g) \mc{R} = \omega(g)\mc{L}(g)\mc{R}.
\end{equation}
Hence, the depolarizing gauge is precisely the gauge in which each super-operator $\phi_{\rm dep}(g)$ is viewed as the ideal super-operator $\omega(g)$ preceded by the noise in between gates $\mc{L}(g)\mc{R}$ (in the sense of ref.\ \cite{wallman2018randomized}). Hence, these two interpretations of the \ac{RB} decay rates as corresponding to an average fidelity of `something' neatly map to each other.

A central open question in both the above constructions is whether the noise in between gates, or equivalently the noise in the implementation in the depolarizing gauge, can always be chosen to be a completely positive implementation map. This is essential if we want to consider these interpretations as actual descriptions of reality. Here we answer this question in the negative by giving an example (An adaptation of a construction given in ref.\ \citep{Proctor17}) of a point-wise CP implementation map $\phi$ where the noise in between gates (the implementation in the depolarizing gauge) is not completely positive. Let $\gr $ be the single qubit Clifford group, and consider, in the Pauli basis, the following super-operators

\begin{equation}
T(\gamma) = \begin{pmatrix} 1 & 0 & 0 & 0 \\ 0 & \sqrt{\gamma} & 0 & 0 \\ 0 & 0 & \sqrt{\gamma} & 0 \\ 1-\gamma & 0 & 0 & \gamma \end{pmatrix},\;\;\;\;\;
M_1(\alpha) = \begin{pmatrix} 1 & 0 & 0 & 0 \\ 0 & \alpha & 0 & 0 \\ 0 & 0 & 1 & 0 \\ 0 & 0 & 0 & 1\end{pmatrix},\;\;\;\;\;
M_2(\alpha) = \begin{pmatrix} 1 & 0 & 0 & 0 \\ 0 & 1 & 0 & 0 \\ 0 & 0 & 1 & 0 \\ 0 & 0 & 0 & \alpha^{-1}\end{pmatrix}.
\end{equation}
From these we can construct the implementation $\phi(g) = T(\gamma)M_1(\alpha)\omega(g)M_2(\alpha)$, with $\omega(g)(\rho) = U_g\rho U_g\ct$ the standard reference representation. It is easy to see that the transformation to the depolarizing gauge is given by $M_2(\alpha) \phi(g)M_2(\alpha)^{-1} = M_2(\alpha) T(\gamma) M_1(\alpha)\omega(g)$. Equivalently, the noise in between gates is given by $M_2(\alpha) T(\gamma) M_1(\alpha)$. The claim is now that there exists pairs $\alpha, \gamma$ such that $\phi(g)$ is completely positive for all $g\in \gr$ but $M_2(\alpha) T(\gamma) M_1(\alpha)$ is not. An easy pathological example can be obtained by setting $\gamma =0$.
In this case we have
\begin{equation}
\phi(g) = \begin{pmatrix} 1 & 0 & 0 & 0 \\ 0 & 0 & 0 & 0 \\ 0 & 0 & 0 & 0 \\ 1 & 0 & 0 & 0 \end{pmatrix},\;\;\;\;\; M_2(\alpha) T(1) M_1(\alpha) = \begin{pmatrix} 1 & 0 & 0 & 0 \\ 0 & 0 & 0 & 0 \\ 0 & 0 & 0 & 0 \\ \alpha^{-1} & 0 & 0 & 0 \end{pmatrix}.
\end{equation}
Hence, for all $\alpha<1$ the maps $\phi(g)$ are CP while the map $M_2(\alpha) T(0) M_1(\alpha)$ is not (this can be verified by using the complete positivity conditions for qubit channels from ref.\ \cite{ruskai2002analysis}). For $\gamma<1$ one can always construct interval conditions on $\alpha$ such that the same holds. Hence, the interpretations \citep{wallman2018randomized,Merkel18} both suffer from a problem, namely that in order to imagine \ac{RB} as `measuring the average fidelity' of some object, this object has to be chosen in a way that is not necessarily physical. This possibility was already indicated by both papers, but no explicit example was given. It is unclear how to resolve this problem: one could for instance try to find natural conditions on $\phi$ such that the noise in between gates, or equivalently the implementation in the depolarizing gauge, is always completely positive. Alternatively one could adopt the framework of ref.\ \cite{carignan2018randomized} where one relaxes the problem by asking for a positive gauged implementation map that has a fidelity approximately given by the \ac{RB} decay rates (with approximate meaning small relative to $1-f_\lambda$). This can be done for Clifford \ac{RB} on a single qubit~\citep{carignan2018randomized} but generalizing to higher dimensions seems difficult (although some work in this direction has been done~\cite{carignan2019polar}).

%----------------------------------------------------------------------------
%----------------------------------------------------------------------------
\subsection{Connecting average fidelity and \acl{RB} decay rates}\label{sec:fidelity_vs_RB}
In the previous subsection we showed that the depolarizing gauge does not always give rise to a CP implementation map, and hence, cannot be connected in all cases to the average fidelity of a physical process. Here we want to investigate the link between fidelity and the \ac{RB} decay parameters under a general gauge choice $S$.
We will do this using the tools of perturbation theory we have used earlier to establish theorem \ref{thm:mother}.

%-----------------------------------------------------------------------------
\subsubsection{The \acl{RB} measurement outcome}\label{sec:survival_probability}
Let us consider a special case of theorem~\ref{thm:mother} corresponding to reference representations $\omega$ that are multiplicity-free (for simplicity), and making the gauge freedom $S$ explicit.
In this situation, we can write the Fourier operator $F (\omega)$ as a direct sum of rank-1 orthogonal projections, since from eq.~\eqref{eq:FT_projection} and~\eqref{eq:FT_projection_rank} it follows that for each unitary irreducible representation $\sigma_\lambda$ of~$\gr$
\begin{equation}\label{eq:FT_1-projection}
\mc{F}(\pi)[\sigma_\lambda]
=
\begin{cases}
\ketbra{z (\sigma_\lambda)}{z (\sigma_\lambda)}	& \text{rank-1 orthogonal projection if $\pi$ and $\sigma_\lambda$ are equivalent irreducible representations,}\\
\\[0.5pt]
0 &\text{otherwise}.
\end{cases}
\end{equation}

Furthermore, we also assume that the Fourier transform $\mathcal{F}(\sigma_\lambda)$ is a diagonalizable operator. Since the set of diagonalizable matrices is dense~\cite{DenseDiagMatrices}, it is always possible to find such a diagonalizable matrix at arbitrary proximity of any given operator.
We can thus write the Fourier transform of the implementation map on the irreducible representation appearing in the decomposition of $\omega$ as the perturbation
$\wh E(\sigma_\lambda) \coloneqq \mc{F}(S\phi S^{-1}- \omega)[\sigma_\lambda]$
of the rank-1 operator $\mc{F}(\omega)[\sigma_\lambda]$,
\begin{align}
\mc F (S \phi S^{-1}) [\sigma_\lambda]
&=
\mc{F}(\omega)[\sigma_\lambda] + \mc{F}(S\phi S^{-1}- \omega)[\sigma_\lambda] \\
&=
\mc{F}(\omega)[\sigma_\lambda] +\wh E(\sigma_\lambda) \\
&=
f_{\max}(\sigma_\lambda) \ketbra{r_{\max}(\sigma_\lambda)}{\ell_{\max} (\sigma_\lambda)}
+
\sum_{j_\lambda=1}^{d_\lambda - 1} f_{j_\lambda} (\sigma_\lambda) \ketbra{r_{j_\lambda} (\sigma_\lambda)}{\ell_{j_\lambda} (\sigma_\lambda)}, \label{eq:eigval_decomposition_FT}
\end{align}
where $f_{\max}(\sigma_\lambda)$ is the largest eigenvalue of $\mc F (S \phi S^{-1}) [\sigma_\lambda]$ and $\set{f_{j_\lambda}}_{j_\lambda=1}^{d_\lambda-1}$ are the other eigenvalues. The sets of left- and right eigenvectors form a bi-orthogonal system, that is,
$\braket{\ell_{\max}(\sigma_\lambda)}{r_{\max}(\sigma_\lambda)}=\braket{\ell_{j_\lambda}(\sigma_\lambda)}{r_{j_\lambda} (\sigma_\lambda)}=1$
and
$\braket{\ell_{\max}(\sigma_\lambda)}{r_{j_\lambda} (\sigma_\lambda)} = \braket{\ell_{j_\lambda} (\sigma_\lambda)}{r_{\max}(\sigma_\lambda)} =
\braket{\ell_{j_\lambda} (\sigma_\lambda)}{r_{k_\lambda} (\sigma_\lambda)}=0$, for $j_\lambda\neq k_\lambda$. The important remark that we should make here is that this basis of eigenvectors reflects the gauge transformation $S\phi S^{-1}$.

In this scenario, we can thus write eq.~\eqref{eq:RB_outcome_simple_form} in the proof of theorem~\ref{thm:mother} for $\gend=\1$
\begin{align}
p(i,m)
&=
\sum_{\lambda \in \Irr} d_{\lambda}\braa{\mc{E}_{\mathrm{M}}(\Pi_i)}\tr_{V_\lambda}\big[\mc{F}(S \phi S^{-1})^{m+1}[\sigma_\lambda] (\overline{\sigma_\lambda}(\1)\otimes \1 )\big]\kett{\mc{E}_{\mathrm{SP}}(\rho_0)} \\
&=
\sum_{\lambda \in \Lambda} \bigg\{ d_{\lambda} \, (f_{\max} (\sigma_\lambda))^{m+1} \braa{\mc{E}_{\mathrm{M}}(\Pi_i)}\tr_{V_{\lambda}} \big[ \ketbra{r_{\max}(\sigma_\lambda)}{\ell_{\max} (\sigma_\lambda)} \big]\kett{\mc{E}_{\mathrm{SP}}(\rho_0)} \\
&+
\sum_{j_\lambda}  d_{\lambda} \, (f_{j_\lambda} (\sigma_\lambda))^{m+1} \braa{\mc{E}_{\mathrm{M}}(\Pi_i)}\tr_{V_{\lambda}} \big[ \ketbra{r_{j_\lambda} (\sigma_\lambda)}{\ell_{j_\lambda} (\sigma_\lambda)} \big] \kett{\mc{E}_{\mathrm{SP}}(\rho_0)}\bigg\} \\
&+
\sum_{\gamma \notin \Lambda} d_{\gamma}\braa{\mc{E}_{\mathrm{M}}(\Pi_i)}\tr_{V_{\sigma_\gamma}}\big[\mc{F}(S \phi S^{-1})^{m+1}[\sigma_\gamma] (\overline{\sigma}_\gamma (\1)\otimes \1 )\big]\kett{\mc{E}_{\mathrm{SP}}(\rho_0)}.
\end{align}
By eq.~\eqref{bound:eigval_perturbation}, it follows that $ f_{\max}(\sigma_\lambda)$ for each $\sigma_\lambda$ in the irreducible decomposition of $\omega$ is lower bounded by $1-\infnorm{\wh E(\sigma_\lambda)}$, while the sub-dominant eigenvalues, correspond to perturbations of the kernel of $\mc{F}(\omega)[\sigma_\lambda]$, are upper bounded by  $\infnorm{\wh E(\sigma_\lambda)}$.
Moreover, by Theorem~\ref{thm:stability} presented in Section~\ref{sec:RBdiamond}, the eigenvalues in those subspaces not related to irreducible representations appearing in decomposition are again dominated by $\infnorm{\wh E(\sigma_\lambda)}$.
Hence, we can choose $m$ large enough such that $f_{\max}^m (\lambda) \gg f_{j_\lambda}^m (\sigma_\lambda)$ for all $f_{j_\lambda} (\sigma_\lambda)$ and for each irreducible representations $\sigma_\lambda$ occurring in the decomposition of $\omega$, and such that the leakage of the perturbation in non-occurring irreducible subspaces is suppressed.

For these values of $m$, we then retrieve the formula for the power law in eq.~\eqref{eq:rb_data_decay}, but here with respect to 1-dim parameters,
\begin{equation}
p (i,m)
\approx
\sum_{\lambda \in \Lambda} (f_{\max} (\lambda) )^{m+1}\, \xi(S,\sigma_\lambda,\Pi_i,\rho_0) , \label{eq:survival_probability_model}
\end{equation}
where $\xi(S,\sigma_\lambda,\Pi_i,\rho_0)
\coloneqq
d_{\lambda} \braa{\mc{E}_{\mathrm{M}}(\Pi_i)}\tr_{V_{\lambda}} \big[ \ketbra{r_{\max}(\sigma_\lambda)}{\ell_{\max} (\sigma_\lambda)} \big]\kett{\mc{E}_{\mathrm{SP}}(\rho_0)}$.

%================================================================================
\subsubsection{Average gate fidelity and entanglement fidelity}\label{sec:entanglement_fidelity}
%================================================================================

The first \ac{RB} protocols based on the Clifford group~\citep{CliffordRBPRL,knill2008randomized} linked a single decay parameter $f$ to the average fidelity of a quantum channel $\mc R$, under the assumption of gate-independent noise, i.e., $\phi (g) = \mc R \omega(g)$. The relation is given by
\begin{equation}
	F_{\mathrm{\rm avg}}(\mc R) = f +\frac{1-f}{d} .
\end{equation}
This formula generalizes to uniform \ac{RB} with an arbitrary group $\gr$ with reference representation $\omega = \bigoplus_{\lambda\in \Lambda}\sigma_\lambda^{\oplus n_\lambda}$, again under the assumption of gate-independent noise. However, it is  more convenient to express it in terms of the \emph{entanglement fidelity}, defined as
\begin{equation}
	F_e (\mc{R})
	\coloneqq
	\braa{\Psi} \big(\1 \otimes \mc{R} \big)\kett{\Psi}
	=
	\frac{1}{d^2} \tr(\mc{R}),
\end{equation}
where the trace is taken over the super-operators, and related to the average gate fidelity by
\begin{equation}
	F_{\mathrm{\rm avg}}(\mc{R}) = \frac{ d F_e (\mc{R}) +1}{d+1}.
\end{equation}
In particular we have (first formally written down in ref.\ \cite{francca2018approximate})
\begin{equation}\label{eq:fid}
	F_{\mathrm{\rm avg}}(\mc R) =\frac{1}{d^2}\sum_{\lambda\in \Lambda}d_\lambda \tr(M_\lambda)
\end{equation}
with $M_{\lambda}$ again an $n_\lambda\times n_\lambda$ matrix.

The connection between the \ac{RB} decay rates and the fidelity has been challenged in ref.~\citep{Proctor17}, where it has been argued that the average fidelity and the output of \ac{RB} are not related in a unique way. In doing so they introduced the concept of gauge freedom into the \ac{RB} literature.

In the context of \ac{RB}, gauge freedom  is the observation that two implementation maps $\phi$ and $\phi'$ give rise to the same \ac{RB} output data $p(m)$ if they are related by a similarity transformation $S$, i.e., $\phi'= S\phi S^{-1}$. However, the average fidelity of these implementation maps (relative to some reference implementation) will generally differ. Note that this an issue even with the assumption of gate-independent noise, however,
in this case there is a `canonical' choice of gauge for which the \ac{RB} decay rates and the fidelity are related. In the gate-dependent noise scenario there is no such obvious gauge choice. The rest of this section will be concerned with this question.

The entanglement fidelity -- averaged over $\gr$ -- can be expressed in terms of Fourier transforms (as has first been noted in ref.~\cite{Merkel18}). Indeed,
we have
\begin{align}
\mathbb{E}_g\, F_e (S \phi  S^{-1} (g),\omega(g))
&=
\mathbb{E}_g\, F_e (\omega^\dagger(g) S \phi  S^{-1} (g))\\
&=
\frac {1}{d^2}\, \mathbb{E}_g \, \Tr[\omega^\dagger(g) S \phi  S^{-1} (g) ] \\
&=
\frac {1}{d^2}\, \sum_{\lambda \in \Irr} d_{\lambda} \Tr \big[ (\mc{F}(\omega)[{\sigma_\lambda}] )^\dagger \, \mc{F}(S \phi S^{-1})[{\sigma_\lambda}] \big] ,
\end{align}
where we have used the second Parseval identity~\eqref{Parseval_identity}.

\medskip

At this point we can again use of the property in eq.~\eqref{eq:FT_1-projection} for $\mc{F}(\omega)[\sigma_\lambda]$ and the re-formulation in eq.~\eqref{eq:eigval_decomposition_FT} for $\mc{F}(S \phi S^{-1})[\sigma_\lambda]$ and write
\begin{align}
\mathbb{E}_g\, F_e (S \phi  S^{-1} (g),\omega(g))
&=
\frac{1}{d^2} \sum_{\lambda \in \Lambda} d_{\lambda}  \Tr \Big[\ketbra{z (\sigma_\lambda)}{z (\sigma_\lambda)} \big( f_{\max}(\sigma_\lambda) \ketbra{r_{\max}(\sigma_\lambda)}{\ell_{\max} (\sigma_\lambda)} \\
&+
\sum_{j_\lambda} f_{j_\lambda} (\sigma_\lambda) \ketbra{r_{j_\lambda} (\sigma_\lambda)}{\ell_{j_\lambda} (\sigma_\lambda)} \big) \Big] \\
&=
\frac{1}{d^2} \sum_{\lambda \in \Lambda} d_{\sigma_\lambda} \, f_{\max}(\sigma_\lambda) \braket{z (\sigma_\lambda) }{ r_{\max} (\sigma_\lambda) } \braket{ \ell_{\max} (\sigma_\lambda)}{z (\sigma_\lambda)}
+
\alpha_{\mathrm{Res}} \label{eq:ent_fidelity_model}  ,
\end{align}
where we have defined the residuum term
\begin{equation}\label{def_residuum_alpha}
 \alpha_{\mathrm{Res}}  \coloneqq\frac{1}{d^2} \sum_{\lambda \in \Lambda} d_{\sigma_\lambda} \sum_{j_\lambda} f_{j_\lambda} (\sigma_\lambda)  \braket{z (\sigma_\lambda) }{ r_{j_\lambda} (\sigma_\lambda) } \braket{ \ell_{j_\lambda} (\sigma_\lambda)}{z (\sigma_\lambda)} .
\end{equation}
This establishes a connection between the decay parameters $f_{\max} (\sigma_\lambda)$ retrieved from eq.~\eqref{eq:survival_probability_model} and the entanglement fidelity as expressed in eq.~\eqref{eq:ent_fidelity_model}.

We observe that this connection is complicated by two factors. Firstly, it depends on the gauge-dependent overlap $\braket{z (\sigma_\lambda)}{r_{\max}(\sigma_\lambda)}\braket{\ell_{\max} (\sigma_\lambda)}{z (\sigma_\lambda)}$ between the rank-1 projection and the perturbed dominant eigenvectors --~a quantity that we cannot retrieve from \ac{RB} data~-- which might deviate significantly from 1 depending on the gauge choice.  Secondly, the residuum $\alpha_{\mathrm{Res}}$ may be large, constituting a non-negligible part of the entanglement fidelity. The rest of the section with be concerned with analyzing these gauge dependent connective factors.

We begin by deriving a bound on $\alpha_{\mathrm{Res}}$, showing that this term is small, more precisely, of third order in the gauge-dependent perturbation term $\wh E(\sigma_\lambda)$.
For this, we use Corollary~\ref{thm:eigenvector_perturbation}, where in this specific case $a_1 = 1$ and $A_2 = 0_{(d^2-1),(d^2-1)}$ and where
\begin{equation}
	Q_{z(\sigma_\lambda)} \coloneqq X_2X_2^\dagger =\1 - \ketbra{z (\sigma_\lambda)}{z (\sigma_\lambda)}
\end{equation}
is the orthogonal complement of the projection $\ketbra{z (\sigma_\lambda)}{z (\sigma_\lambda)}$.
Then, the relations between unperturbed and perturbed dominant eigenvectors is given by
\begin{align}
\ket{r_{\max}(\sigma_\lambda)} &= \ket{z (\sigma_\lambda)} + Q_{z(\sigma_\lambda)} \wh E(\sigma_\lambda) \ket{z (\sigma_\lambda)}  + O(\infnorm{\wh E (\sigma_\lambda)}^2) \label{perturbed_r_eigvec} , \\
\bra{\ell_{\max}(\sigma_\lambda)} &= \bra{z (\sigma_\lambda)} + \bra{z (\sigma_\lambda)}  \wh E(\sigma_\lambda) Q_{z(\sigma_\lambda)} + O(\infnorm{\wh E (\sigma_\lambda)}^2) \label{perturbed_ell_eigvec} .
\end{align}
Furthermore, let us define the matrix
\begin{equation}\label{def:K_matrix}
\wh K (\sigma_\lambda)
\coloneqq
\sum_{j_\lambda} f_{j_\lambda} (\sigma_\lambda) \ketbra{r_{j_\lambda} (\sigma_\lambda)}{\ell_{j_\lambda} (\sigma_\lambda)},
\end{equation}
where we have
\begin{equation}
\wh K (\sigma_\lambda)\ket{r_{\max} (\sigma_\lambda)} = \ket{0}
\qquad \text{and} \qquad
\bra{\ell_{\max}(\sigma_\lambda)}\wh K (\sigma_\lambda) = \bra{0},
\end{equation}
and the bound on the 2-norm
\begin{align}
\infnorm{\wh K(\sigma_\lambda) }
&=
\infnorm{\mc{F}(S \phi S^{-1})[\sigma_\lambda] - f_{\max}(\sigma_\lambda) \ketbra{r_{\max}(\sigma_\lambda)}{\ell_{\max} (\sigma_\lambda)}} \\
&=
\infnorm{\ketbra{z(\sigma_\lambda)}{z(\sigma_\lambda)}+ \wh E (\sigma_\lambda)  - f_{\max}(\sigma_\lambda) \ketbra{r_{\max}(\sigma_\lambda)}{\ell_{\max} (\sigma_\lambda)}} \\
&=
\infnorm{\ketbra{z(\sigma_\lambda)}{z(\sigma_\lambda)}+ \wh E (\sigma_\lambda)  - f_{\max}(\sigma_\lambda)
	\big(
	\ketbra{z(\sigma_\lambda)}{z(\sigma_\lambda)} + Q_{z(\sigma_\lambda)} \wh E (\sigma_\lambda) \ketbra{z(\sigma_\lambda)}{z(\sigma_\lambda)} \\
	&+\ketbra{z(\sigma_\lambda)}{z(\sigma_\lambda)} \wh E (\sigma_\lambda) Q_{z(\sigma_\lambda)}
	+  Q_{z(\sigma_\lambda)} \wh E (\sigma_\lambda) \ketbra{z(\sigma_\lambda)}{z(\sigma_\lambda)}\wh E (\sigma_\lambda) Q_{z(\sigma_\lambda)}
	\big)
}\\
&\leq
\abs{1-f_{\max}(\sigma_\lambda)} \infnorm{\ketbra{z(\sigma_\lambda)}{z(\sigma_\lambda)}} + O(\infnorm{\wh E (\sigma_\lambda)}) \\
&\leq
O(\infnorm{\wh E (\sigma_\lambda)}),
\end{align}
where we have used the fact that $\abs{1-f_{\max}(\sigma_\lambda)} \leq \infnorm{ \wh E (\sigma_\lambda)}$.

Now, inserting eq.~\eqref{perturbed_r_eigvec}-\eqref{def:K_matrix} into eq~\eqref{def_residuum_alpha} and using the Cauchy-Schwarz inequality, we obtain the following bound on the residuum,
\begin{align}
\abs{\alpha_{\mathrm{Res}} }
&=
\dsum \big\lvert \bra{\ell_{\max}(\sigma_\lambda)}\big( \1- \wh E(\sigma_\lambda) Q_{z(\sigma_\lambda)}+ O(\infnorm{\wh E (\sigma_\lambda)}^2) \, \wh K (\sigma_\lambda) \, \big(\1- Q_{z(\sigma_\lambda)} \wh E(\sigma_\lambda)+ O(\infnorm{\wh E (\sigma_\lambda)}^2) \big) \ket{r_{\max}(\sigma_\lambda)} \big\rvert  \\
&\leq
\dsum \big\lvert \bra{\ell_{\max}(\sigma_\lambda) }  \wh E(\sigma_\lambda)Q_{z(\sigma_\lambda)} \, \wh K (\sigma_\lambda) \, Q_{z(\sigma_\lambda)}  \wh E(\sigma_\lambda) \ket{r_{\max}(\sigma_\lambda)} \big\rvert \\
&+
\dsum  O(\infnorm{\wh E(\sigma_\lambda)}^3) \infnorm{ \wh K(\sigma_\lambda)} \norm{\ell_{\max} (\sigma_\lambda)} \norm{r_{\max}(\sigma_\lambda)} +  O(\infnorm{\wh E(\sigma_\lambda)}^4)\\
&\leq
\dsum  O (\infnorm{\wh E(\sigma_\lambda)}^3 ) \norm{\ell_{\max} (\sigma_\lambda)} \norm{r_{\max}(\sigma_\lambda)}  + O(\infnorm{\wh E(\sigma_\lambda)}^4).
\end{align}
This bound for $\alpha_{\mathrm{Res}}$ has a significant implication: it means that the residuum will not cover the leading term in eq.~\eqref{eq:ent_fidelity_model} if the latter is~$\Omega(\infnorm{\wh E (\sigma_\lambda)}^2)$, for all gauge choices $S$ that yield $\norm{\ell_{\max} (\sigma_\lambda)}\cdot \norm{r_{\max}(\sigma_\lambda)}$ smaller than $1/\infnorm{\wh E (\sigma_\lambda)}$.

Note that it is important to compare $\alpha_{\mathrm{Res}}$ to the difference between 1 (the value of the entanglement fidelity of a perfect implementation) and the dominant eigenvalues in eq.~\eqref{eq:ent_fidelity_model}. This distance is indeed what \ac{RB} protocols are designed to detect, and in order for the connection between fidelity and decay rates to be meaningful we require $\alpha_{\mathrm{Res}}$ to be negligible in comparison. To analyze this further, we first write
\begin{equation}
	\Delta_{\mathrm{max}}
	\coloneqq
	\dsum  \big\lvert 1- f_{\max}(\sigma_\lambda) \braket{z (\sigma_\lambda) }{ r_{\max} (\sigma_\lambda) } \braket{ \ell_{\max} (\sigma_\lambda)}{z (\sigma_\lambda)} \rvert \big\rvert,
\end{equation}
and we calculate deviation of the absolute of the overlap from 1, which is remarkably only in second order in perturbation,
\begin{align}
&\braket{z (\sigma_\lambda) }{ r_{\max} (\sigma_\lambda) } \braket{ \ell_{\max} (\sigma_\lambda)}{z (\sigma_\lambda)}
=\\
&=\bigg\lvert \left(\bra{z (\sigma_\lambda)}\1 + Q_{z(\sigma_\lambda)} \wh E (\sigma_\lambda) +O( \infnorm{\wh E(\sigma_\lambda)}^2)\ket{z (\sigma_\lambda)} \right)
\left(\bra{z (\sigma_\lambda)}\1 +  \wh E(\sigma_\lambda) Q_{z(\sigma_\lambda)} +O(\infnorm{\wh E(\sigma_\lambda)}^2)\ket{z (\sigma_\lambda)} \right) \bigg\rvert\\
&\leq
\Big\vert 1 + O(\infnorm{\wh E(\sigma_\lambda)}^2) \Big\vert^2 \\
&\leq
1 +  O(\infnorm{\wh E(\sigma_\lambda)}^2) .
\end{align}
This bound on the overlap, together with the one on the residuum, implies that the parameters $f_{\max}(\sigma_\lambda)$ obtained from the fitting of the \ac{RB} model in eq.~\eqref{eq:survival_probability_model} yield a meaningful characterization of the fidelity on the condition when they are $\Omega(\infnorm{\wh E (\sigma_\lambda)})$.

Having derived a bound on the residuum we can consider eq.~\eqref{eq:ent_fidelity_model} in different regimes (always assuming small perturbations, i.e., $\infnorm{\wh E (\sigma_\lambda)} \ll 1$).
In the first regime we make the assumption
\begin{equation}
	\Omega(\infnorm{\wh E (\sigma_\lambda)}^2) = \abs{1-f_{\max}(\sigma_\lambda)} \gg \onevecOLcases,
\end{equation}
corresponding to the situation where the parameters $\set{f_{\max}(\sigma_\lambda)}_{\lambda \in \Lambda}$ are more sensitive to the perturbation than the overlap of the dominant eigenvectors.
 As we mentioned before, this is indeed the regime where \ac{RB} provides a meaningful estimation of the fidelity. Indeed, we have
\begin{align}
	\Delta_{\mathrm{max}}
	&\geq
	\dsum \bigg\{ \abs{f_{\max}(\sigma_\lambda)-1} \cdot \vecOL \\
	&- \onevecOL \bigg\} \\
	&\geq
	\dsum \bigg\{ \abs{f_{\max}(\sigma_\lambda)-1} \, \Big(1-O(\infnorm{\wh E (\sigma_\lambda)}^2) \Big)
	- \onevecOL  \bigg\} \\
	&=
	\dsum  \Omega(\infnorm{\wh E (\sigma_\lambda)}^2)  \\
	&\gg
	\abs{\alpha_{\mathrm{Res}} }.
\end{align}
In a second regime we can assume the converse, namely that
\begin{equation}
\abs{1-f_{\max}(\sigma_\lambda)} \ll \onevecOLcases = \Theta (\infnorm{\wh E (\sigma_\lambda)}^2)
\end{equation}
holds true.
This case is analogous, since we now have
\begin{align}
\Delta_{\mathrm{max}}
&\geq
\dsum \bigg\{ \onevecOL \\
&- \abs{f_{\max}(\sigma_\lambda)-1} \cdot \vecOL \bigg\} \\
&\geq
\dsum \bigg\{ \onevecOL - \abs{f_{\max}(\sigma_\lambda)-1} \Big(1+O(\infnorm{\wh E (\sigma_\lambda)}^2) \Big) \bigg\} \\
	&=
\dsum  \Theta(\infnorm{\wh E (\sigma_\lambda)}^2)  \\
&\gg
\abs{\alpha_{\mathrm{Res}} }.
\end{align}
This situation is, however, problematic, since \ac{RB} gives us no information in this regime about a significant quantity: the vector overlap $\braket{z (\sigma_\lambda) }{ r_{\max} (\sigma_\lambda) } \braket{ \ell_{\max} (\sigma_\lambda)}{z (\sigma_\lambda)}$.

The last regime we consider is when $\abs{1-f_{\max}(\sigma_\lambda)}$ is close to the deviation of the vector overlap from 1, that is,
\begin{equation}
\abs{1-f_{\max}(\sigma_\lambda)} \approx \onevecOLcases ,
\end{equation}
which is troublesome not only for the fact that we cannot retrieve the overlap but also because in this case $\Delta_{\mathrm{max}}$ may be of the same magnitude or smaller than $\abs{\alpha_{\mathrm{Res}} }$. Indeed, in this regime the residuum can then play a significant role in the characterization of the average gate fidelity.

The conclusion we draw from this analysis is that the overlap $\braket{z (\sigma_\lambda) }{ r_{\max} (\sigma_\lambda) } \braket{ \ell_{\max} (\sigma_\lambda)}{z (\sigma_\lambda)}$ is the key factor to consider when relating \ac{RB} decays to the fidelity. This overlap must be sufficiently close to 1 under the adopted gauge relative to the difference $\abs{1-f_{\max}(\sigma_\lambda)}$.

Finally, we wish to relate $\{\infnorm{\wh E(\sigma_\lambda)} \}_\lambda$ to a promise on a physical quantity related to the perturbation of the ideal gate implementation $\omega$. We recall that $\wh E(\sigma_\lambda) = \mc{F}(S\phi S^{-1}- \omega)[\sigma_\lambda]$ and consider that $\infnorm{\cdot} \leq \Fnorm{\cdot} $ such that
\begin{eqnarray}\label{eq:fourier_vs_physical_space}
\sum_{\lambda \in \Irr} d_\lambda \infnorm{\wh E(\sigma_\lambda)}^2
&=&
\sum_{\lambda \in \Irr} d_\lambda \infnorm{\mc{F}(S\phi S^{-1}- \omega)[\sigma_\lambda]}^2 \\
&\leq&
\sum_{\lambda \in \Irr} d_\lambda \Fnorm{\mc{F}(S\phi S^{-1}- \omega)[\sigma_\lambda]}^2 \\
&=&
\EE_G \Fnorm{S\phi S^{-1}(g)- \omega(g)}^2 ,
\end{eqnarray}
where we have applied Parseval's identity. Note, however, that the LHS of this expression runs over all irreducible representations of $\gr$ and not the only ones decomposing $\omega$.

\section{Randomized benchmarking under diamond norm and fidelity constraints}\label{sec:RBdiamond}

In theorem \ref{thm:mother}, we have argued that \acf{RB} output data associated with an implementation of a group $\gr$ could be approximated as a sum of (matrix) exponentials provided the implementation map $\phi$ was close to a reference representation $\omega$ w.r.t.~to the diamond norm (averaged over all group elements). Here we will argue that this is a natural condition to demand in the context of \ac{RB}.
In particular we will show that this condition is stable, in the sense that it is impossible to be close (in the sense of eq.~\eqref{eq:norm_assump}) to two inequivalent representations at once, and, moreover, we show that this requirement cannot be replaced with a weaker one involving the average fidelity, resolving an open question in ref.~\cite{Merkel18}.

\subsection{Stability of representations under diamond norm}
First, we prove that `closeness to a representation' is a stable concept, that is, it is impossible to be close to two representations at once (in a suitable sense).

\begin{theorem}[Stability of representations]\label{thm:stability}
	Let $\phi$ be an implementation map of a group $\gr$ taking values in $\mc{S}_d$ such that
	\begin{align}
		\gsum{g}\dnorm{\1 - \phi(g)\phi(g^{-1})}&\leq \delta
	\end{align}
	 and let $\omega, \omega' $ be representations of $G$ on $V_n, V_{n'}$ with embedding maps $\mc{L}:V_n \rightarrow V_{n}, \mc{L}':V_{n'} \rightarrow V_{n'}$ and $\mc{R}:V_n \rightarrow V_n, \mc{R}':V_{n'} \rightarrow V_{n'}$
  such that
	 \begin{align}
 		\gsum{g}\dnorm{\phi(g) - \mc{R}\omega(g)\mc{L}} &\leq \epsilon , \\
 		\gsum{g}\dnorm{\phi(g) - \mc{R}'\omega'(g)\mc{L}'}&\leq \epsilon'.
 	\end{align}
	Moreover, assume that there exists $K$ such that $\dnorm{\mc{R}\omega(g)\mc{L}}\leq K, \dnorm{\mc{R}'\omega'(g)\mc{L}'}\leq K$ for all $g\in \gr$. If the inequality $K(\epsilon + \epsilon') + 3\delta + 2\epsilon + \epsilon^2 <1$ holds then the representations $\omega, \omega'$ are equivalent on a subspace of dimension at least $d^2$.
\end{theorem}

\begin{proof}

	Consider the map $\mc{L}\mc{R'}: V_{n'}\rightarrow V_{n}$, as well as its twirled version
	\begin{align}
	T = \gsum{g} \omega(g)\mc{L}\mc{R'}\omega'(g)^\dagger.
	\end{align}
	We would like to argue that $T$ is a map of rank at least $d^2$, as then we can decide the theorem by application of Schur's lemma. To do this, consider the distance to the identity of the natural pullback of $T$ to $\mc{S}_d$, namely $\mc{R}T\mc{L'}$. We can calculate
	\begin{align}
		\dnorm{\1 -\mc{R}T\mc{L'} } &\leq \dnorm{\1 -\gsum{g}\mc{R}\omega(g)\mc{L}\mc{R}\omega(g)^\dagger\mc{L} } + \dnorm{ \gsum{g}\mc{R}\omega(g)\mc{L}\mc{R}\omega(g)^\dagger\mc{L}-\mc{R}T\mc{L'} }\\
																& \leq \gsum{g}\dnorm{\1 -\mc{R}\omega(g)\mc{L}\mc{R}\rho(g)^\dagger\mc{L} } +
																\dnorm{ \mc{R}\omega(g)\mc{L}\mc{R}\omega(g)^\dagger\mc{L} - \mc{R}\rho(g)\mc{L}\mc{R}'\omega'(g)^\dagger\mc{L}'} .
\end{align}
We will upper bound these two terms separately. For the first term,
consider
\begin{align}
	\gsum{g}&\dnorm{\1 -\mc{R}\omega(g)\mc{L}\mc{R}\omega(g)^\dagger\mc{L} } \\
	&\hspace{1em}\leq
	\gsum{g} \dnorm{\1 - \phi(g) \phi(g^{-1})} + \dnorm{ \1  - \phi(g) \mc{R}\omega(g^{-1})\mc{L}}\notag\\
	&\hspace{5em}+ \dnorm{ \1  -  \mc{R}\omega(g)\mc{L}\phi(g^{-1})}+ \dnorm{\big(\phi(g)  -  \mc{R}\omega(g)\mc{L}\big)\big(\phi(g^{-1})  -  \mc{R}\omega(g^{-1})\mc{L}\big)}\\
	&\hspace{1em}\leq
	\delta + \epsilon^2 + \gsum{g}\dnorm{\1 - \phi(g)\phi(g^{-1})} + \dnorm{\phi(g) \big[\phi(g^{-1}) - \mc{R}\omega(g^{-1})\mc{L}\big]} \\
	&\hspace{14em}+ \dnorm{\1 - \phi(g)\phi(g^{-1})} + \dnorm{ \big[\phi(g) - \mc{R}\omega(g)\mc{L}\big]\phi(g^{-1})}\\
	&\hspace{1em}\leq  3\delta + 2\epsilon + \epsilon^2,
	\end{align}
where we have exploited the sub-multiplicativity of the diamond norm and the fact that $\dnorm{\phi(g)}=1$ for all $g\in \gr$. Similarly, for the second term we get
\begin{align}
\gsum{g}\dnorm{ \mc{R}\omega(g)\mc{L}\mc{R}\omega(g)^\dagger\mc{L} - \mc{R}\omega(g)\mc{L}\mc{R}'\omega'(g)^\dagger\mc{L}'} &= \gsum{g}\dnorm{ \mc{R}\omega(g)\mc{L}[\mc{R}\omega(g)^\dagger\mc{L} - \mc{R}'\omega'(g)^\dagger\mc{L}']}\\
&\leq\gsum{g}\dnorm{\mc{R}\omega(g)\mc{L}}\dnorm{ \mc{R}\omega(g)^\dagger\mc{L} - \mc{R}'\omega'(g)^\dagger\mc{L}'}\\
&\leq K(\epsilon+ \epsilon') .
\end{align}
Combining all of this we get
\begin{equation}
\dnorm{\1 -\mc{R}T\mc{L'} } \leq (K+ \epsilon)\epsilon' + 3\delta + (2+K)\epsilon <1,
\end{equation}
by the assumptions of the theorem. Now assume that $T$ has an image of dimension strictly less than $d^2$. This means there exists a Hermitian $X\in \mc{M}_d$ such that $\mc{R}T\mc{L'}(X)=0$. But this implies that
\begin{equation}
\dnorm{\1 -\mc{R}T\mc{L'} }\geq \frac{\norm{X- \mc{R}T\mc{L'}(X)}_1}{\norm{X}_1} =1,
\end{equation}
which is a contradiction. Hence, the rank of $T$ is at least $d^2$. Since $T$ by construction commutes with the representations $\omega, \omega'$ we can decide that there exists a representation $\omega''$ of degree at least $d^2$ which is a sub-representation of both $\omega$ and $\omega'$ and moreover that both $\mc{R}\omega''\mc{L}(g)$ and $\mc{R}'\omega''\mc{L}'(g)$ are of rank at least $d^2$ for all $g\in \gr$.
\end{proof}

Next, we state a complementary theorem, saying that closeness to a representation is a concept stable under perturbations of the implementation. This is just a trivial consequence of the triangle inequality.

\begin{theorem}[Stability of the closeness to a representation]
	Let $\phi, \phi'$ be implementations of a group $\gr$ on the super-operators $\mc{S}_d$ such that
	\begin{align}
		\gsum{g}\dnorm{\phi(g)-\phi'(g)} &\leq \delta
	\end{align}
	 and let $\omega $ be a representation of $\gr$ on $V_n$ with associated maps $\mc{L}:\mc{S}_d \rightarrow V_{n}$ and $\mc{R}:V_n \rightarrow \mc{S}_d$ such that
	\begin{align}
		\gsum{g}\dnorm{\phi(g) - \mc{R}\omega(g)\mc{L}} &\leq \epsilon\\
	\end{align}
	then
	\begin{align}
		\gsum{g}\dnorm{\mc{R}\omega(g)\mc{L}-\phi'(g)} &\leq \delta + \epsilon.
	\end{align}
	\end{theorem}

\subsection{Randomized benchmarking under fidelity constraints}
In this subsection, we argue that the condition eq.\ (\ref{eq:norm_assump}) is in some sense necessary for the correct behaviour of \ac{RB}, in the sense that it can not be replaced with a natural weaker condition.
Given the worst-case nature of the diamond norm eq.\ (\ref{eq:norm_assump}) is rather restrictive, and one might wonder if it is possible to replace this diamond norm constraint with a more congenial constraint based on the average fidelity. That is, one can imagine replacing eq.\ (\ref{eq:norm_assump})
with a constraint of the form
\begin{equation}
\gsum{g} F_{\mathrm{\rm avg}}(\phi(g), \omega(g)) \geq 1-\delta'
\end{equation}
for some $\delta'>0$. Indeed, this is the assumption made  in ref.\ \cite{Merkel18} to prove a version of theorem\ \ref{thm:mother} for the Clifford group. Here, it has been noted that in order to guarantee correct behaviour the constant $\delta'$ must be chosen inversely proportional to the Hilbert space dimension $(\delta'
\sim 1/d)$. It has been speculated that this dimensional scaling could perhaps be
an artifact of the proof techniques used.

We will argue that this scaling is in fact real, by providing an explicit family (inspired by example $8.1$ in ref. \cite{wolf2012quantum}) of examples of implementations $\phi_L$ (where $L$ is an integer independent of $d$) of a group $\gr$ with
\begin{equation}
\gsum{g} F_{\mathrm{\rm avg}}(\phi(g), \omega(g)) \geq 1 - \frac{2L}{d}
\end{equation}
relative to a reference implementation $\omega$ but with associated \ac{RB} output data that is not even qualitatively of the form
eq.\ (\ref{eq:rb_data_decay}). In fact, by choosing $L$ large (but constant in $d$) we can obtain almost arbitrary non-exponential behaviour in the \ac{RB} output data associated with $\phi_L$.

\begin{example}[Real scaling.]
Choose $\gr$ to be the $q$-qubit Clifford group with standard reference implementation $\omega(g) = U_g \cdot U_g\ct$. Now let $\Lambda_L^\mu$ be a super-operator indexed by an integer $L$ and a real number $0\leq\mu\leq 1$, defined by its action on the basis matrices $\ket{i}\!\!\bra{j}$ as
\begin{equation}
\Lambda_L^\mu(\ket{i}\!\bra{j}) = \sum_{k = 1}^d \delta_{i,j}[S^\mu_L]_{i,k} \ket{k}\!\!\bra{k}
\end{equation}
with $S^\mu$ a $d\times d$ stochastic matrix of the form
\begin{equation}
[S^\mu_L] = \begin{cases}  \mu &\text{  if  } i=j \leq L-1\\
1&\text{  if  } i=j \geq L\\
1-\mu &\text{  if  } i=j-1 \leq L\\
0&\text{  otherwise}.
\end{cases}
\end{equation}
For convenience we write $\Lambda$ for $\Lambda_L^\mu$ in the following.
It is easy to see that $\Lambda$ is a quantum channel and moreover that if $i,j\leq L$ then $\Lambda(\ket{i}\!\!\bra{j}) \in \mathrm{Span}\{\ket{i'}\!\!\bra{j'} \;\;\|\;\;i','j \leq L\}$.

Consider now the following implementation map defined by its action on $X\in \mc{M}_{d}$,
\begin{equation}
\phi(g)(X) = \Lambda(P_L XP _L) + U_g(I-P_L) X (I-P_L)U_g\ct
\end{equation}
where $P_L$ is the projection onto the space $\mathrm{Span}\{\ket{i}\;\;\|\;\;i\leq L\}$. This map can be seen as checking whether a state is in the support of $P_L$  (though a measurement) and then applying $\Lambda$ or $U_g$ depending on the outcome. We can calculate the average fidelity $F_{\mathrm{\rm avg}}(\phi(g),\omega(g))$ directly
as
\begin{align}
F_{\mathrm{\rm avg}}(\phi(g),\omega(g)) &= \int\! d\psi \tr\big[\phi(g)(\!\,\dens{\psi})\omega(g)\ct(\!\,\dens{\psi}) \big]\\
&= \int\! d\psi \tr\big[U_g\!\dens{\psi}\!U_g\ct \Lambda( P_L\! \dens{\psi}\!P_L)\big] + \int\! d\psi \tr\big[\! \dens{\psi}\!(I\!-\!P_L)\! \dens{\psi}\!(I\!-\!P_L)\big]\\
&=\int\! d\psi \tr\big[U_g\dens{\psi}U_g\ct \Lambda( P_L \!\dens{\psi}\!P_L)\big] + \int\! d\psi \,(1 \!-\! 2\bra{\psi}\!P_L\! \ket{\psi}  \!+\! (\bra{\psi}\!P_L\! \ket{\psi})^2)\\
&\leq 1 -2  \int\! d\psi \bra{\psi}\!P_L \!\ket{\psi}\\
&\leq 1- \frac{2L}{d},
\end{align}
where we have made use of the fact that $\Lambda( P_L \dens{\psi}P_L)\geq 0$, since $\Lambda$ is CP. Note that for constant $L$ we can make the fidelity arbitrarily high by choosing $d=2^q$ large enough. Now consider \ac{RB} with input state $\rho = \dens{1}$ and measurement POVM $\{\dens{1}+ \dens{L}, \1 -\dens{1}-\dens{L}\}$ and implementation map $\phi_L$ as defined above. The \ac{RB} probability for the POVM element $\dens{1} + \dens{L}$ is going to be (setting $\gend=e$ and assuming no SPAM errors)
\begin{equation}
p(\dens{1} + \dens{L},m) = \tr((\dens{1} + \dens{L})\phi_L^{*m}(\dens{1})) .
\end{equation}
Note that since $P_L\dens{1} = \dens{1} P_L$ we have that $\phi_L(g)(\dens{1}) = (\Lambda^\mu_L)^m (\dens{1})$ for all $g$. From this it follows that
\begin{equation}
p(\dens{1} + \dens{L},m) = \tr((\dens{1} + \dens{L})(\Lambda^\mu_L)^m(\dens{1}))  = [{S_L^\mu}^m]_{1, L}+[{S_L^\mu}^m]_{1,1}.
\end{equation}
This data shows curious behaviour. For small sequence lengths we have $p(\dens{1} + \dens{L},m) \approx \mu^m$, but with increasing sequence length we observe wildly non-exponential behaviour.
\end{example}

\begin{acronym}
\acro{RB}[RB]{randomized benchmarking}
\end{acronym}

\section{Conclusions}

\newtext{In this work, we have introduced a comprehensive theory of
\ac{RB}. As such, it goes beyond a mere classification of known
protocols (a task that we also hope to achieve). But at the same time, it
provides a deeper understanding, a more precise formulation and interpretation
of what the data acquired in \ac{RB} means,
actionable advice to experimentalists and theoretical practitioners and
a conceptual platform from which new schemes can be derived. Specifically,
we show how \ac{RB} gives rise to exponential decays under broad classes of
Markovian circumstances, show -- importantly in practical contexts -- in what sense
\ac{RB} is robust to deviations from uniform sampling and provides further evidence
to the interpretation in terms of average gate fidelities.
Maybe most important for our work to serve as a basis for substantial
further development of methods and protocols
are new conceptual insights into
how inversion gates are -- in contrast to common belief -- not required for \ac{RB}
and into how large classes of groups in \ac{RB} can become available by means of
new filtering techniques. This contributes to overcoming the problem of
isolating exponential decays in a fully scalable manner.
First steps into exploiting the insights established here when devising
new schemes have already been made
\cite{MatchgateRB,LinghangKongCompactGroups,RandomSequences}.
We hope that this work provides a starting point of a further rich class of new protocols
of quantum certification and benchmarking, providing stringent and rigorous
quality criteria, while respecting experimental needs and desiderata.}

\section{Acknowledgments}

J.~H.~would like to acknowledge helpful conversations with Michael Walter, Bas Dirkse, and Freek Witteveen. I.~R.~would like to thank Richard Kueng, Martin Kliesch, Marios Ioannou, Dominik Hangleiter and Jonas Haferkamp for helpful discussions and Susane Calegari for contributions to the illustration.
{The authors would also like to acknowledge an anonymous referee for pointing out the correct way to include cycle benchmarking into the framework of theorem \ref{thm:mother}. }
\newtext{The Berlin team} has been supported by the BMBF project DAQC, for which it
introduces new methods for randomized benchmarking of near-term superconducting quantum platforms,
 and BMBF project MUNIQC-ATOMS, for which it introduces a starting point to develop schemes 
 of analog randomized benchmarking.
It has also been funded by the DFG (EI 519/9-1, for which this work develops ideas of
signal processing, and DFG CRC 183, for which this is an internode work
Berlin-Copenhagen, as well as DFG EI 519/14-1), and the Munich Quantum Valley (K-8).
This work has also received funding from the European Union's Horizon 2020 research and innovation
programme under grant agreement No. 817482 (PASQuanS), for which it
assesses feasible benchmarking schemes in quantum computing and simulation,
and the Einstein Foundation. E.~O.~has been supported by the Royal Society. A.~H.~W.~thanks the VILLUM
FONDEN for its support with a Villum Young Investigator
Grant (Grant No.~25452) and its support via the QMATH
Centre of Excellence (Grant No.~10059).

%%%=============================================
%%%=============================================
%\bibliography{frameworklib}

\begin{thebibliography}{96}
\providecommand{\natexlab}[1]{#1}
\providecommand{\url}[1]{\texttt{#1}}
\expandafter\ifx\csname urlstyle\endcsname\relax
  \providecommand{\doi}[1]{doi: #1}\else
  \providecommand{\doi}{doi: \begingroup \urlstyle{rm}\Url}\fi

\bibitem[Emerson et~al.(2005)Emerson, Alicki, and Zyczkowski]{FirstRB}
J.~Emerson, R.~Alicki, and K.~Zyczkowski.
\newblock Scalable noise estimation with random unitary operators.
\newblock \emph{J. Opt. B}, 7:\penalty0 S347--S352, 2005.

\bibitem[Dankert et~al.(2009)Dankert, Cleve, Emerson, and
  Livine]{dankert_exact_2009}
C.~Dankert, R.~Cleve, J.~Emerson, and E.~Livine.
\newblock Exact and approximate unitary 2-designs and their application to
  fidelity estimation.
\newblock \emph{Phys. Rev. A}, 80:\penalty0 012304, 2009.

\bibitem[L\'evi et~al.(2007)L\'evi, L\'opez, Emerson, and
  Cory]{PhysRevA.75.022314}
B.~L\'evi, C.~C. L\'opez, J.~Emerson, and D.~G. Cory.
\newblock Efficient error characterization in quantum information processing.
\newblock \emph{Phys. Rev. A}, 75:\penalty0 022314, 2007.

\bibitem[Magesan et~al.(2012)Magesan, Gambetta, Johnson, Ryan, Chow, Merkel,
  Da~Silva, Keefe, Rothwell, Ohki, et~al.]{magesan2012efficient}
E.~Magesan, J.~M. Gambetta, B.~R. Johnson, C.~A. Ryan, J.~M. Chow, S.~T.
  Merkel, M.~P. Da~Silva, G.~A. Keefe, M.~B. Rothwell, T.~A. Ohki, et~al.
\newblock Efficient measurement of quantum gate error by interleaved randomized
  benchmarking.
\newblock \emph{Phys. Rev. Lett.}, 109:\penalty0 080505, 2012.

\bibitem[Knill et~al.(2008)Knill, Leibfried, Reichle, Britton, Blakestad, Jost,
  Langer, Ozeri, Seidelin, and Wineland]{knill2008randomized}
E.~Knill, D.~Leibfried, R.~Reichle, J.~Britton, R.~B. Blakestad, J.~D. Jost,
  C.~Langer, R.~Ozeri, S.~Seidelin, and D.~J. Wineland.
\newblock Randomized benchmarking of quantum gates.
\newblock \emph{Phys. Rev. A}, 77:\penalty0 012307, 2008.

\bibitem[Emerson et~al.(2007)Emerson, Silva, Moussa, Ryan, Laforest, Baugh,
  Cory, and Laflamme]{emerson2007symmetrized}
J.~Emerson, M.~Silva, O.~Moussa, C.~Ryan, M.~Laforest, J.~Baugh, D.~G. Cory,
  and R.~Laflamme.
\newblock Symmetrized characterization of noisy quantum processes.
\newblock \emph{Science}, 317:\penalty0 1893--1896, 2007.

\bibitem[Campbell et~al.(2017)Campbell, Terhal, and Vuillot]{Roads}
E.~T. Campbell, B.~M. Terhal, and C.~Vuillot.
\newblock Roads towards fault-tolerant universal quantum computation.
\newblock \emph{Nature}, 549:\penalty0 172--179, 2017.

\bibitem[Barends et~al.(2014)Barends, Kelly, Veitia, Megrant, Fowler, Campbell,
  Chen, Chen, Chiaro, Dunsworth, Hoi, Jeffrey, Neill, O'Malley, Mutus,
  Quintana, Roushan, Sank, Wenner, White, Korotkov, Cleland, and
  Martinis]{PhysRevA.90.030303}
R.~Barends, J.~Kelly, A.~Veitia, A.~Megrant, A.~G. Fowler, B.~Campbell,
  Y.~Chen, Z.~Chen, B.~Chiaro, A.~Dunsworth, I.-C. Hoi, E.~Jeffrey, C.~Neill,
  P.~J.~J. O'Malley, J.~Mutus, C.~Quintana, P.~Roushan, D.~Sank, J.~Wenner,
  T.~C. White, A.~N. Korotkov, A.~N. Cleland, and J.~M. Martinis.
\newblock Rolling quantum dice with a superconducting qubit.
\newblock \emph{Phys. Rev. A}, 90:\penalty0 030303, 2014.

\bibitem[Onorati et~al.(2019)Onorati, Werner, and
  Eisert]{PhysRevLett.123.060501}
E.~Onorati, A.~H. Werner, and J.~Eisert.
\newblock Randomized benchmarking for individual quantum gates.
\newblock \emph{Phys. Rev. Lett.}, 123:\penalty0 060501, 2019.

\bibitem[Carignan-Dugas et~al.(2015{\natexlab{a}})Carignan-Dugas, Wallman, and
  Emerson]{PhysRevA.92.060302}
A.~Carignan-Dugas, J.~J. Wallman, and J.~Emerson.
\newblock Characterizing universal gate sets via dihedral benchmarking.
\newblock \emph{Phys. Rev. A}, 92:\penalty0 060302, 2015{\natexlab{a}}.

\bibitem[Cross et~al.(2016{\natexlab{a}})Cross, Magesan, Bishop, Smolin, and
  Gambetta]{BeyondCliffordRB}
A.~W. Cross, E.~Magesan, L.~S. Bishop, J.~A. Smolin, and J.~M. Gambetta.
\newblock {Scalable randomized benchmarking of non-Clifford gates}.
\newblock \emph{npj Quant. Inf.}, 2:\penalty0 16012, 2016{\natexlab{a}}.

\bibitem[Helsen et~al.(2019{\natexlab{a}})Helsen, Xue, Vandersypen, and
  Wehner]{HelsenEtAl:2019:character}
J.~Helsen, X.~Xue, L.~M.~K. Vandersypen, and S.~Wehner.
\newblock A new class of efficient randomized benchmarking protocols.
\newblock \emph{npj Quant. Inf.}, 5:\penalty0 1--9, 2019{\natexlab{a}}.

\bibitem[Erhard et~al.(2019)Erhard, Wallman, Postler, Meth, Stricker, Martinez,
  Schindler, Monz, Emerson, and Blatt]{CycleBenchmarking}
A.~Erhard, J.~J. Wallman, L.~Postler, M.~Meth, R.~Stricker, E.~A. Martinez,
  P.~Schindler, T.~Monz, J.~Emerson, and R.~Blatt.
\newblock Characterizing large-scale quantum computers via cycle benchmarking.
\newblock \emph{Nature Comm.}, 10, 2019.

\bibitem[Franca and Hashagen(2018)]{francca2018approximate}
D.~S. Franca and A.~K. Hashagen.
\newblock Approximate randomized benchmarking for finite groups.
\newblock \emph{J. Phys. A}, 51:\penalty0 395302, 2018.

\bibitem[Proctor et~al.(2019)Proctor, Carignan-Dugas, Rudinger, Nielsen,
  Blume-Kohout, and Young]{proctor2019direct}
T.~J. Proctor, A.~Carignan-Dugas, K.~Rudinger, E.~Nielsen, R.~Blume-Kohout, and
  K.~Young.
\newblock Direct randomized benchmarking for multiqubit devices.
\newblock \emph{Phys. Rev. Lett.}, 123:\penalty0 030503, 2019.

\bibitem[Wallman et~al.(2015{\natexlab{a}})Wallman, Granade, Harper, and
  Flammia]{wallman2015estimating}
J.~Wallman, C.~Granade, R.~Harper, and S.~T. Flammia.
\newblock Estimating the coherence of noise.
\newblock \emph{New J. Phys.}, 17:\penalty0 113020, 2015{\natexlab{a}}.

\bibitem[Gambetta et~al.(2012{\natexlab{a}})Gambetta, C\'orcoles, Merkel,
  Johnson, Smolin, Chow, Ryan, Rigetti, Poletto, Ohki, Ketchen, and
  Steffen]{GambettaEtAl:2012:simultaneousRB}
J.~M. Gambetta, A.~D. C\'orcoles, S.~T. Merkel, B.~R. Johnson, J.~A. Smolin,
  J.~M. Chow, C.~A. Ryan, C.~Rigetti, S.~Poletto, T.~A. Ohki, M.~B. Ketchen,
  and M.~Steffen.
\newblock Characterization of addressability by simultaneous randomized
  benchmarking.
\newblock \emph{Phys. Rev. Lett.}, 109:\penalty0 240504, 2012{\natexlab{a}}.

\bibitem[Wallman et~al.(2015{\natexlab{b}})Wallman, Barnhill, and
  Emerson]{WallmanEtAl:2015:LossRates}
J.~J. Wallman, M.~Barnhill, and J.~Emerson.
\newblock Robust characterization of loss rates.
\newblock \emph{Phys. Rev. Lett.}, 115:\penalty0 060501, 2015{\natexlab{b}}.

\bibitem[Wallman et~al.(2016{\natexlab{a}})Wallman, Barnhill, and
  Emerson]{WallmanEtAl:2016:Leakage}
J.~J. Wallman, M.~Barnhill, and J.~Emerson.
\newblock Robust characterization of leakage errors.
\newblock \emph{New J. Phys.}, 18:\penalty0 043021, 2016{\natexlab{a}}.

\bibitem[Kimmel et~al.(2014{\natexlab{a}})Kimmel, da~Silva, Ryan, Johnson, and
  Ohki]{KimmOhki}
S.~Kimmel, M.~P. da~Silva, C.~A. Ryan, B.~R. Johnson, and T.~Ohki.
\newblock Robust extraction of tomographic information via randomized
  benchmarking.
\newblock \emph{Phys. Rev. X}, 4:\penalty0 011050, 2014{\natexlab{a}}.

\bibitem[Roth et~al.(2018)Roth, Kueng, Kimmel, Liu, Gross, Eisert, and
  Kliesch]{AverageGateFidelities}
I.~Roth, R.~Kueng, S.~Kimmel, Y.-K. Liu, D.~Gross, J.~Eisert, and M.~Kliesch.
\newblock Recovering quantum gates from few average gate fidelities.
\newblock \emph{Phys. Rev. Lett.}, 121:\penalty0 170502, 2018.

\bibitem[{Flammia} and {Wallman}(2019)]{2019arXiv190712976F}
S.~T. {Flammia} and J.~J. {Wallman}.
\newblock {Efficient estimation of Pauli channels}.
\newblock \emph{arXiv:1907.12976}, 2019.

\bibitem[Eisert et~al.(2020)Eisert, Hangleiter, Walk, Roth, Markham, Parekh,
  Chabaud, and Kashefi]{Review}
J.~Eisert, D.~Hangleiter, N.~Walk, I.~Roth, D.~Markham, R.~Parekh, U.~Chabaud,
  and E.~Kashefi.
\newblock Quantum certification and benchmarking.
\newblock \emph{Nature Reviews Phys.}, 2:\penalty0 382--390, 2020.

\bibitem[Wallman(2018)]{wallman2018randomized}
J.~J. Wallman.
\newblock Randomized benchmarking with gate-dependent noise.
\newblock \emph{Quantum}, 2:\penalty0 47, 2018.

\bibitem[Merkel et~al.(2018)Merkel, Pritchett, and Fong]{Merkel18}
S.~T. Merkel, E.~J. Pritchett, and B.~H. Fong.
\newblock Randomized benchmarking as convolution: Fourier analysis of gate
  dependent errors.
\newblock 2018.

\bibitem[Proctor et~al.(2017)Proctor, Rudinger, Young, Sarovar, and
  Blume-Kohout]{Proctor17}
T.~Proctor, K.~Rudinger, K.~Young, M.~Sarovar, and R.~Blume-Kohout.
\newblock What randomized benchmarking actually measures.
\newblock \emph{Phys. Rev. Lett.}, 119:\penalty0 130502, 2017.

\bibitem[Carignan-Dugas et~al.(2018)Carignan-Dugas, Boone, Wallman, and
  Emerson]{carignan2018randomized}
A.~Carignan-Dugas, K.~Boone, J.~J. Wallman, and J.~Emerson.
\newblock From randomized benchmarking experiments to gate-set circuit
  fidelity: how to interpret randomized benchmarking decay parameters.
\newblock \emph{New J. Phys.}, 20:\penalty0 092001, 2018.

\bibitem[Acin et~al.(2018)Acin, Bloch, Buhrman, Calarco, Eichler, Eisert,
  Esteve, Gisin, Glaser, Jelezko, Kuhr, Lewenstein, Riedel, Schmidt, Thew,
  Wallraff, Walmsley, and Wilhelm]{Roadmap}
A.~Acin, I.~Bloch, H.~Buhrman, T.~Calarco, C.~Eichler, J.~Eisert, J.~Esteve,
  N.~Gisin, S.~J. Glaser, F.~Jelezko, S.~Kuhr, M.~Lewenstein, M.~F. Riedel,
  P.~O. Schmidt, R.~Thew, A.~Wallraff, I.~Walmsley, and F.~K. Wilhelm.
\newblock {The European quantum technologies roadmap}.
\newblock \emph{New J. Phys.}, 20:\penalty0 080201, 2018.

\bibitem[Arute et~al.(2019)Arute, Arya, Babbush, Bacon, Bardin, Barends,
  Biswas, Boixo, Brandao, Buell, et~al.]{arute2019quantum}
F.~Arute, K.~Arya, R.~Babbush, D.~Bacon, J.~C. Bardin, R.~Barends, R.~Biswas,
  S.~Boixo, F.~G. S.~L. Brandao, D.~A. Buell, et~al.
\newblock Quantum supremacy using a programmable superconducting processor.
\newblock \emph{Nature}, 574:\penalty0 505--510, 2019.

\bibitem[Bouland et~al.(2019)Bouland, Fefferman, Nirkhe, and
  Vazirani]{bouland2019complexity}
A.~Bouland, B.~Fefferman, C.~Nirkhe, and U.~Vazirani.
\newblock On the complexity and verification of quantum random circuit
  sampling.
\newblock \emph{Nature Physics}, 15:\penalty0 159--163, 2019.

\bibitem[Noh et~al.(2020)Noh, Jiang, and Fefferman]{noh2020efficient}
K.~Noh, L.~Jiang, and B.~Fefferman.
\newblock Efficient classical simulation of noisy random quantum circuits in
  one dimension.
\newblock \emph{Quantum}, 4:\penalty0 318, 2020.

\bibitem[Dalzell et~al.(2021)Dalzell, Hunter-Jones, and
  Brand{\~a}o]{dalzell2021random}
A.~M. Dalzell, N.~Hunter-Jones, and F.~G. S.~L. Brand{\~a}o.
\newblock Random quantum circuits transform local noise into global white
  noise.
\newblock 2021.
\newblock arXiv:a2111.14907.

\bibitem[Liu et~al.(2021)Liu, Otten, Bassirianjahromi, Jiang, and
  Fefferman]{liu2021benchmarking}
Y.~Liu, M.~Otten, R.~Bassirianjahromi, L.~Jiang, and B.~Fefferman.
\newblock Benchmarking near-term quantum computers via random circuit sampling.
\newblock 2021.
\newblock arXiv:2105.05232.

\bibitem[Magesan et~al.(2011{\natexlab{a}})Magesan, Gambetta, and
  Emerson]{CliffordRBPRL}
E.~Magesan, J.~M. Gambetta, and J.~Emerson.
\newblock Scalable and robust randomized benchmarking of quantum processes.
\newblock \emph{Phys. Rev. Lett.}, 106:\penalty0 180504, 2011{\natexlab{a}}.

\bibitem[Hashagen et~al.(2018)Hashagen, Flammia, Gross, and
  Wallman]{hashagen2018real}
A.~K. Hashagen, S.~T. Flammia, D.~Gross, and J.~J. Wallman.
\newblock Real randomized benchmarking.
\newblock \emph{Quantum}, 2:\penalty0 85, 2018.

\bibitem[Gambetta et~al.(2012{\natexlab{b}})Gambetta, C{\'o}rcoles, Merkel,
  Johnson, Smolin, Chow, Ryan, Rigetti, Poletto, Ohki,
  et~al.]{gambetta2012characterization}
J.~M. Gambetta, A.~D. C{\'o}rcoles, S.~T. Merkel, B.~R. Johnson, J.~A. Smolin,
  J.~M. Chow, C.~A. Ryan, C.~Rigetti, S.~Poletto, T.~A. Ohki, et~al.
\newblock Characterization of addressability by simultaneous randomized
  benchmarking.
\newblock \emph{Phys. Rev. Lett.}, 109:\penalty0 240504, 2012{\natexlab{b}}.

\bibitem[Carignan-Dugas et~al.(2015{\natexlab{b}})Carignan-Dugas, Wallman, and
  Emerson]{carignan2015characterizing}
A.~Carignan-Dugas, J.~J. Wallman, and J.~Emerson.
\newblock Characterizing universal gate sets via dihedral benchmarking.
\newblock \emph{Phys. Rev. A}, 92:\penalty0 060302, 2015{\natexlab{b}}.

\bibitem[Cross et~al.(2016{\natexlab{b}})Cross, Magesan, Bishop, Smolin, and
  Gambetta]{cross2016scalable}
A.~W. Cross, E.~Magesan, L.~S. Bishop, J.~A. Smolin, and J.~M. Gambetta.
\newblock {Scalable randomised benchmarking of non-Clifford gates}.
\newblock \emph{npj Quant. Inf.}, 2:\penalty0 16012, 2016{\natexlab{b}}.

\bibitem[Helsen et~al.(2019{\natexlab{b}})Helsen, Xue, Vandersypen, and
  Wehner]{helsen2019new}
J.~Helsen, X.~Xue, L.~M.~K. Vandersypen, and S.~Wehner.
\newblock A new class of efficient randomized benchmarking protocols.
\newblock \emph{npj Quant. Inf.}, 5:\penalty0 1--9, 2019{\natexlab{b}}.

\bibitem[Brown and Eastin(2018)]{brown2018randomized}
W.~G. Brown and B.~Eastin.
\newblock Randomized benchmarking with restricted gate sets.
\newblock \emph{Phys. Rev. A}, 97:\penalty0 062323, 2018.

\bibitem[Chasseur and Wilhelm(2015)]{chasseur2015complete}
T.~Chasseur and F.~K. Wilhelm.
\newblock Complete randomized benchmarking protocol accounting for leakage
  errors.
\newblock \emph{Phys. Rev. A}, 92:\penalty0 042333, 2015.

\bibitem[Wood and Gambetta(2018)]{wood2018quantification}
C.~J. Wood and J.~M. Gambetta.
\newblock Quantification and characterization of leakage errors.
\newblock \emph{Phys. Rev. A}, 97:\penalty0 032306, 2018.

\bibitem[Alexander et~al.(2016)Alexander, Turner, and
  Bartlett]{alexander2016randomized}
R.~N. Alexander, P.~S. Turner, and S.~D. Bartlett.
\newblock Randomized benchmarking in measurement-based quantum computing.
\newblock \emph{Phys. Rev. A}, 94:\penalty0 032303, 2016.

\bibitem[Combes et~al.(2017)Combes, Granade, Ferrie, and
  Flammia]{combes2017logical}
J.~Combes, C.~Granade, C.~Ferrie, and S.~T. Flammia.
\newblock Logical randomized benchmarking.
\newblock \emph{arXiv:1702.03688}, 2017.

\bibitem[Flammia and Wallman(2020)]{flammia2019efficient}
S.~T. Flammia and J.~J. Wallman.
\newblock Efficient estimation of pauli channels.
\newblock \emph{Nature Phys.}, 2020.
\newblock arXiv:1907.12976.

\bibitem[Harper et~al.(2020)Harper, Flammia, and Wallman]{harper2019efficient}
R.~Harper, S.~T Flammia, and J.~J. Wallman.
\newblock Efficient learning of quantum noise.
\newblock \emph{Nature Phys.}, 2020.

\bibitem[Harper and Flammia(2017)]{harper2017estimating}
R.~Harper and S.~T. Flammia.
\newblock {Estimating the fidelity of T gates using standard interleaved
  randomized benchmarking}.
\newblock \emph{Quant. Sc. Tech.}, 2:\penalty0 015008, 2017.

\bibitem[Sheldon et~al.(2016)Sheldon, Bishop, Magesan, Filipp, Chow, and
  Gambetta]{sheldon2016characterizing}
S.~Sheldon, L.~S. Bishop, E.~Magesan, S.~Filipp, J.~M. Chow, and J.~M.
  Gambetta.
\newblock Characterizing errors on qubit operations via iterative randomized
  benchmarking.
\newblock \emph{Phys. Rev. A}, 93:\penalty0 012301, 2016.

\bibitem[Chasseur et~al.(2017)Chasseur, Reich, Koch, and
  Wilhelm]{chasseur2017hybrid}
T.~Chasseur, D.~M. Reich, C.~P. Koch, and F.~K. Wilhelm.
\newblock Hybrid benchmarking of arbitrary quantum gates.
\newblock \emph{Phys. Rev. A}, 95:\penalty0 062335, 2017.

\bibitem[Kimmel et~al.(2014{\natexlab{b}})Kimmel, da~Silva, Ryan, Johnson, and
  Ohki]{kimmel2014robust}
S.~Kimmel, M.~P. da~Silva, C.~A. Ryan, B.~R. Johnson, and T.~Ohki.
\newblock Robust extraction of tomographic information via randomized
  benchmarking.
\newblock \emph{Phys. Rev. X}, 4:\penalty0 011050, 2014{\natexlab{b}}.

\bibitem[Boone et~al.(2019)Boone, Carignan-Dugas, Wallman, and
  Emerson]{boone2019randomized}
K.~Boone, A.~Carignan-Dugas, J.~J. Wallman, and J.~Emerson.
\newblock Randomized benchmarking under different gate sets.
\newblock \emph{Phys. Rev. A}, 99:\penalty0 032329, 2019.

\bibitem[Ryan et~al.(2009)Ryan, Laforest, and Laflamme]{ryan2009randomized}
C.~A. Ryan, M.~Laforest, and R.~Laflamme.
\newblock Randomized benchmarking of single-and multi-qubit control in
  liquid-state nmr quantum information processing.
\newblock \emph{New J. Phys.}, 11:\penalty0 013034, 2009.

\bibitem[Kato(1995)]{KatoPTLO}
T.~Kato.
\newblock \emph{Perturbation theory for linear operators}, volume 132.
\newblock Springer-Verlag Berlin Heidelberg, 1995.

\bibitem[Stewart and Sun(1990)]{StewartSun}
G.~W. Stewart and Ji-Guang Sun.
\newblock \emph{Matrix perturbation theory}.
\newblock Academic Press, 1990.

\bibitem[Huang et~al.(2020)Huang, Kueng, and Preskill]{huang2020predicting}
H.-Y. Huang, R.~Kueng, and J.~Preskill.
\newblock Predicting many properties of a quantum system from very few
  measurements.
\newblock \emph{Nat. Phys.}, 16:\penalty0 1050--1057, 2020.

\bibitem[{Kliesch} and {Roth}(2021)]{KlieschRoth:2020:Tutorial}
M.~{Kliesch} and I.~{Roth}.
\newblock Theory of quantum system certification.
\newblock \emph{PRX Quantum}, 2:\penalty0 010201, 2021.

\bibitem[Helsen et~al.(2019{\natexlab{c}})Helsen, Wallman, Flammia, and
  Wehner]{helsen2019multiqubit}
J.~Helsen, J.~J. Wallman, S.~T. Flammia, and S.~Wehner.
\newblock Multiqubit randomized benchmarking using few samples.
\newblock \emph{Phys. Rev. A}, 100:\penalty0 032304, 2019{\natexlab{c}}.

\bibitem[Helsen et~al.(2020)Helsen, Nezami, Reagor, and Walter]{MatchgateRB}
J.~Helsen, S.~Nezami, M.~Reagor, and M.~Walter.
\newblock Matchgate benchmarking: Scalable benchmarking of a continuous family
  of many-qubit gates.
\newblock 2020.
\newblock arXiv:2011.13048.

\bibitem[Kong()]{LinghangKongCompactGroups}
L.~Kong.
\newblock A framework for randomized benchmarking over compact groups.
\newblock arXiv:2111.10357.

\bibitem[Helsen et~al.(2021)Helsen, Ioannou, Roth, Kitzinger, Onorati, Werner,
  and Eisert]{RandomSequences}
J.~Helsen, M.~Ioannou, I.~Roth, J.~Kitzinger, E.~Onorati, A.~H. Werner, and
  J.~Eisert.
\newblock Estimating gate-set properties from random sequences.
\newblock 2021.
\newblock arXiv:2110.13178.

\bibitem[Goodman and Wallach(2000)]{goodman2000representations}
R.~Goodman and N.~R. Wallach.
\newblock \emph{Representations and invariants of the classical groups}.
\newblock Cambridge University Press, 2000.

\bibitem[Fulton and Harris(2013)]{fulton2013representation}
W.~Fulton and J.~Harris.
\newblock \emph{Representation theory: a first course}, volume 129.
\newblock Springer Science \& Business Media, 2013.

\bibitem[Gowers and Hatami(2017)]{GowersHatami}
W.~T. Gowers and O.~Hatami.
\newblock Inverse and stability theorems for approximate representations of
  finite groups.
\newblock \emph{Sbornik: Mathematics}, 208:\penalty0 1784--1817, 2017.

\bibitem[Kilmer and O'Leary(2010)]{Stewartselected}
M.~E. Kilmer and D.~P. O'Leary, editors.
\newblock \emph{G.W. Stewart; Selected works with commentaries}.
\newblock Birkhäuser Basel, 2010.

\bibitem[Dirkse et~al.(2019)Dirkse, Helsen, and Wehner]{dirkse2019efficient}
B.~Dirkse, J.~Helsen, and S.~Wehner.
\newblock {Efficient unitarity randomized benchmarking of few-qubit Clifford
  gates}.
\newblock \emph{Phys. Rev. A}, 99:\penalty0 012315, 2019.

\bibitem[Wallman et~al.(2016{\natexlab{b}})Wallman, Barnhill, and
  Emerson]{wallman2016robust}
J.~J. Wallman, M.~Barnhill, and J.~Emerson.
\newblock Robust characterization of leakage errors.
\newblock \emph{New J. Phys.}, 18:\penalty0 043021, 2016{\natexlab{b}}.

\bibitem[Magesan et~al.(2011{\natexlab{b}})Magesan, Blume-Kohout, and
  Emerson]{magesan2011gate}
E.~Magesan, R.~Blume-Kohout, and J.~Emerson.
\newblock Gate fidelity fluctuations and quantum process invariants.
\newblock \emph{Phys. Rev. A}, 84:\penalty0 012309, 2011{\natexlab{b}}.

\bibitem[Wallman and Flammia(2014)]{wallman2014randomized}
J.~J. Wallman and S.~T. Flammia.
\newblock Randomized benchmarking with confidence.
\newblock \emph{New J. Phys.}, 16:\penalty0 103032, 2014.

\bibitem[Epstein et~al.(2014)Epstein, Cross, Magesan, and
  Gambetta]{epstein2014investigating}
J.~M. Epstein, A.~W. Cross, E.~Magesan, and J.~M. Gambetta.
\newblock Investigating the limits of randomized benchmarking protocols.
\newblock \emph{Phys. Rev. A}, 89:\penalty0 062321, 2014.

\bibitem[Fong and Merkel(2017)]{fong2017randomized}
B.~H. Fong and S.~T. Merkel.
\newblock {Randomized benchmarking, correlated noise, and Ising models}.
\newblock \emph{arXiv:1703.09747}, 2017.

\bibitem[Fogarty et~al.(2015)Fogarty, Veldhorst, Harper, Yang, Bartlett,
  Flammia, and Dzurak]{fogarty2015nonexponential}
M.~A. Fogarty, M.~Veldhorst, R.~Harper, C.~H. Yang, S.~D. Bartlett, S.~T.
  Flammia, and A.~S. Dzurak.
\newblock Nonexponential fidelity decay in randomized benchmarking with
  low-frequency noise.
\newblock \emph{Phys. Rev. A}, 92:\penalty0 022326, 2015.

\bibitem[Carignan-Dugas et~al.(2019{\natexlab{a}})Carignan-Dugas, Wallman, and
  Emerson]{carignan2019bounding}
A.~Carignan-Dugas, J.~J. Wallman, and J.~Emerson.
\newblock Bounding the average gate fidelity of composite channels using the
  unitarity.
\newblock \emph{New J. Phys.}, 21:\penalty0 053016, 2019{\natexlab{a}}.

\bibitem[Kelley(1999)]{Kelley:1999:optimization}
C.~T. Kelley.
\newblock \emph{Iterative methods for optimization}.
\newblock SIAM, 1999.

\bibitem[Harper et~al.(2019)Harper, Hincks, Ferrie, Flammia, and
  Wallman]{Harper:2019:Stats}
R.~Harper, I.~Hincks, C.~Ferrie, S.~T. Flammia, and J.~J. Wallman.
\newblock Statistical analysis of randomized benchmarking.
\newblock \emph{Phys. Rev. A}, 99:\penalty0 052350, 2019.

\bibitem[Prony(1795)]{Prony:1795}
P.~R. Prony.
\newblock Essai experimentale et analytique.
\newblock \emph{J. de l’Ecole Polytechnique}, 1:\penalty0 24--76, 1795.

\bibitem[Cand{\`e}s and Fernandez-Granda(2013)]{CandesFernandez-Granda:2013}
E.~J. Cand{\`e}s and C.~Fernandez-Granda.
\newblock Super-resolution from noisy data.
\newblock \emph{J. Fourier An. App.}, 19:\penalty0 1229--1254, 2013.

\bibitem[Candes and Fernandez-Granda(2014)]{CandesFernandez-Granda:2014}
E.~J. Candes and C.~Fernandez-Granda.
\newblock Towards a {mathematical} {theory} of {super}-resolution.
\newblock \emph{Comm. Pure App. Math.}, 67:\penalty0 906--956, 2014.

\bibitem[Schmidt(1986)]{Schmidt:1986}
R.~Schmidt.
\newblock Multiple emitter location and signal parameter estimation.
\newblock \emph{IEEE Trans. Ant. Prop.}, 34:\penalty0 276--280, 1986.

\bibitem[Roy et~al.(1986)Roy, Paulraj, and Kailath]{RoyPaulrajKailath:1986}
R.~Roy, A.~Paulraj, and T.~Kailath.
\newblock Estimation of signal parameters via rotational invariance
  techniques-esprit.
\newblock In \emph{MILCOM 1986-IEEE Military Communications Conference:
  Communications-Computers: Teamed for the 90's}, volume~3, pages 41--6. IEEE,
  1986.

\bibitem[Liao and Fannjiang(2016)]{LiaoFannjiang:2014:MUSIC}
W.~Liao and A.~Fannjiang.
\newblock Music for single-snapshot spectral estimation: Stability and
  super-resolution.
\newblock \emph{Appl. Comp. Harm. An.}, 40:\penalty0 33--67, 2016.

\bibitem[Fannjiang(2016)]{Fannjiang:2016:ESPRIT}
A.~Fannjiang.
\newblock Compressive spectral estimation with single-snapshot esprit:
  Stability and resolution.
\newblock \emph{arXiv:1607.01827}, 2016.

\bibitem[Li and Liao(2017)]{LiLiao:2017}
W.~Li and W.~Liao.
\newblock Stable super-resolution limit and smallest singular value of
  restricted fourier matrices.
\newblock \emph{arXiv:1709.03146}, 2017.

\bibitem[Li et~al.(2019)Li, Liao, and
  Fannjiang]{LiLiaoFannjiang:2019:Super-resolution}
W.~Li, W.~Liao, and A.~Fannjiang.
\newblock {Super-resolution limit of the ESPRIT algorithm}.
\newblock \emph{arXiv:1905.03782}, 2019.

\bibitem[Badeau et~al.(2006)Badeau, David, and
  Richard]{BadeauDavidRichard:2006}
R.~Badeau, B.~David, and G.~Richard.
\newblock High-resolution spectral analysis of mixtures of complex exponentials
  modulated by polynomials.
\newblock \emph{IEEE Trans. Sig. Proc.}, 54:\penalty0 1341--1350, 2006.

\bibitem[Badeau et~al.(2008)Badeau, Richard, and
  David]{BadeauRichardDavid:2008}
R.~Badeau, G.~Richard, and B.~David.
\newblock Performance of esprit for estimating mixtures of complex exponentials
  modulated by polynomials.
\newblock \emph{IEEE Trans. Sig. Proc.}, 56:\penalty0 492--504, 2008.

\bibitem[Bazan(2006)]{Bazan}
F.~S.~V. Bazan.
\newblock {Conditioning of rectangular Vandermonde matrices with nodes in the
  unit disk}.
\newblock \emph{SIAM J. Mat. An. App.}, 21:\penalty0 679--693, 2006.

\bibitem[{Nguyen} et~al.(2019){Nguyen}, {Kim}, and
  {Shim}]{NguyenKimShim:2019:Low-rank}
L.~T. {Nguyen}, J.~{Kim}, and B.~{Shim}.
\newblock Low-rank matrix completion: A contemporary survey.
\newblock \emph{IEEE Access}, 7:\penalty0 94215--94237, 2019.

\bibitem[Tropp(2012)]{TroppBernstein}
J.~A. Tropp.
\newblock User-friendly tail bounds for sums of random matrices.
\newblock \emph{Found. Comput. Math.}, 12:\penalty0 389--434, 2012.

\bibitem[Ahlswede and Winter(2002)]{AW02}
R.~Ahlswede and A.~Winter.
\newblock Strong converse for identification via quantum channels.
\newblock \emph{IEEE Trans. Inform. Th.}, 48:\penalty0 569--579, 2002.

\bibitem[Ginory and Kim(2019)]{ginory2019weingarten}
A.~Ginory and J.~Kim.
\newblock {Weingarten calculus and the IntHaar package for integrals over
  compact matrix groups}.
\newblock \emph{Journal Symb. Comp.}, 2019.

\bibitem[Webb(2016)]{webb2016clifford}
Z.~Webb.
\newblock {The Clifford group forms a unitary 3-design}.
\newblock \emph{Quantum Information \& Computation}, 16:\penalty0 1379--1400,
  2016.

\bibitem[Zhu(2017)]{zhu2017multiqubit}
H.~Zhu.
\newblock {Multiqubit Clifford groups are unitary 3-designs}.
\newblock \emph{Phys. Rev. A}, 96:\penalty0 062336, 2017.

\bibitem[Ruskai et~al.(2002)Ruskai, Szarek, and Werner]{ruskai2002analysis}
M.~B. Ruskai, S.~Szarek, and E.~Werner.
\newblock An analysis of completely-positive trace-preserving maps on m2.
\newblock \emph{Lin. Alg. App.}, 347:\penalty0 159--187, 2002.

\bibitem[Carignan-Dugas et~al.(2019{\natexlab{b}})Carignan-Dugas, Alexander,
  and Emerson]{carignan2019polar}
A.~Carignan-Dugas, M.~Alexander, and J.~Emerson.
\newblock A polar decomposition for quantum channels (with applications to
  bounding error propagation in quantum circuits).
\newblock \emph{Quantum}, 3:\penalty0 173, 2019{\natexlab{b}}.

\bibitem[Hartfiel(1995)]{DenseDiagMatrices}
D.~J. Hartfiel.
\newblock Dense sets of diagonalizable matrices.
\newblock \emph{Proc. Am. Math. Soc.}, 123:\penalty0 1669--1672, 1995.

\bibitem[Wolf(2012)]{wolf2012quantum}
M.~M. Wolf.
\newblock Quantum channels \& operations: Guided tour.
\newblock \emph{Lecture notes available at http://www-m5. ma. tum.
  de/foswiki/pub M}, 5, 2012.

\end{thebibliography}

%%%=============================================
\end{document}